\documentclass[runningheads,orivec,colorlinks=true]{llncs}
\usepackage[T1]{fontenc}
\newcommand{\macrospath}{./macros}

\usepackage{ifthen}
 \usepackage{xspace}
 
\newboolean{talk}
\setboolean{talk}{false}
\newboolean{paper}
\setboolean{paper}{false}

\newboolean{IEEEstyle}
\setboolean{IEEEstyle}{false}
\newboolean{lipicsstyle}
\setboolean{lipicsstyle}{false}
\newboolean{eptcsstyle}
\setboolean{eptcsstyle}{false}
\newboolean{entcsstyle}
\setboolean{entcsstyle}{false}
\newboolean{sigplanstyle}
\setboolean{sigplanstyle}{false}
\newboolean{easychairstyle}
\setboolean{easychairstyle}{false}
\newboolean{scrartclstyle}
\setboolean{scrartclstyle}{false}
\newboolean{lncsstyle}
\setboolean{lncsstyle}{false}
\newboolean{acmartstyle}
\setboolean{acmartstyle}{false}

\newboolean{needstheorems}
\setboolean{needstheorems}{false}
\newboolean{withimages}
\setboolean{withimages}{false}
\newboolean{withproofs}
\setboolean{withproofs}{true}
\newboolean{complexity}
\setboolean{complexity}{false}

\newboolean{techr}
\setboolean{techr}{false}

\newboolean{french}
\setboolean{french}{false}

\setboolean{lncsstyle}{true}
\setboolean{withproofs}{true}
\ifthenelse{\boolean{talk}}{}{\usepackage[utf8]{inputenc}}
\ifthenelse{\boolean{french}}{\usepackage[french]{babel}}{
  \ifthenelse{\boolean{lipicsstyle}}{}{\usepackage[english]{babel}}}
\ifthenelse{\boolean{acmartstyle}}{}{
\usepackage{amssymb}	
}
\usepackage{amsmath}
\usepackage{amsfonts}
\usepackage{graphicx}
\usepackage{bussproofs}
\usepackage{cmll}
\usepackage{stmaryrd}
 \usepackage{multirow}
\usepackage[dvipsnames]{xcolor}
 \ifthenelse{\boolean{IEEEstyle}}{
 \usepackage[bookmarks={false}]{hyperref}}{
 \usepackage{hyperref}}
 \usepackage{fancybox}
 \usepackage{hhline}
\usepackage{nameref}
\usepackage[basic]{complexity}

\usepackage[normalem]{ulem}
\usepackage{proof}
\usepackage{mathtools}
\ifthenelse{\boolean{lipicsstyle}}{}{\usepackage{tabls}}
\usepackage{cleveref}

\newboolean{appendix}
\setboolean{appendix}{false}
\usepackage{appendix}
\capturecounter{theorem}
\capturecounter{proposition}
\capturecounter{lemma}
\usepackage{marginnote}

\newcommand{\NoteProof}[1]{
	\ifthenelse{\boolean{withproofs}}{\ifthenelse{\boolean{appendix}}{
	\marginnote{Originally at p. \pageref{#1}}
	}{
	\marginnote{{Proof at p.\,{\pageref{app:#1}}}}
	}
	}{}
}

\newcommand{\applabel}[1]{$\phantomsection\label{app:#1}$}

\ifthenelse{\boolean{IEEEstyle}}{
\setboolean{needstheorems}{true}}{}

\ifthenelse{\boolean{eptcsstyle}}{
\setboolean{needstheorems}{true}}{}

\ifthenelse{\boolean{scrartclstyle}}{
\setboolean{needstheorems}{true}}{}

\ifthenelse{\boolean{sigplanstyle}}{
\setboolean{needstheorems}{true}}{}

\ifthenelse{\boolean{easychairstyle}}{
\setboolean{needstheorems}{true}}{}

\ifthenelse{\boolean{lipicsstyle}}{
    }{}

\ifthenelse{\boolean{needstheorems}}{
\input{\macrospath/ben-macros-thm-environments}}{}

\newcommand{\myproof}[1]{
\ifthenelse{\boolean{withproofs}}{#1}{}
}

\newcommand{\la}[1]{\lambda #1.}

\newcommand{\tm}{t}
\newcommand{\tmtwo}{s}
\newcommand{\tmthree}{u}
\newcommand{\tmfour}{r}
\newcommand{\tmfive}{p}
\newcommand{\tmsix}{q}

\newcommand{\var}{x}
\newcommand{\vartwo}{y}
\newcommand{\varthree}{z}
\newcommand{\varfour}{w}

\newcommand{\dom}[1]{\symfont{dom}(#1)}

\newcommand{\rootRew}[1]{\mapsto_{#1}}
\newcommand{\Rew}[1]{\rightarrow_{#1}}

\renewcommand{\to}{\Rew{}}
\newcommand{\tob}{\Rew{\beta}}
\newcommand{\rtob}{\rootRew{\beta}}

\newcommand{\sn}[1]{\mathrm{SN}_{#1}}

\newcommand{\symfont}[1]{\mathtt{#1}}
\newcommand{\ssym}{\symfont{s}}
\newcommand{\vssym}{\symfont{vs}}

\newcommand{\db}{\symfont{dB}}

\newcommand{\cbv}{{CbV}\xspace}

\newcommand{\val}{v}
\newcommand{\valtwo}{\val'}

\newcommand{\ctxholep}[1]{\langle #1\rangle}
\newcommand{\ctxhole}{\ctxholep{\cdot}}

\makeatletter
\newsavebox{\@brx}
\newcommand{\llangle}[1][]{\savebox{\@brx}{\(\m@th{#1\langle}\)}%
  \mathopen{\copy\@brx\kern-0.7\wd\@brx\usebox{\@brx}}}
\newcommand{\rrangle}[1][]{\savebox{\@brx}{\(\m@th{#1\rangle}\)}%
  \mathclose{\copy\@brx\kern-0.7\wd\@brx\usebox{\@brx}}}
\makeatother

\newcommand{\ctxholefp}[1]{\llangle #1\rrangle}
\newcommand{\ctxholef}{\ctxholefp{\cdot}}

\newcommand{\ctx}{C}
\newcommand{\ctxtwo}{\ctx'}

\newcommand{\ctxp}[1]{\ctx\ctxholep{#1}}

\newcommand{\evctx}{E}

\newcommand{\evctxp}[1]{\evctx\ctxholep{#1}}

\newcommand{\nbvctxtwo}[1]{\nbvctxtwo{#1}}

\newcommand{\sctx}{L}

\newcommand{\sctxp}[1]{\sctx\ctxholep{#1}}

\newcommand{\defeq}{:=}

\newcommand{\grameq}{::=}

\newcommand{\esub}[2]{[#1{\shortleftarrow}#2]}
\newcommand{\isub}[2]{\{#1{\shortleftarrow}#2\}}
\newcommand{\subi}[2]{\{#1{\shortrightarrow}#2\}}

\ifthenelse{\boolean{complexity}}{ }{
  
  }

\newcommand{\letexp}[3]{\letexpsym\ #1=#2\ \symfont{in}\ #3}

\newcommand{\rtodb}{\rootRew{\db}}
\newcommand{\rtos}{\rootRew{\ssym}}
\newcommand{\tos}{\Rew{\ssym}}
\newcommand{\todbp}[1]{\Rew{\db#1}}
\newcommand{\todb}{\todbp{}}

\newcommand{\llbrace}{\{ \kern -0.27em \vert}
\newcommand{\rrbrace}{\vert \kern -0.27em \}}

\renewcommand{\l}{\lambda}
\newcommand{\ie}{i.e.\xspace}
\newcommand{\eg}{e.g.\xspace}
\newcommand{\ih}{{\textit{i.h.}}\xspace}
\newcommand{\fv}[1]{\symfont{fv}(#1)}

\ifthenelse{\boolean{entcsstyle}}{}{
	}

\newcommand{\custom}[1]{{\color{Emerald} {#1}}}
\newcommand{\red}[1]{{\color{red} {#1}}}

\definecolor{dgreen}{rgb}{0.0, 0.5, 0.0}

\newcommand{\ignore}[1]{}

\newcommand{\colspace}{@{\hspace{.5cm}}}
\newcommand{\hcolspace}{@{\hspace{.25cm}}}
\newcommand{\myinput}[1]{\ifthenelse{\boolean{withimages}}{\input{#1}}{}}
\newcommand{\inputproof}[1]{\ifthenelse{\boolean{withproofs}}{\input{#1}}{}}

\newcommand{\reflemma}[1]{Lemma~\ref{l:#1}}
\newcommand{\reflemmap}[2]{Lemma~\ref{l:#1}.\ref{p:#1-#2}}

\newcommand{\reflemmaeq}[1]{{L.\ref{l:#1}}}

\newcommand{\reftm}[1]{Theorem~\ref{tm:#1}}

\newcommand{\refprop}[1]{Proposition~\ref{prop:#1}}
\newcommand{\refsect}[1]{Sect.~\ref{sect:#1}}

\ifthenelse{\boolean{easychairstyle}}{}{
  \ifthenelse{\boolean{lncsstyle}}{
  	\newcommand{\refequa}[1]{(\ref{eq:#1})}
	}{	
	
  }
}
\newcommand{\reffig}[1]{Fig.~\ref{fig:#1}}

\newcommand{\refcorop}[2]{Corollary~\ref{coro:#1}.\ref{p:#1-#2}}

\newcommand{\refrem}[1]{Remark~\ref{rem:#1}}

\newcommand{\form}{A}
\newcommand{\formtwo}{B}
\newcommand{\formthree}{C}
\newcommand{\formfour}{D}

\newcommand{\aform}{X}

\newcommand{\multiForm}{\Gamma}
\newcommand{\multiFormtwo}{\Delta}
\newcommand{\multiFormthree}{\Pi}
\newcommand{\multiform}{\multiForm}
\newcommand{\multiformtwo}{\multiFormtwo}
\newcommand{\multiformthree}{\multiFormthree}

\newcommand{\exch}{\symfont{exch}}

\newcommand{\ax}{\symfont{ax}}
\newcommand{\cut}{\symfont{cut}}

\newcommand{\lolli}{\multimap}

\newcommand{\rightRuleSym}{r}
\newcommand{\leftRuleSym}{l}

\newcommand{\implyRightRule}{\imply_\rightRuleSym}
\newcommand{\implyLeftRule}{\imply_\leftRuleSym}

\newcommand{\set}[1]{\{#1\}}

\newcommand{\nat}{\mathbb{N}}

\newcommand{\settone}{\mathcal{S}}

\newcommand{\tosc}{\Rew{\symfont{SC}}}

\newcommand{\betav}{\beta_\val}

\newcommand{\tobv}{\Rew{\betav}}

\newcommand{\towbv}{\Rew{w\betav}}

\newcommand{\rtobv}{\rootRew{\betav}}

\newcommand{\cbn}{CbN\xspace}

\newcommand{\snsubst}{\sn{\to}}

\newcommand{\redc}[1]{\llbracket#1\rrbracket}

\newcommand{\size}[1]{|#1|}

\ifthenelse{\boolean{acmartstyle}}{
	
	}{
	
}

\newcommand{\withproofs}[1]{\ifthenelse{\boolean{withproofs}}{#1}{}}

\newcommand{\withoutproofs}[1]{\ifthenelse{\boolean{withproofs}}{}{#1}}

\ifthenelse{\boolean{entcsstyle}}{}{
  \ifthenelse{\boolean{lncsstyle}}{}{	
    
  }
 }

\newcounter{numberone}

\newcommand{\med}[2]{
($(#1)!.5!(#2)$)
}

\newcommand{\gregoire}{Gr{\'{e}}goire\xspace}

\ifthenelse{\boolean{acmartstyle}}{}{
	
}

\newcommand{\typctx}{\Gamma}
\newcommand{\typctxtwo}{\Delta}

\newcommand{\tderiv}{\pi}
\newcommand{\tderivtwo}{\rho}

\newcommand{\hastype}{\,{:}\,}

\makeatletter
\newcommand{\exder}{%
  \def\exderW[##1]{\triangleright_{##1}\ }%
  \def\exderWO{\triangleright\ }%
  \@ifnextchar[\exderW\exderWO%
  }
\makeatother

\renewcommand{\exder}{\,\triangleright}

\newcommand\Deribbase[5]{{#3}\ {\pmb\vdash}_{#2}^{#1} {#4}\  {:}\  {#5}}

\makeatletter

\newcommand{\Deribase}[1]{%
  \def\DeribW[##1]{\Deribbase{##1}{#1}}%
  \def\DeribWO{\Deribbase{}{#1}}%
  \@ifnextchar[\DeribW\DeribWO%
  }

  \newcommand{\Deri}{%
  \def\DeriW_##1{\Deribase{##1}}%
  \def\DeriWO{\Deribase{}}%
  \@ifnextchar_\DeriW\DeriWO%
  }
\makeatother

\newcommand{\app}{\symfont{app}}

\makeatletter
\newcommand{\appresult}{%
  \def\appresultW<##1>{\app_\result^{##1}}%
  \def\appresultWO{\app_\result}%
  \@ifnextchar<\appresultW\appresultWO%
  }
\makeatother

\newcommand\mydots{\hbox to .6em{.\hss.}}

\newcommand{\pof}{\;\triangleright}

\newcommand{\rtovs}{\rootRew{\vssym}}
\newcommand{\tovs}{\Rew{\vssym}}

\newcommand{\tovsc}{\Rew{\symfont{VSC}}}

\newcommand{\llterms}{\symfont{Terms}}
\newcommand{\lltermsp}[1]{\llterms_{#1}}

\newcommand{\elctxs}{\symfont{ElCtxs}}
\newcommand{\lltermsform}{\lltermsp{\vdash\form}}

\newcommand{\elctxsform}{\elctxs_{\form}}
\newcommand{\varssym}{\mathcal{V}}
\newcommand{\ndvars}{\varssym_{\ndsym}}
\newcommand{\scvars}{\varssym_{\scsym}}
\newcommand{\elvars}{\symfont{CtxVars}}

\newcommand{\elctx}{\scfont{E}}

\newcommand{\elctxfp}[1]{\elctx\ctxholefp{#1}}
\newcommand{\elctxtwo}{\elctx'}

\newcommand{\elctxtwofp}[1]{\elctxtwo\ctxholefp{#1}}

\newcommand{\gctx}{C}

\newcommand{\rcuteq}{\sim}
\newcommand{\cuteq}{\equiv}

\renewcommand{\letexp}{\mathsf{let}}
\newcommand{\letin}[3]{{\sf let}\ #1=#2\ {\sf in}\ #3}

\newcommand{\cuta}[2]{\redbbrackets{#1{\red\shortrightarrow}#2}}
\newcommand{\cutsub}[2]{\lBrace#1{\shortrightarrow}#2\rBrace}

\newcommand{\sube}[2]{\redbrackets{#1{\red\shortrightarrow}#2}}

\newcommand{\redbrackets}[1]{\red[#1\red]}

\newcommand{\redbbrackets}[1]{\red\llbracket#1\red\rrbracket}

\newcommand{\customexbrackets}[1]{\custom\lbrbrak#1\custom\rbrbrak}

\newcommand{\suba}[3]{\cyanbrackets{#1{\substr}#2,#3}}
\newcommand{\nsuba}[3]{\customexbrackets{#1\,{\custom\rhd}\,#2\custom |#3}}

\newcommand{\lctx}{L}
\newcommand{\lctxtwo}{\lctx'}

\newcommand{\lctxp}[1]{\lctx\ctxholep{#1}}

\newcommand{\lctxtwop}[1]{\lctxtwo\ctxholep{#1}}

\newcommand{\cameratech}[2]{\ifthenelse{\boolean{techr}}{#2}{#1}}

\newcommand{\trndtosc}[1]{\overline{#1}^{\scsym}}

\newcommand{\trsctond}[1]{\underline{#1}_{\ndsym}}

\newcommand{\scfont}[1]{\mathtt{#1}}

\newcommand{\svar}{\scfont\var}
\newcommand{\svartwo}{\scfont\vartwo}
\newcommand{\svarthree}{\scfont\varthree}
\newcommand{\svarfour}{\scfont\varfour}
\newcommand{\stm}{\scfont\tm}
\newcommand{\stmtwo}{\scfont\tmtwo}
\newcommand{\stmthree}{\scfont\tmthree}
\newcommand{\stmfour}{\scfont\tmfour}
\newcommand{\stmfive}{\scfont\tmfive}
\newcommand{\stmsix}{\scfont\tmsix}

\newcommand{\gsctx}{\scfont{\gctx}}

\newcommand{\gsctxp}[1]{\gsctx\ctxholep{#1}}

\newcommand{\lsctx}{\scfont{\lctx}}
\newcommand{\lsctxtwo}{\lsctx'}
\newcommand{\lsctxthree}{\lsctx''}
\newcommand{\lsctxfour}{\lsctx'''}
\newcommand{\lsctxp}[1]{\lsctx\ctxholep{#1}}

\newcommand{\lsctxtwop}[1]{\lsctxtwo\ctxholep{#1}}
\newcommand{\lsctxthreep}[1]{\lsctxthree\ctxholep{#1}}
\newcommand{\lsctxfourp}[1]{\lsctxfour\ctxholep{#1}}

\newcommand{\sval}{\scfont\val}
\newcommand{\svaltwo}{\scfont\valtwo}

\newcommand{\srootRew}[1]{\rightarrowtail_{#1}}
\newcommand{\sRew}[1]{\leadsto_{#1}}

\newcommand{\srtocut}{\srootRew{\cutsym}}

\definecolor{LightGray}{gray}{.80}
\definecolor{DarkGray}{gray}{.60}

\newcommand{\cutsym}{\symfont{cut}}
\newcommand{\stocut}{\sRew{\cutsym}}

\newcommand{\storencut}{\sRew{\symfont{ren}\mbox{-}\cutsym}}

\newcommand{\ndsym}{\symfont{N}}
\newcommand{\scsym}{\symfont{S}}

\newcommand{\nproves}{\vdash_{\ndsym}}

\newcommand{\sproves}{\vdash_{\scsym}}

\newcommand{\imply}{\Rightarrow}

\newcommand{\asym}{a}

\newcommand{\ndterms}{\Lambda_{\symfont{N}}}
\newcommand{\ndesterms}{\Lambda_{\symfont{N}}^{\symfont{ES}}}

\newcommand{\scterms}{\Lambda_{\symfont{S}}}

\newcommand{\tctxplus}{\cup}

\newcommand{\vars}{\scvars}

\newcommand{\wsctx}{\scfont{W}}
\newcommand{\wsctxtwo}{\wsctx'}

\newcommand{\wsctxp}[1]{\wsctx\ctxholep{#1}}
\newcommand{\wsctxtwop}[1]{\wsctxtwo\ctxholep{#1}}

\newcommand{\crcuteq}{\approx}

\usepackage{tikz}
\usetikzlibrary{calc}
\usetikzlibrary{matrix,arrows}
\usetikzlibrary{decorations.shapes}
\usetikzlibrary{decorations.text}
\usetikzlibrary{decorations.pathmorphing}
\usetikzlibrary{decorations.markings}
\usetikzlibrary{positioning}

\tikzset{
node distance=1.3cm, auto,
every node/.style={font=\scriptsize },
ocenter/.style={baseline={([yshift=-.5ex, xshift=-.5ex]current bounding box)}},  
labelBeginAbove/.style={postaction={decorate,decoration={markings,mark=at position 0 with {\node[inner sep= 0.6pt, above=1pt]{\tiny #1};}} } },
labelBeginBelow/.style={postaction={decorate,decoration={markings,mark=at position 0 with {\node[inner sep= 0.6pt, below=1pt]{\tiny #1};}}}},
labelEndAbove/.style={postaction={decorate,decoration={markings,mark=at position 1 with {\node[inner sep= 0.6pt, above=1pt]{\tiny #1};}}}},
labelEndBelow/.style={postaction={decorate,decoration={markings,mark=at position 1 with {\node[inner sep= 0.6pt, below=1pt]{\tiny #1};}}}},
labelEndRight/.style={postaction={decorate,decoration={markings,mark=at position 1 with {\node[inner sep= 0.6pt, right=1pt]{\tiny #1};}}}},
labelEndLeft/.style={postaction={decorate,decoration={markings,mark=at position 1 with {\node[inner sep= 0.6pt, left=1pt]{\tiny #1};}}}}
}

\usepackage{stix}
\usepackage{tipa}

\renewcommand{\suba}[3]{\nsuba{#1}{#2}{#3}}

\numberwithin{lemma}{section}
\numberwithin{proposition}{section}

\begin{document}
\title{The Vanilla Sequent Calculus is Call-by-Value}
\subtitle{(Fresh Perspective)}
%
%
\author{Beniamino Accattoli\orcidID{0000-0003-4944-9944}}
\authorrunning{B. Accattoli}
\institute{Inria \& LIX, Ecole Polytechnique, UMR 7161, Palaiseau, France
\email{\href{mailto:beniamino.accattoli@inria.fr}{beniamino.accattoli@inria.fr}}
}
\maketitle              
\begin{abstract}
Existing Curry-Howard interpretations of call-by-value evaluation for the $\lambda$-calculus are either based on ad-hoc modifications of intuitionistic proof systems or involve additional logical concepts such as classical logic or linear logic, despite the fact that call-by-value was introduced in an intuitionistic setting without linear features.

This paper shows that the most basic sequent calculus for minimal intuitionistic logic---dubbed here \emph{vanilla}---can naturally be seen as a logical interpretation of call-by-value evaluation. This is obtained by establishing mutual simulations with a well-known formalism for call-by-value evaluation.
%
\end{abstract}
%
%
%
\section{Introduction}
\label{sect:intro}
The connection between functional languages and proof theory stems from Howard's insight that the system of simple types for the $\l$-calculus is exactly Gentzen's natural deduction for minimal intuitionistic logic \cite{Howard1980-HOWTFN-2}. Additionally, \emph{$\beta$-reduction} exactly matches the logical process of \emph{detour elimination}, also called \emph{normalization}. This correspondence concerns the unrestricted, or \emph{call-by-name}, notion of $\beta$-reduction. 

In practice, the evaluation mechanism at work in functional languages never follows that of call-by-name of the ordinary $\l$-calculus. Plotkin's call-by-value $\l$-calculus \cite{DBLP:journals/tcs/Plotkin75} restricts $\beta$-reduction to fire when the argument is a \emph{value}, that is, a variable or an abstraction, and it is often seen as a better fit for applications. 

\paragraph{Call-by-Value and Natural Deduction.} Natural deduction, unfortunately, does not provide a solid logical foundation for call-by-value (shortened to \cbv). On the positive side, values are proofs ending in a logical introduction rule. The restriction to values in the proof normalization process, however, has to be \emph{enforced}; it does not arise naturally from the structure of natural deduction. Moreover, when one considers open terms and/or strong evaluation (that is, under abstraction), the normal forms of Plotkin's \cbv calculus  have a complex inductive structure not corresponding to any natural concept in natural deduction. 
 The study of strong evaluation is essential for proof assistants based on dependent types. Notably, Leroy and \gregoire use strong \cbv for Coq \cite{DBLP:conf/icfp/GregoireL02}. Strong evaluation is also advocated by Scherer and Rémy \cite{DBLP:conf/esop/SchererR15} as relevant for a solid theory of functional languages. 
Lastly, in proof theory one speaks of \emph{cut/detour elimination}, and they can be eliminated only if evaluation goes under abstractions.

The issue with normal forms is delicate. It is well-known that the rewriting rules of Plotkin's calculus are not suited for open terms and strong evaluation, causing \emph{premature normal forms}, as pointed out by Paolini and Ronchi della Rocca \cite{DBLP:journals/ita/PaoliniR99,DBLP:conf/ictcs/Paolini01,DBLP:series/txtcs/RoccaP04}. The literature contains many alternative \cbv $\l$-calculi fixing this defect, as surveyed by Accattoli and Guerrieri \cite{DBLP:conf/aplas/AccattoliG16}, often extending the syntax of the $\l$-calculus with $\letexp$-expressions or explicit substitutions, as for instance is the case in Moggi's calculus \cite{Moggi88tech,DBLP:conf/lics/Moggi89}. 

Logically, both $\letexp$-expressions and explicit substitutions are decorations for the (intuitionistic) cut rule. Even when one considers these alternative \cbv $\l$-calculi not suffering of premature normal forms, natural deduction with cut is not really a solid logical foundation, because \cbv \emph{normal} terms correspond to proofs \emph{with cuts}. That is, \emph{not all cuts are eliminable}.

\paragraph{Call-by-Value and Proof Theory.} Beyond natural deduction for minimal intuitionistic logic (shortened to MIL), there are two main ways of looking at \cbv via proof theory. One of them is as a certain intuitionistic fragment of linear logic, as first done by Girard \cite{DBLP:journals/tcs/Girard87}, which can also be seen as a certain embedding into the modal logic S4, as pointed out recently by Esp{\'{\i}}rito Santo et al. \cite{DBLP:journals/jlap/SantoPU22}. Another way is via the duality between \cbv and call-by-name (\cbn) in classical logic, usually traced back to Filinski \cite{DBLP:conf/ctcs/Filinski89}. This duality was then analyzed through linear logic lenses by Danos et al. \cite{danos93wll}, and crystalized as the computational interpretation of a sequent calculus by Curien and Herbelin \cite{DBLP:conf/icfp/CurienH00}. 

What is not ideal, however, is that these are \emph{additional logical concepts}. In Plotkin's \cbv (and in many of the alternative \cbv $\l$-calculi), indeed,  there is no trace of linearity nor of classical principles. It is natural to wonder whether there is a neat logical way of modeling \cbv using \emph{only} MIL. There are works connecting \cbv and MIL, such as Dyckhoff and Lengrand's study of Herbelin's LJQ \cite{HerbelinPhD,DBLP:conf/cie/DyckhoffL06,DBLP:journals/logcom/DyckhoffL07} or Ohori \cite{DBLP:conf/tlca/Ohori99}, but they are based on modifications of the deductive system \emph{imposing} the \cbv reading.  This paper presents a fresh perspective on this question.
\begin{figure}[t!]
\centering
\begin{tabular}{c\colspace\colspace c}	
	\AxiomC{$  \multiForm \vdash \form$}
	\AxiomC{$  \multiForm, \form \vdash \formtwo$}
	\RightLabel{$\cut $}
	\BinaryInfC{$  \multiForm \vdash  \formtwo$}
	\DisplayProof
&
	\AxiomC{}
	\RightLabel{$\ax$}
	\UnaryInfC{$  \multiForm, \form \vdash  \form$}	
	\DisplayProof
\\\\
	\AxiomC{$  \multiForm \vdash \form$}
	\AxiomC{$  \multiForm, \formtwo \vdash \formthree$}
	\RightLabel{$ \implyLeftRule $}
	\BinaryInfC{$  \multiForm,  \form \imply \formtwo \vdash \formthree$}
	\DisplayProof
&
		\AxiomC{$  \multiForm, \form \vdash \formtwo$}
	\RightLabel{$ \implyRightRule $}
	\UnaryInfC{$  \multiForm \vdash \form \imply \formtwo$}
	\DisplayProof
\end{tabular}
\caption{Vanilla sequent calculus for minimal intuitionistic logic (MIL).}
\label{fig:SC-intro}
\end{figure}

\paragraph{The Vanilla Sequent Calculus.} Our result is that the most basic presentation of Gentzen's sequent calculus \cite{Gentzen1964} for MIL, here dubbed \emph{vanilla sequent calculus} and shown in \reffig{SC-intro} (the proper additive presentation of MIL requires some further details about contractions, see \refsect{vanilla}), is a natural Curry-Howard reading of \cbv evaluation. The system comes with a natural notion of \emph{values}, which are proofs ending with right deduction rules. Crucially, we exploit a \emph{splitting property} that is specific to sequent proofs of MIL: every proof of $\typctx\vdash \form$ can be uniquely split into a value sub-proof of $\typctxtwo\vdash \form$ followed by a sequence of \emph{left} rules turning $\typctxtwo$ into $\typctx$ without touching $\form$. Thanks to this feature, there is no need to modify the sequent calculus with a so-called \emph{stoup} (\ie a distinguished formula in a formula context) as in Curien and Herbelin \cite{DBLP:conf/icfp/CurienH00}, adding restrictions on the shape of proofs as in Herbelin's presentation of LJQ \cite{HerbelinPhD} or Ohori \cite{DBLP:conf/tlca/Ohori99}, or having a separate judgement for values as in Lengrand and Dyckhoff's presentation of LJQ \cite{DBLP:conf/cie/DyckhoffL06,DBLP:journals/logcom/DyckhoffL07}. Essentially, values are recovered on-the-fly, with no need to mark them out explicitly.

Beyond the splitting property, our fresh perspective stems from the observation that the ineliminable cuts  in the \cbv reading of natural deduction correspond exactly to occurrences of the left rule $(\implyLeftRule)$ for $\imply$ in the vanilla sequent calculus. In other words, the cut rule in \cbv natural deduction is \emph{overloaded}. When one looks at it through the vanilla sequent calculus, it represents \emph{both} the vanilla $(\cut)$ and $(\implyLeftRule)$ rules. Disentangling the two concepts, the vanilla sequent calculus provides a neater logical foundation where \emph{\cbv normal terms have no cuts}, as one would expect.

Concretely, we adopt proof terms for the vanilla sequent calculus defined without the application $\tm\tmtwo$ construct (which is specific to natural deduction) and with \emph{distinct} explicit substitutions constructs for rules $(\cut)$ and $(\implyLeftRule)$, what we dub here \emph{vanilla $\l$-terms}.

\paragraph{Cut Elimination at a Distance: the Vanilla $\l$-Calculus.} Once it is clarified that cut-free vanilla proofs are a good formalism for \cbv normal forms, one needs to pair them with a natural notion of cut elimination. Quoting Zucker \cite{ZUCKER19741}: "we know at the outset what a cut-free derivation is, and the problem is to define the conversions". 

For that, we give a very compact definition of cut elimination following the \emph{at a distance} style of rewriting rules promoted by Accattoli and Kesner  for natural deduction calculi \cite{DBLP:conf/csl/AccattoliK10,DBLP:conf/popl/AccattoliBKL14,DBLP:journals/pacmpl/Kesner22}, and recently shown to be a good fit for the intuitionistic linear logic sequent calculus (without stoup) by Accattoli \cite{DBLP:journals/lmcs/Accattoli23}. To avoid misunderstandings, note that, while rewriting at a distance is also an additional concept, it is fundamentally different from linear or classical logic. Linear and classical logic are \emph{logical} concepts; they change the logic. Rewriting at a distance is a \emph{rewriting} concept, not a logical one; the logic stays the same, what changes is only how proofs/terms are rewritten.

The key point of rewriting at a distance is that it avoids the many rewriting rules propagating explicit substitutions through the term structure, which are found in many extended $\l$-calculi, via generalizations of the rewriting rules exploiting contexts (that is, terms with a hole). This is particularly useful in the study of sequent calculi, where commuting cut elimination cases are a notorious burden. Essentially, rewriting at a distance allows one to have only \emph{principal} cut elimination cases, avoiding completely the commutative ones. This is similar to what happens with proof nets, except that it is simpler, because the graphical language is avoided altogether by means of---again---contexts.

Rules at a distance rest on \emph{on-the-fly} decompositions of a term into a context and a sub-term. In our setting, this decomposition becomes the splitting property mentioned above. Distance thus allows for a smooth treatment of values at the rewriting level.

Our proof terms plus cut elimination at a distance form what we dub as the \emph{vanilla $\l$-calculus}.

\paragraph{Results.} The central result of the paper is that two standard translations, from natural deduction to sequent calculus, and back, induce termination-preserving simulations between the vanilla $\l$-calculus and one of the formalisms extending Plotkin's \cbv. To make things as smooth as possible, our choice of formalism for \cbv is Accattoli and Paolini's \emph{value substitution calculus} \cite{DBLP:conf/flops/AccattoliP12}, a natural deduction $\l$-calculus with explicit substitutions at a distance, used in various recent studies about \cbv \cite{DBLP:conf/aplas/AccattoliG16,DBLP:conf/lics/AccattoliCC21,DBLP:journals/pacmpl/AccattoliG22,DBLP:conf/ictac/AccattoliGL23,DBLP:conf/fossacs/AccattoliL24}.

For the sake of keeping this \emph{fresh perspective} light and readable, we focus on explaining the inception of the calculus, rewriting rules at a distance, and the simulation results. As a check of good design for the vanilla $\l$-calculus, we also prove strong normalization for typed terms. This result is more challenging than the simulations, but it is expected and proved via an established technique, namely, the bi-orthogonal reducibility method.  Therefore, we only give the statements and refer to the \withproofs{Appendix}\withoutproofs{technical report \cite{accattoli2024vanillasequentcalculuscallbyvalue}} for the details of the proof, which is adapted from Accattoli \cite{DBLP:journals/lmcs/Accattoli23}.

\paragraph{How Our Approach Fits in the Literature.} Certainly, ours is not the first computational interpretation of a sequent calculus for MIL, whether \emph{vanilla} or not. Essentially, there are three kinds of interpretations in the literature. For the first two, we adopt Mint's local / global terminology \cite{DBLP:conf/tableaux/Mints97}:
\begin{enumerate}
\item \emph{The local approach}. It decorates sequent proofs using a language of terms \emph{without application} and with $\letexp$-expressions (or explicit substitutions) for both the cut rule and the left rule for $\imply$, as we do. In particular, this approach does \emph{not} use ordinary $\l$-terms to decorate proofs. These alternative terms are then endowed with (usually many) rewriting rules mimicking the \emph{propagation} of cuts (that is, they are not at a distance). This approach is followed \eg by Gallier \cite{DBLP:journals/tcs/Gallier93}, Ohori \cite{DBLP:conf/tlca/Ohori99}, Dyckhoff and Lengrand \cite{DBLP:conf/cie/DyckhoffL06,DBLP:journals/logcom/DyckhoffL07}, and Cerrito and Kesner \cite{DBLP:journals/tcs/CerritoK04}. Gallier and Cerrito and Kesner do not notice the connection with \cbv, while Ohori and Dyckhoff and Lengrand do embrace it, but use non-standard sequent calculi and complex sets of rewriting rules. Gallier   (1993) revisits standard proof theoretical concepts under the influence of linear logic. Cerrito and Kesner (2004) develop a Curry-Howard for pattern matching. Ohori's (1999) aim is to show a connection between proof theory and Sabry and Felleisen's \emph{A-normal forms} \cite{DBLP:conf/lfp/SabryF92,DBLP:conf/pldi/FlanaganSDF93}. This paper can be seen as a re-elaboration of Ohori's work, where the focus is a minimalistic logical foundation rather than A-normal forms. Our aim is also similar to Dyckhoff and Lengrand \cite{DBLP:journals/logcom/DyckhoffL07} (2007) but the outcome is considerably simpler: they have three notions of cut and 14 rewriting rules, while we have only one cut and one rewriting rule.

\item \emph{The global approach}. It amounts to decorate sequent proofs with ordinary $\l$-terms and endowing them with call-by-name rewriting. It is rooted in Prawitz's \emph{many-to-one} translation of sequent calculus to natural deduction \cite{prawitz65}, it possibly first appears with proof terms in Pottinger \cite{pottinger77}, and it is nicely presented by Barendregt and Ghilezan using the vanilla sequent calculus \cite{DBLP:journals/jfp/BarendregtG00}. The idea is to decorate rules $(\cut)$ and $(\implyLeftRule)$ using meta-level substitution, as follows:
\begin{center}
\begin{tabular}{c@{\hspace{.12cm}}|@{\hspace{.12cm}}  c}
	\AxiomC{$  \multiForm \vdash \tmtwo \hastype\form$}
	\AxiomC{$  \multiForm, \var\hastype\form \vdash \tm\hastype\formtwo$}
	\RightLabel{$ \cut $}
	\BinaryInfC{$  \multiForm \vdash \subi{\tmtwo}\var \tm \hastype\formtwo$}
	\DisplayProof
&
	\AxiomC{$  \multiForm \vdash \tmtwo \hastype\form$}
	\AxiomC{$  \multiForm, \var\hastype\formtwo \vdash \tm\hastype\formthree$}
	\RightLabel{$ \implyLeftRule $}
	\BinaryInfC{$  \multiForm, \vartwo\hastype \form \imply \formtwo \vdash \subi{\vartwo\tmtwo}\var \tm \hastype\formthree$}
	\DisplayProof
	\end{tabular}
\end{center}
The serious drawback of the global approach is the potential \emph{size mismatch} between a sequent calculus proof $\pi$ and the $\l$-term $\tmthree_\pi$ decorating $\pi$, stemming from the fact that in the sequent calculus the formula $\formtwo$ in $(\implyLeftRule)$ can  be treated non-linearly, that is, it can be weakened or contracted, leading to duplications/erasures of $\vartwo\tmtwo$. The non-linear use of $\formtwo$ is referred to as \emph{the root of all evil} by Danos et al. \cite{danos93wll}, in their study of representations of System F in linear sequent calculi.

\item \emph{The stoup approach.} A third approach departs from the vanilla sequent calculus, adopting a form of enriched sequent or adding further judgements to the deductive system. Typically, at the logical level it uses \emph{two} judgements $\multiForm;\vdash \form$ and $\multiForm;\formtwo\vdash \form$, where the space between the semicolon and $\vdash$ is called \emph{stoup}---terminology due to Girard \cite{DBLP:journals/mscs/Girard91}---and it is either empty or contains a single, distinguished formula. The non-decorated version of rule $(\implyLeftRule)$ then becomes:
\begin{center}
	\AxiomC{$  \multiForm ;\vdash \form$}
	\AxiomC{$  \multiForm ; \formtwo \vdash \formthree$}
	\RightLabel{$ \implyLeftRule $}
	\BinaryInfC{$  \multiForm; \form \imply \formtwo \vdash \formthree$}
	\DisplayProof
\end{center}
The key point is that the rules of the system (which can be defined in various ways) treat the stoup \emph{linearly}, that is, the stoup formula is never weakened nor contracted. This fact is what circumvents the size mismatch of the global approach, and thus allows for a good match between sequent proofs with stoup and ordinary $\l$-terms. It is usually associated with \cbn evaluation.

The stoup approach is first studied by Danos et al. \cite{danos93wll}, building on Girard \cite{DBLP:journals/mscs/Girard91}. It then became the basis for Herbelin \cite{DBLP:conf/csl/Herbelin94}, who studies the intuitionistic \cbn case,  and then Curien and Herbelin \cite{DBLP:conf/icfp/CurienH00}, who study the classical case in both \cbn and \cbv, with three judgements and stoup on both sides of the sequent. From \cite{DBLP:conf/icfp/CurienH00}, it is easy to extract a \cbv intuitionistic fragment (as done for instance by Accattoli and Guerrieri \cite{DBLP:conf/aplas/AccattoliG16}), but  one obtains a sequent calculus with stoup, not the \emph{vanilla} one. A different presentation of essentially the same system is the mentioned approach to LJQ by Dyckhoff and Lengrand \cite{DBLP:conf/cie/DyckhoffL06,DBLP:journals/logcom/DyckhoffL07}, who use two distinct judgements---with one dedicated to values and connected to focalization---rather than the stoup.
\end{enumerate}
Nowadays, a number of works have built over Curien and Herbelin's work, for instance \cite{DBLP:conf/icalp/Santo00,DBLP:conf/icfp/Wadler03,DBLP:conf/tlca/Santo07,DBLP:journals/tcs/DoughertyGL08,DBLP:conf/tlca/AriolaHS11,DBLP:conf/lics/Munch-Maccagnoni15,DBLP:conf/popl/CurienFM16,DBLP:conf/icfp/DownenMAJ16,DBLP:journals/jfp/DownenA18,DBLP:journals/toplas/Miquey19}, to the point that it became the standard computational interpretation of sequent calculi. 

Our work is orthogonal, and somewhat more basic, as it looks at the intuitionistic case without any form of stoup, focalization, or separate judgement.

\paragraph{Non-Canonicity \emph{vs} Sharing.} In the proof theoretical literature, the vanilla sequent calculus is often criticized as \emph{non-canonical} because (\cbn) normal $\l$-terms have more than one cut-free proof, when one embraces the global approach---this is another face of Danos et al.'s \emph{root of all evil} mentioned above. 

A key point of our fresh perspective is that this is a \emph{feature} rather than a drawback: it is simply the fact that \cbv normal terms such as $\letin\var{\vartwo\varthree}{(\varfour\var\var\var)}$ can be seen as \emph{shared} representations of \cbn normal forms such as $\varfour(\vartwo\varthree)(\vartwo\varthree)(\vartwo\varthree)$, and that a term can be shared in various ways---another one is $\letin\var{\vartwo\varthree}{(\varfour\var\var(\vartwo\varthree))}$. Therefore, the vanilla sequent calculus gives a first-class status to sub-term sharing for normal forms.

More precisely, $\letexp$-expressions annotate cuts and give a first-class status to sub-term sharing independently of the proof system. The concept of \emph{value} induces a dynamic / static refinement of sub-term sharing. Dynamically, only values are \emph{unshared} (that is, duplicated or erased). Non-values are never unshared, thus end up being \emph{statically shared} in normal forms,  as in $\letin\var{\vartwo\varthree}{(\varfour\var\var\var)}$, making sub-term sharing visible \emph{denotationally}. These two roles are entangled in natural deduction for MIL while they are handled separately in sequent calculus via its two \emph{left rules}, namely cut captures dynamic sharing and subtraction captures static sharing. Therefore, the left rules of the sequent calculus for MIL encapsulate the sharing features of values.

For an extensive informal discussion about the deep relationship between sharing and \cbv see Accattoli's dissemination paper \cite{DBLP:conf/onward/Accattoli23}.

\paragraph{Proofs} Omitted proofs are in the \withproofs{Appendix}\withoutproofs{associated technical report on arXiv \cite{accattoli2024vanillasequentcalculuscallbyvalue}}.
\section{Natural Deduction, Call-by-Name, and Rewriting at a Distance}
\label{sect:natural-name-distance}
In this section, we give our presentation of Howard's standard correspondence between natural deduction for minimal intuitionistic logic and the simply typed (call-by-name) $\l$-calculus \cite{Howard1980-HOWTFN-2}, here also referred to as the \emph{natural $\l$-calculus}, to distinguish it from the one that shall be associated to the sequent calculus. The correspondence is in \reffig{ND-name}. Curiously, in \cite{Howard1980-HOWTFN-2} Howard mentions that his work stems from Tait's "discovery of the close correspondence between cut elimination and reduction of $\l$-terms" \cite{TAIT1965176}.
\begin{figure}[t!]
\centering
 \begin{tabular}{c}
 \textsc{(Natural) $\l$-Calculus}
\\[4pt]

$\begin{array}{r\colspace rll\colspace|\colspace r\colspace rll}
\textsc{Terms} &\ndterms \ni \tm,\tmtwo & \grameq & \var\in\ndvars \mid \la\var\tm \mid \tm\tmtwo
&
\textsc{Root $\beta$}
&
(\la\var\tm)\tmtwo \rtob \subi\tmtwo\var\tm
\\[4pt]
\textsc{Contexts} &\ctx,\ctxtwo & \grameq & \ctxhole \mid \la\var\ctx \mid \ctx\tmtwo \mid \tm\ctx
&
\textsc{Ctx closure}
&
\AxiomC{$\tm \rtob \tm'$}
\UnaryInfC{$\ctxp{\tm} \tob \ctxp{\tm'}$}
\DisplayProof
\end{array}$
\\[3pt]\hline
\textsc{Decorated (Additive) Natural Deduction}
\\[4pt]
$\begin{array}{r\colspace rll}
\textsc{Formulas} & \form,\formtwo,\formthree & \grameq & \aform \mid \form \imply \formtwo
\end{array}$
\\[4pt]
\begin{tabular}{c\colspace|\colspace c}	
\begin{tabular}{c}	
	\AxiomC{}
	\RightLabel{$\ax$}
	\UnaryInfC{$  \multiForm, \var\hastype\form \nproves \var\hastype \form$}	
	\DisplayProof
\\[6pt]
		\AxiomC{$  \var\hastype\form,\multiForm \nproves \tm\hastype\formtwo$}
	\RightLabel{$ \implyRightRule $}
	\UnaryInfC{$  \multiForm \nproves \la\var\tm\hastype\form \imply \formtwo$}
	\DisplayProof
\\[6pt]
	\AxiomC{$  \multiForm \nproves \tm \hastype\form \imply \formtwo$}
	\AxiomC{$  \multiForm \nproves \tmtwo\hastype\form$}
	\RightLabel{$ @ $}
	\BinaryInfC{$  \multiForm \nproves  \tm\tmtwo \hastype\formtwo$}
	\DisplayProof
\end{tabular}
&
\begin{tabular}{c}
\textsc{Detours}
\\[3pt]
		\AxiomC{$  \var\hastype\form,\multiForm \nproves \tm\hastype\formtwo$}
	\RightLabel{$ \implyRightRule $}
	\UnaryInfC{$  \multiForm \nproves \la\var\tm\hastype\form \imply \formtwo$}

	\AxiomC{$  \multiForm \nproves \tmtwo\hastype\form$}
	\RightLabel{$ @ $}
	\BinaryInfC{$  \multiForm \nproves  (\la\var\tm)\tmtwo \hastype\formtwo$}
	\DisplayProof
\end{tabular}
\end{tabular}
%
%
%
%
\end{tabular}
\caption{$\l$-calculus and natural deduction $\nproves$.
}
\label{fig:ND-name}
\end{figure}
\paragraph{Routine Definitions: Terms and Contexts.} We assume given a countable set of variables $\ndvars$, where $\ndsym$ stresses that they are for the natural $\l$-calculus. The meta-level capture-avoiding substitution of $\tmtwo$ for $\var$ in $\tm$ is denoted with $\subi\tmtwo\var\tm$. 

 Contexts are terms with exactly one occurrence of the \emph{hole} $\ctxhole$, an additional constant, standing for a removed sub-term. We shall use various notions of contexts. For the $\l$-calculus, the most general ones are \emph{(general) contexts} $\ctx$, which simply allow the hole to be anywhere. The main operation about contexts is \emph{plugging} $\ctxp{\tm}$ where the hole $\ctxhole$ in context $\ctx$ is replaced by $\tm$. Plugging, as usual with contexts, can
capture variables---for instance $(\la\var\ctxhole \tm)\ctxholep\var = \la\var\var\tm$. 

\paragraph{Routine Definitions: Types and Derivations.} Types are built out of an unspecified atomic type $\aform$ and the implication $\cdot \imply \cdot$ connective. Type contexts $\multiForm$ are implicitly considered modulo exchange, that is, $\multiForm$ is a partial function from variables to formulas  
such that  $\dom{\typctx} \defeq \{\var \mid \typctx(\var) \mbox{ is defined}\}$ is finite, usually written as $\var_1 \hastype \form_1, \dots, \var_n \hastype \form_n$ (with $n \in 
\nat$) if $\dom{\typctx} \subseteq \{\var_1, \dots, \var_n\}$ and $\typctx(\var_i) = \form_{i}$ for $1 \leq i \leq 
n$. As it is standard, writing $\multiForm, \var \hastype \form$ implicitly assumes that $\var\notin\dom\multiForm$. 

We write $\tderiv \exder \multiForm \nproves \tm \hastype \form$ if $\tderiv$ is a (\emph{type}) \emph{derivation}, that is, a tree constructed using the rules for $\nproves$, having axioms as leaves, and of final judgment $\multiForm \nproves \tm \hastype \form$.


\paragraph{Left and Right Rules.} The typing rule $(\implyRightRule)$ for abstraction carries a right $r$ subscript, to later distinguish it from the left rule for implication of the sequent calculus. We shall constantly refer to left and right rules, and yet we shall dodge abstract definitions of these notions. Rules $\ax$, $\implyRightRule$, and $@$ are right rules.

\paragraph{Detours.} The logical analogues of $\beta$-redexes are \emph{detours}, which are simply given by a $\implyRightRule$ rule followed by a $@$ rule, that is exactly what is required to type a $\beta$-redex---see \reffig{ND-name}.
$\l$-terms in normal form are exactly those typed by proofs of $\nproves$ without detours.


\begin{figure}[t!]
\centering
 \begin{tabular}{c}
$\begin{array}{r\colspace rll}
\textsc{Terms} &\ndesterms \ni \tm,\tmtwo & \grameq & \var \mid \la\var\tm \mid \tm\tmtwo \mid \sube\tmtwo\var\tm
\\[4pt]
\textsc{Contexts} &\ctx,\ctxtwo & \grameq & \ctxhole \mid \la\var\ctx \mid \ctx\tmtwo \mid \tm\ctx \mid \sube\ctx\var\tm \mid \sube\tmtwo\var\ctx
\\
\textsc{Substitutions ctxs} &\lctx,\lctxtwo & \grameq & \ctxhole \mid  \sube\tmtwo\var\lctx
\end{array}$
\\[5pt]\hline
\begin{tabular}{c|c}
\begin{tabular}{c}
\textsc{Root rules}
\\[4pt]
$\begin{array}{r\colspace rlll}
\textsc{$\beta$ at a distance} &\lctxp{\la\var\tm}\tmtwo & \rtodb & \lctxp{\sube\tmtwo\var\tm}
\\[4pt]
\textsc{Substitution} &\sube\tmtwo\var\tm & \rtos & \subi\tmtwo\var\tm
\end{array}$
\end{tabular}
&
\begin{tabular}{c}
\textsc{Contextual closure}
\\[4pt]
\RightLabel{$\asym{\in}\{\db,\ssym\}$}
					\AxiomC{$\tm \rootRew{\asym} \tm'$}
					\UnaryInfC{$\ctxp\tm \Rew{\asym} \ctxp{\tm'}$}
					\DisplayProof
\\[4pt]
\textsc{Subst. calculus rewriting}
\\
$\tosc \grameq \todb \cup \tos$	
\end{tabular}
\end{tabular}
\\[5pt]\hline
\textsc{Cut rule decorated with an explicit substitution}
\\[4pt]
	\AxiomC{$  \multiForm \nproves \tmtwo\hastype\form$}
	\AxiomC{$  \multiForm, \var\hastype\form \nproves \tm \hastype\formtwo$}
	\RightLabel{$\cut $}
	\BinaryInfC{$  \multiForm \nproves  \sube\tmtwo\var\tm \hastype\formtwo$}
	\DisplayProof
\end{tabular}
\caption{Substitution Calculus (SC).}
\label{fig:ND-name-ES}
\end{figure}
\paragraph{Explicit Substitutions.} In \reffig{ND-name-ES}, the $\l$-calculus is extended with an explicit substitution (shortened to ES) constructor $\sube\tmtwo\var\tm$, which binds $\var$ in $\tm$. It can be thought as a more compact notation for $\letin\var\tmtwo\tm$, with the slight difference that the evaluation order is not fixed in $\sube\tmtwo\var\tm$. At the logical level, ESs are decorations for cuts, as shown in \reffig{ND-name-ES}. Note that in the literature about $\l$-calculi at a distance, ESs are usually rather written using the mirrored construct $\tm\esub\var\tmtwo$. We here prefer to use $\sube\tmtwo\var\tm$ because it more faithfully reflects the structure of type derivations.

A point of view underlying our study is that ESs are a form of \emph{sub-term sharing}, meaning that $\sube\tm\var(\var\var)$ is a version of $\tm\tm$ where $\tm$ is shared.

In this paper, cut shall be considered as a left rule (for its right premise), according to the view that left rules account for sharing outlined at the end of the introduction. 

\paragraph{Rewriting Rules.} We endow natural $\l$-terms with ESs with \emph{small-step} rewriting rules, that is, rules based on meta-level substitution, obtaining Accattoli and Kesner's \emph{substitution calculus} (SC) \cite{DBLP:conf/lpar/AccattoliK12}, of rewriting relation $\tosc$. It is composed of two rewriting rules, namely the standard substitution rule $\tos$, and the perhaps less standard \emph{$\beta$ at a distance} rule $\todb$. Rule $\todb$ generalizes $\beta$-redexes as to fire even when there are some ESs (\ie left rules) in between the abstraction and the argument, fact that is formalized in the definition of the rule via the (possibly empty) substitution context $\sctx$ (mnemonic: $\sctx$ stands for \emph{List} of substitutions). For instance, 
\begin{equation}
\begin{array}{c\colspace c\colspace c}
(\sube\tmthree\varthree\sube\tmtwo\vartwo(\la\var\tm)) \tmfour\tmfive 
& \todb &
(\sube\tmthree\varthree\sube\tmtwo\vartwo\sube\tmfour\var\tm) \tmfive.
\end{array}
\label{eq:distance}
\end{equation}
This kind of rule circumvents the need of having commuting rewriting rules such as $\sube\vartwo\tmtwo(\la\var\tm) \to \la\var\sube\vartwo\tmtwo\tm$ or $(\sube\vartwo\tmtwo\tm)\tmthree \to \sube\vartwo\tmtwo(\tm\tmthree)$ in order to expose the $\beta$-redex in \refequa{distance}. The rewriting theory of ESs at a distance has simpler proofs and stronger properties (\eg residuals) than for commuting-based ESs, see Accattoli and Kesner \cite{DBLP:conf/csl/AccattoliK10,DBLP:conf/lpar/AccattoliK12,DBLP:conf/popl/AccattoliBKL14}.

Note that ESs can always be reduced, thus there are no ESs in $\tosc$-normal forms. In other words, normal forms are sharing-free.

 Clearly, the substitution calculus simulates $\beta$, since $(\la\var\tm)\tmtwo \todb\sube\tmtwo\var\tm \tos \subi\tmtwo\var\tm$.

\section{The Natural $\l$-Calculus By Value}
\label{sect:vanilla}
In this section, we discuss the presentation of \cbv in a natural deduction calculus, first in Plotkin's style and then with cuts / ESs. In particular, we shall see the advantages and the limits of ESs for \cbv.


\begin{figure}[t!]
	\begin{center}
		\begin{tabular}{c}
			\begin{tabular}{c@{\hspace{0.5cm}}|@{\hspace{0.5cm}}c}
				{\textsc{Language}}
				
				&
				
				{\textsc{Root rule}}
				\\[4pt]
				$\arraycolsep=3pt
				\begin{array}{rrll}
					\textsc{Terms} & \ndterms \ni \tm,\tmtwo,\tmthree & \grameq& \val \mid \tm\tmtwo
					\\
					\textsc{Values} & \val,\valtwo & \grameq & \var \mid \la\var\tm
					
				\end{array}$ & $\arraycolsep=3pt
			\begin{array}{lll}
			(\la\var\tm)\val & \rtobv & \tm\isub\var\val
		\end{array}$
				
			\end{tabular}
			\\[5pt]
			\hline
			\\[-9pt]
			\tabcolsep = 10pt
				\begin{tabular}{c@{\hspace{0.1cm}}|@{\hspace{0.1cm}}c}
				 \begin{tabular}{c}
					\textsc{Weak Evaluation }
					\\[4pt]
					$\begin{array}{r@{\hspace{.15cm}}r@{\hspace{.1cm}}l@{\hspace{.1cm}}ll}
							\textsc{Weak Ctxs} & \evctx & \grameq & \ctxhole \mid \evctx\tm \mid \tm\evctx 
						\end{array}$
					\\	
						\AxiomC{$\tm \rtobv \tm'$}
						\UnaryInfC{$\evctxp\tm \towbv \evctxp{\tm'}$}
						\DisplayProof
					\\
				\end{tabular}
				&
				\begin{tabular}{c@{\hspace{0.5cm}}c}
					\textsc{Strong Evaluation}
					\\[4pt]
					$\begin{array}{r@{\hspace{.15cm}}r@{\hspace{.1cm}}l@{\hspace{.1cm}}ll}
						\textsc{Strong Ctxs} & \ctx & \grameq & \ctxhole \mid \ctx\tm \mid \tm\ctx  \mid \la\var\ctx
					\end{array}$
					\\ 
					\AxiomC{$\tm \rtobv \tm'$}
					\UnaryInfC{$\ctxp\tm \tobv \ctxp{\tm'}$}
					\DisplayProof
				\end{tabular}
				\end{tabular}
				\end{tabular}

\caption{Plotkin's (natural) \cbv $\l$-calculus.}
\label{fig:plotkin}
\end{center}
\end{figure}

\paragraph{Plotkin's Call-by-Value and Open Terms.} The definition of \cbv \emph{à la Plotkin} following the modern presentation by Dal Lago and Martini \cite{DBLP:journals/tcs/LagoM08} is in \reffig{plotkin}. Values $\val$ are defined as variables and abstraction, and the $\betav$-rule is obtained by restricting $\beta$-redexes to fire only when the argument is a value. Strong evaluation allows one to reduce $\betav$-redexes everywhere in a term, and weak evaluation instead forbids $\betav$-redexes under abstraction.

It is well-known that Plotkin's approach works smoothly only in the important and yet limited case of weak evaluation of closed terms. Open terms are an issue, and even more so is strong evaluation, because they cause \emph{stuck $\beta$-redexes} such as $(\la\var \tm) (\vartwo\varthree)$ where the argument is $\tobv$-normal and not a value, thus the $\betav$-rule cannot fire. These stuck configurations are problematic, as first noticed by Paolini and Ronchi della Rocca \cite{DBLP:journals/ita/PaoliniR99,DBLP:conf/ictcs/Paolini01,DBLP:series/txtcs/RoccaP04}. The easiest way of stating the problem is that the paradigmatic looping term $\Omega \defeq \delta\delta$, where $\delta \defeq \la\var\var\var$, and its variant $\Omega_{\symfont{stuck}} \defeq (\la\var\delta)(\vartwo\varthree)\delta$ are contextual equivalent and yet $\Omega$ loops while  $\Omega_{\symfont{stuck}}$ is normal in Plotkin's approach. Terms such as $\Omega_{\symfont{stuck}}$ are sometimes called \emph{premature normal forms}, and break expected properties of Plotkin's calculus with respect to denotational models, see Accattoli and Guerrieri \cite{DBLP:conf/aplas/AccattoliG18}. 

Logically, the issue is that not all detours can be eliminated using Plotkin's $\betav$-rule.

\paragraph{Substitutions by Value.} The issue with open terms has been studied at length by Accattoli and Guerrieri and co-authors \cite{DBLP:conf/aplas/AccattoliG16,DBLP:journals/lmcs/GuerrieriPR17,DBLP:conf/aplas/AccattoliG18,DBLP:journals/pacmpl/AccattoliG22,DBLP:conf/fossacs/AccattoliL24}, who in \cite{DBLP:conf/aplas/AccattoliG16} study and compare various ways of circumventing it. One of the most flexible and studied solutions amounts to add \emph{ESs at a distance}. The framework is Accattoli and Paolini's \emph{value substitution calculus} (VSC) \cite{DBLP:conf/flops/AccattoliP12}, defined in \reffig{ND-value-ES}, which is the variant of the substitution calculus (of the previous section) modelled over the \cbv translation of $\l$-calculus in linear logic proof nets. 
\begin{figure}[t!]
\centering
 \begin{tabular}{c}
\textsc{Terms}\ \ \ $\ndesterms$\ \ \ (as for the SC)
\\[4pt]
\hline 
\begin{tabular}{c\colspace |\colspace c}
\begin{tabular}{c}
\textsc{Root rules}
\\[4pt]
$\begin{array}{r\colspace rlll}
\textsc{$\beta$ at a distance} &\lctxp{\la\var\tm}\tmtwo & \rtodb & \lctxp{\sube\tmtwo\var\tm}
\\[4pt]
\textsc{Value substitution} &\sube{\sctxp\val}\var\tm & \rtovs & \sctxp{\subi\val\var\tm}
\\[4pt]
\textsc{Notation} &\tovsc & \grameq & \todb \cup \tovs
\end{array}$
\end{tabular}
&
\begin{tabular}{c}
\textsc{Contextual closure}
\\
$\asym{\in}\{\db,\vssym\}$
\\[4pt]

					\AxiomC{$\tm \rootRew{\asym} \tm'$}
					\UnaryInfC{$\ctxp\tm \Rew{\asym} \ctxp{\tm'}$}
					\DisplayProof
\end{tabular}
%
\end{tabular}
\end{tabular}
\caption{Value Substitution Calculus (VSC).}
\label{fig:ND-value-ES}
\end{figure}
The VSC has the same $\beta$ at a distance rule $\todb$ of the SC, which does \emph{not} require the argument to be a value. It has instead a different substitution rule, which is where the value restriction takes place. Its value substitution rule $\tovs$ uses distance (\ie the substitution context $\sctx$ around the value $\val$). For instance, $\sube{\sube\vartwo\tmthree\val}\var(\varthree\var\var) \to \sube\vartwo\tmthree(\varthree\val\val)$.
Distance allows one to avoid dedicated commuting rules such as $\sube{\sube\vartwo\tmthree\tmtwo}\var\tm \to \sube\vartwo\tmthree\sube{\tmtwo}\var\tm$ which are often found in \cbv calculi with $\letexp$-expressions, for instance in Moggi's calculus \cite{Moggi88tech,DBLP:conf/lics/Moggi89}, where that rule is called \emph{assoc}.

\paragraph{Advantages of ESs by Value.} The VSC handles open terms correctly thanks to distance and having moved the value restrictions from $\beta$-redexes to substitution redexes. As evidence of correct behavior, let's have a look at the variant $\Omega_{\symfont{stuck}}$ of $\Omega$ that is a premature normal form in Plotkin's calculus. It now (correctly) diverges in the VSC:
\begin{center}
$\begin{array}{lllllllllll}
\Omega_{\symfont{stuck}} &\defeq& (\la\var\delta)(\vartwo\varthree)\delta & \todb & (\sube{\vartwo\varthree}\var\delta)\delta 
\\
& \todb & \sube{\vartwo\varthree}\var\sube\delta\varfour(\varfour\varfour) &\tovs& \sube{\vartwo\varthree}\var(\delta\delta) &\todb&\ldots
\end{array}$
\end{center}
The VSC has a number of operational and denotational good properties, as shown by Accattoli and co-authors in various recent works \cite{DBLP:conf/lics/AccattoliCC21,DBLP:journals/pacmpl/AccattoliG22,DBLP:conf/fossacs/AccattoliL24}, essentially providing a \cbv calculus that mimics in \cbv most of the good properties of the \cbn $\l$-calculus, considerably improving over Plotkin's presentation. At the same time, the VSC and Plotkin's calculus induce the same contextual equivalence on closed terms (see \cite{DBLP:journals/pacmpl/AccattoliG22}), thus the VSC is a \emph{conservative} refinement of Plotkin's calculus.

\paragraph{Limits of ESs by Value.} The VSC, however, is not free from glitches. In contrast to the \cbn case, in the VSC \emph{not all ESs are eliminable}, because of the value constraint.    
Namely, an ES such as $\sube{\vartwo\varthree}\var\tm$ cannot be eliminated, since $\vartwo\varthree$ is normal and not a value. It might seem that this is the same issue that we pointed out for Plotkin's approach, for which some detours are not eliminable.
The situation however is different: in Plotkin's case, ineliminable detours break operational and denotational properties, while ineliminable ESs do not cause similar problems in the VSC. 

It is important to stress that---in itself---the presence of ineliminable ESs is not a drawback. They actually are a \emph{feature} of \cbv and of the VSC. The fact that some ESs are ineliminable, indeed, means that sub-term sharing ends up in normal forms, and thus becomes \emph{denotationally visible}, which is a good thing, since it opens the way to a mathematical understanding of sharing and efficiency. In Ehrhard's \cbv relational denotational model \cite{DBLP:conf/csl/Ehrhard12}, for instance, $\sube{\vartwo\varthree}\var(\varfour\var\var)$ and $\varfour(\vartwo\varthree)(\vartwo\varthree)$ have \emph{different} interpretations; they instead have the \emph{same} interpretation in the \cbn relational model. An extensive high-level discussion of the relationship between sharing and \cbv can be found in Accattoli's dissemination paper \cite{DBLP:conf/onward/Accattoli23}.

The glitch of the VSC is the fact that its ESs are still typed with cuts, as in the SC. Therefore, if some ESs are ineliminable then some cuts are ineliminable, in the \cbv interpretation of natural deduction with cuts. This clearly goes against the expected property of the cut rule, which is \emph{admissibility}, that is, that all cuts are eliminable. What happens in the VSC is that only cuts containing values are eliminable. 

Ideally, one would like to have two \emph{separate} constructors, one for eliminable cuts and one for ineliminable cuts. This does not seem to be naturally achievable in natural deduction, where the two are entangled. The aim of this paper is to show that, instead, it is exactly what naturally happens in the sequent calculus.
\section{Proof Terms for the Vanilla Sequent Calculus}
In this section, we present the sequent calculus for minimal intuitionistic logic (MIL). We adopt an additive presentation of the MIL fragment of Gentzen's original one \cite{Gentzen1964} (which was multiplicative), mimicking the standard use of  additive natural deduction for Howard's correspondence with the $\l$-calculus (see \eg Sørensen and Urzyczyn \cite{sorensen2006lectures}), despite the fact that his original presentation was multiplicative \cite{Howard1980-HOWTFN-2}. The rules are in \reffig{SC-static}, which we refer to as \emph{vanilla}, to stress that we do not considered distinguished formulas (also called \emph{stoup}) on the left of sequent symbol $\sproves$, nor any other tweak. 

The main difference with natural deduction is that the application rule $(@)$ is replaced by the left rule  $(\implyLeftRule)$ for $\imply$ (see \reffig{SC-static}, the notation $\multiForm \tctxplus \svartwo\hastype \form \imply \formtwo$ is explained below), beyond having the cut rule from the start (while in natural deduction we added it only at a later moment). For the sequent calculus, we shall re-use most basic concepts introduced for natural deduction without re-defining them.

\begin{figure}[t!]
\centering
 \begin{tabular}{c}
$\begin{array}{r\colspace rllllll}
\textsc{Vanilla terms} & \scterms \ni \stm,\stmtwo,\stmthree & \grameq &  \svar \in\scvars \mid \la\svar\stm \mid
\cuta\stmtwo\svar\stm \mid \nsuba\svartwo\stmtwo\svar \stm
\end{array}$
\\[5pt]\hline
\textsc{Decorated (Additive) Vanilla Sequent Calculus}
\\[4pt]
\begin{tabular}{c\colspace|\colspace c}	
	\AxiomC{$  \multiForm \sproves \stmtwo\hastype\form$}
	\AxiomC{$  \multiForm, \svar\hastype\form \sproves \stm \hastype\formtwo$}
	\RightLabel{$\cut $}
	\BinaryInfC{$  \multiForm \sproves  \cuta\stmtwo\svar\stm \hastype\formtwo$}
	\DisplayProof
&
	\AxiomC{}
	\RightLabel{$\ax$}
	\UnaryInfC{$  \multiForm, \svar\hastype\form \sproves \svar\hastype \form$}	
	\DisplayProof
\\[6pt]
	\AxiomC{$  \multiForm \sproves \stmtwo \hastype\form$}
	\AxiomC{$  \multiForm, \svar\hastype\formtwo \sproves \stm\hastype\formthree$}
	\RightLabel{$ \implyLeftRule $}
	\BinaryInfC{$  \multiForm \tctxplus \svartwo\hastype \form \imply \formtwo \sproves \nsuba\svartwo\stmtwo\svar \stm \hastype\formthree$}
	\DisplayProof
&
		\AxiomC{$  \svar\hastype\form,\multiForm \sproves \stm\hastype\formtwo$}
	\RightLabel{$ \implyRightRule $}
	\UnaryInfC{$  \multiForm \sproves \la\svar\stm\hastype\form \imply \formtwo$}
	\DisplayProof
\end{tabular}
\end{tabular}
\caption{Vanilla $\l$-terms and how they decorate the vanilla sequent calculus $\sproves$ for MIL.}
\label{fig:SC-static}
\end{figure}

\paragraph{Vanilla $\l$-Terms and the Left Rule for $\imply$.} We adopt what Mint calls the \emph{local approach} \cite{DBLP:conf/tableaux/Mints97}, followed also, for instance, by Gallier \cite[Section 10]{DBLP:journals/tcs/Gallier93}, Dyckhoff and Lengrand \cite{DBLP:conf/cie/DyckhoffL06,DBLP:journals/logcom/DyckhoffL07}, Ohori \cite{DBLP:conf/tlca/Ohori99}, and Cerrito and Kesner \cite{DBLP:journals/tcs/CerritoK04}, introducing a language $\scterms$ of \emph{vanilla} $\l$-terms  faithfully coding the structure of sequent proofs. We write the terms of $\scterms$ using a different font $\stm,\stmtwo,\stmthree$ (with respect to the one used for natural $\l$-terms), also for the variables $\svar,\svartwo,\svarthree,\svarfour$, whose set is noted $\scvars$. The main point  of $\scterms$ is that there are no applications (which decorate $(@)$ rules). They are replaced by  a new constructor $\suba\svartwo\stmtwo\svar\stm$, dubbed here \emph{subtraction} (following Accattoli \cite{DBLP:journals/lmcs/Accattoli23}), and which binds $\svar$ in $\stm$ and adds a free occurrence of $\svartwo$. Intuitively, subtraction can be thought of as $\letin\svar{\svartwo\stmtwo}\stm$ (which is the notation used by Gallier \cite{DBLP:journals/tcs/Gallier93}) but we prefer to avoid such a notation, for various reasons. Firstly, writing $\svartwo\stmtwo$ we would still be resting on application, while we want to stress that application is not part of the new language, since there is no $(@)$ rule at the logical level. Secondly, in $\letin\svar{\svartwo\stmtwo}\stm$ the two active premises $\svar$ and $\stmtwo$ of the $(\implyLeftRule)$ logical rule are given asymmetric roles, while in $\suba\svartwo\stmtwo\svar\stm$ their roles are symmetric. Thirdly, $\letexp$-expressions are verbose, thus we adopt a bracket notation. 

Note that the typing rule for subtractions uses a new notation $\multiForm \tctxplus \svar\hastype \form$, defined as:
\[\begin{array}{r\colspace c\colspace l\colspace \colspace l}
\multiForm \tctxplus \svar\hastype \form & \defeq & \multiForm, \svar\hastype \form & \mbox{if }\svar\notin\dom\multiForm
\\[4pt]
\multiForm \tctxplus \svar\hastype \form & \defeq & \multiForm & \mbox{if }\svar\in\multiForm \mbox{ and }\multiForm(\svar) = \form
\end{array}\]
and it is undefined if $\svar\in\multiForm$ and $\multiForm(\svar) \neq \form$. The notation is needed because the additive presentation might have to do on-the-fly contractions on the newly introduced formula $\form\imply\formtwo$, as for instance in the following case, where the free variable occurrence introduced by rule $(\implyLeftRule)$ is immediately identified with the already existing variable $\svartwo$:
\[
\AxiomC{$  \multiForm \sproves \stmtwo \hastype\form$}
	\AxiomC{}
	\RightLabel{$ \ax $}
	\UnaryInfC{$  \multiForm, \svar\hastype\formtwo, \svartwo\hastype\form \imply \formtwo \sproves \svartwo\hastype\form \imply\formtwo$}
	\RightLabel{$ \implyLeftRule $}
	\BinaryInfC{$  \multiForm, \svartwo\hastype \form \imply \formtwo \sproves \nsuba\svartwo\stmtwo\svar \svartwo \hastype\formtwo$}
	\DisplayProof
\]
This phenomenon is not present in natural deduction because therein formulas and variables are introduced on the left of $\nproves$ only by axioms. 

Given that weakenings are freely available (in axioms), rule $(\implyLeftRule)$ can alternatively be formulated without using $\tctxplus$, as follows:
\[	\AxiomC{$  \multiForm,\svartwo\hastype \form \imply \formtwo \sproves \stmtwo \hastype\form$}
	\AxiomC{$  \multiForm, \svar\hastype\formtwo, \svartwo\hastype \form \imply \formtwo \sproves \stm\hastype\formthree$}
	\RightLabel{$ \implyLeftRule' $}
	\BinaryInfC{$  \multiForm, \svartwo\hastype \form \imply \formtwo \sproves \nsuba\svartwo\stmtwo\svar \stm \hastype\formthree$}
	\DisplayProof
\]
This approach is adopted for instance by Herbelin \cite{HerbelinPhD} and Dyckhoff and Lengrand \cite{DBLP:journals/logcom/DyckhoffL07}. We preferred to avoid rule $(\implyLeftRule')$, however, as we find it counter-intuitive. Yet another equivalent approach is adopting a mixed multiplicative-additive presentation: having the simpler rule $(\implyLeftRule)$ of \reffig{SC-intro} that does not use $\tctxplus$ and adding a stand-alone contraction  rule (which is multiplicative) for the left-hand side of $\sproves$.


\paragraph{Different Decoration for Cuts.} Another difference between $\ndterms$ and $\scterms$ is that, since in $\scterms$ cuts shall play a slightly different role than in the VSC / natural deduction by value, for the sake of clarity we decorate them differently, with $\cuta\stmtwo\svar\stm$ rather than with $\sube\stmtwo\svar\stm$.

\begin{figure}[t!]
\centering
 \begin{tabular}{c}
 \begin{tabular}{c\colspace |\colspace c}
 \textsc{Natural $\l$ to Vanilla $\l$} & \textsc{Vanilla $\l$ to Natural $\l$}
 \\[4pt]
$\begin{array}{r l l@{\hspace{.3cm}} l}
\trndtosc\var & \defeq & \svar
\\[3pt]

\trndtosc{\la\var\tm} & \defeq & \la\svar \trndtosc\tm
\\[3pt]

\trndtosc{\sube\tmtwo\var\tm} & \defeq & \cuta{\trndtosc\tmtwo}\svar \trndtosc\tm
\\[3pt]

\trndtosc{\tm\tmtwo} & \defeq & 
\cuta{\trndtosc\tm}\svar \nsuba\svar{\trndtosc\tmtwo}\svartwo \svartwo & \svar,\svartwo \mbox{ fresh} 
\end{array}$
&
$\begin{array}{r l l l}
\trsctond\svar & \defeq & \var
\\[3pt]

\trsctond{\la\svar\stm} & \defeq & \la\var \trsctond\stm
\\[3pt]

\trsctond{\cuta\stmtwo\svar\stm} & \defeq & \sube{\trsctond\stmtwo}\var \trsctond\stm
\\[3pt]

\trsctond{\nsuba\svartwo\stmtwo\svar\stm} & \defeq & \sube{\vartwo\trsctond\stmtwo}\var \trsctond\stm 
\end{array}$
\end{tabular}
\\[5pt]\hline\hline
\\[-5pt]

	\AxiomC{$  \multiForm \nproves \form \imply \formtwo$}
	\AxiomC{$  \multiForm \nproves \form$}
	\RightLabel{$ @ $}
	\BinaryInfC{$  \multiForm \nproves  \formtwo$}
	\DisplayProof
\\[4pt]	
\emph{is simulated in $\sproves$ by}
\\[4pt]
\AxiomC{$  \multiForm \sproves \form \imply \formtwo$}
	\AxiomC{$  \multiForm \sproves \form$}
		\AxiomC{}
		\RightLabel{$ \ax $}
		\UnaryInfC{$  \multiForm,  \formtwo \sproves \formtwo$}
	\RightLabel{$ \implyLeftRule $}
	\BinaryInfC{$  \multiForm,  \form \imply \formtwo \sproves \formtwo$}
\RightLabel{$ \cut $}
\BinaryInfC{$  \multiForm \sproves \formtwo$}
	\DisplayProof
\\[4pt]
\hline\hline
\\[-5pt]
	\AxiomC{$  \multiForm \sproves \form$}
	\AxiomC{$  \multiForm, \formtwo \sproves \formthree$}
	\RightLabel{$ \implyLeftRule $}
	\BinaryInfC{$  \multiForm,  \form \imply \formtwo \sproves \formthree$}
	\DisplayProof
\\[4pt]	
\emph{is simulated in $\nproves$ by}
\\[4pt]
	\AxiomC{}
	\RightLabel{$ \ax $}
	\UnaryInfC{$  \multiForm,  \form \imply \formtwo \nproves \form \imply \formtwo$}
\AxiomC{$  \multiForm \nproves \form$}
\RightLabel{$ @ $}
\BinaryInfC{$  \multiForm, \form \imply \formtwo \nproves  \formtwo$}
		\AxiomC{$  \multiForm,  \formtwo \nproves \formthree$}
	\RightLabel{$ \cut $}
	\BinaryInfC{$  \multiForm, \form \imply \formtwo \nproves \formthree$}
	\DisplayProof
\end{tabular}
\caption{Translations between the natural $\l$-terms (with ESs) and the vanilla $\l$-terms.}
\label{fig:translations}
\end{figure}
\paragraph{Translations.} Translations from the natural $\l$-calculus to the vanilla one and vice-versa are given in \reffig{translations}. They induce translations of the related logical systems, whose key ingredients are highlighted in \reffig{translations}, namely  the simulation of the $(@)$ rule by $\sproves$ and the simulation of the $(\implyLeftRule)$ rule by $\nproves$. Both translations are obtained by introducing axioms and cuts, and are standard. For instance, Girard's \cbn and \cbv translations of natural deduction to linear logic \cite{DBLP:journals/tcs/Girard87} are modal decorations of $\trndtosc\cdot$, and the cuts used in the translation of applications are called \emph{correction cuts} by Danos et al. \cite{danos93wll}.

\paragraph{Translation of Cut-Free Vanilla Terms.} We have not yet introduced rewriting rules for vanilla $\l$-terms (that is, cut elimination), but there are no doubts about the expected notion of normal vanilla terms: they must be the cut-free ones. The next proposition states the starting observation of this work, that is, the fact that cut-free vanilla terms are mapped to natural terms that are \cbv normal, but not necessarily \cbn normal.

\begin{proposition}
\label{prop:cbv-nfs}
Let $\stm\in\scterms$ be cut-free. Then $\trsctond\stm$ is $\tovsc$-normal but not necessarily $\tosc$-normal.
\end{proposition}
\begin{proof}
By induction on $\stm$. For $\tovsc$, simply note that, for cut-free vanilla terms, ESs and applications (which form $\tovsc$ redexes) are only introduced by the translation of subtractions, and they receive as left sub-terms the application of a free variable to a term: thus the introduced application is not a $\todb$-redex, nor it is a value  (nor can reduce to one) thus the introduced ES is not a $\tovs$-redex. For $\tosc$, the translation of subtraction instead introduces $\tos$ redexes.\qed
\end{proof}

In the other direction, $\tovsc$-normal natural terms are not mapped by $\trndtosc\cdot$ to cut-free vanilla terms, because the translation of applications adds cuts. It is possible to optimize the translation as to avoid this mismatch, but the optimized translations would considerably complicate the translation of contexts (introduced and used in the next sections), which is why we refrain from it. At the same time, we shall show that the translation of normal natural terms are \emph{almost} cut-free, in a suitable harmless way.

\section{Defining Cut Elimination for the Vanilla $\l$-Calculus}
\label{sect:small-step-cut-el}
In this section, we define cut elimination on vanilla $\l$-terms, obtaining the \emph{vanilla $\l$-calculus}. Our cut elimination shall be based on meta-level substitution. The notion of value shall emerge naturally, and rules at a distance shall play a role.

\paragraph{Left Variable Occurrences.} At first sight, defining small-step cut elimination for vanilla terms is conceptually easy, since one would simply define it as follows:
\[\cuta{\stmtwo}\svar \stm \ \ \  \Rew{\symfont{cut}}\ \ \   \cutsub{\stmtwo}\svar \stm\]
for a notion of meta-level substitution $\cutsub{\stmtwo}\svar \stm$ for vanilla $\l$-terms.
The problem however is that the definition of $\cutsub{\stmtwo}\svar \stm$ cannot be the same as for natural terms. 

In natural deduction, the free occurrences of a variable are introduced by axioms, and axioms can always be replaced by proofs having the same ending sequent. Since we consider axiom as a right rule, let us say that variables have only \emph{right occurrences}, and that on right occurrences one can always define meta-level substitution $\subi\tmthree\var\tm$ simply as the \emph{replacement} of $\var$ by $\tmthree$ in $\tm$. 

On vanilla terms, subtractions $\suba\svartwo\stmtwo\svar\stm$ introduce \emph{left occurrences} of $\vartwo$. In contrast to right occurrences, the left ones cannot be simply replaced by a vanilla term $\stmthree$, because $\suba\stmthree\stmtwo\svar\stm$ does \emph{not} belong to the grammar of vanilla terms. Therefore, the definition of meta-level substitution $\cutsub\stmthree\svar\stm$ on vanilla terms is a bit tricky on subtractions.

\paragraph{Easy Cases and Values.} To define $\cutsub\stmthree\svar\stm$, we have to inspect the shape of $\stmthree$. The simplest case is when $\stmthree$ is a variable $\svartwo$, since in that case it is possible to simply replace $\svar$ by $\svartwo$. A second clear case is when $\stmthree$ is an abstraction $\la\svartwo\stmtwo$. Then, we can set:
\begin{center}
$\begin{array}{l\colspace l\colspace lll}
\cutsub{\la\svartwo\stmtwo}\svar \suba\svar\stmfour\svarthree\stm & \defeq &  \cuta{\cuta{\cutsub{\la\svartwo\stmtwo}\svar\stmfour}\svartwo \stmtwo}\svarthree \cutsub{\la\svartwo\stmtwo}\svar\stm
\end{array}$
\end{center}
This complex definition amounts to do the expected elimination of the principal cut between the abstraction and the subtraction (roughly, corresponding to the $\todb$ step $(\la\svartwo\stmtwo)\stmfour \todb \sube\stmfour\svartwo \stmtwo$ of the (V)SC) \emph{plus} putting the resulting term in a second created cut on $\svarthree$ \emph{and} propagating the substitution $\cutsub{\la\svartwo\stmtwo}\svar$ to the sub-terms $\stmfour$ and $\stm$ of the subtraction. Note that the two clear cases $\stmthree=\svartwo$ and $\stmthree = \la\svartwo\stmtwo$ are exactly those for \emph{values}, which correspond to the right rules of $\sproves$. The full definition of meta-level substitution $\cutsub\sval\svar\stm$ for values is in \reffig{SC-dynamic-small}.
The following property is proved by induction on $\stm$.
\begin{lemma}
Let $\svar\notin\fv\stm$. Then $\cutsub\sval\svar\stm=\stm$.\label{l:vacuous-subs}
\end{lemma}

\begin{figure}[t!]
\centering
 \begin{tabular}{c}
$\begin{array}{r\colspace rllllll}
\textsc{Terms} & \scterms \ni \stm,\stmtwo,\stmthree & \grameq &  \svar \in\scvars \mid \la\svar\stm \mid
\cuta\stmtwo\svar\stm \mid \nsuba\svartwo\stmtwo\svar \stm
\\
\textsc{Values} & \sval & \grameq &  \svar \mid \la\svar\stm
\\[4pt]
\textsc{Contexts} & \gsctx & \grameq & \ctxhole  \mid \la\svar\gsctx \mid \cuta\stm\svar\gsctx \mid \cuta\gsctx\svar\stm \mid \nsuba\svartwo\stm\svar\gsctx \mid \nsuba\svartwo\gsctx\svar\stm
\\
\textsc{Left ctxs} &\lsctx & \grameq & \ctxhole \mid \cuta\stm\svar\lsctx \mid \nsuba\svartwo\stm\svar\lsctx

\end{array}$
\\[3pt]\hline
$\begin{array}{r@{\hspace{.15cm}} l@{\hspace{.15cm}} l @{\hspace{.15cm}}|@{\hspace{.15cm}} r@{\hspace{.15cm}} l@{\hspace{.15cm}} l}
\multicolumn{6}{c}{\textsc{Meta-level substitution of values for vanilla terms}}
\\[3pt]

\cutsub\sval\svar \svar & \defeq & \sval
&
\cutsub\sval\svar \cuta\stmtwo\svartwo \stm & \defeq &  \cuta{\cutsub\sval\svar\stmtwo}\svartwo \cutsub\sval\svar\stm
\\
 
 \cutsub\sval\svar \svartwo & \defeq & \svartwo
&
\cutsub\sval\svar \suba\svartwo\stmtwo\svarthree\stm & \defeq & \suba\svartwo{\cutsub\sval\svar\stmtwo}\svarthree \cutsub\sval\svar\stm
\\

\cutsub\sval\svar \la\svartwo\stm & \defeq & \la\svartwo\cutsub\sval\svar \stm
&
\cutsub\svartwo\svar \suba\svar\stmtwo\svarthree\stm & \defeq & \suba\svartwo{\cutsub\svartwo\svar\stmtwo}\svarthree \cutsub\svartwo\svar\stm
\\
\multicolumn{6}{c}{\cutsub{\la\svartwo\stmthree}\svar \suba\svar\stmtwo\svarthree\stm \ \ \defeq \ \  \cuta{\cuta{\cutsub{\la\svartwo\stmthree}\svar\stmtwo}\svartwo \stmthree}\svarthree \cutsub{\la\svartwo\stmthree}\svar\stm}

\end{array}$

\\[3pt]\hline
\begin{tabular}{c\colspace |\colspace c}
$\begin{array}{r\colspace l\colspace ll}
\multicolumn{3}{c}{\textsc{Root  cut elimination}}
\\[3pt]
\cuta{\lsctxp\sval}\svar \stm & \srtocut & \lsctxp{\cutsub\sval\svar \stm}
\end{array}$
&
$\begin{array}{c\colspace c}
\textsc{Ctx closure}
&
					\AxiomC{$\stm \srtocut \stm'$}
					\UnaryInfC{$\gsctxp\stm \stocut \gsctxp{\stm'}$}
					\DisplayProof
\end{array}$
\end{tabular}
\end{tabular}
\caption{The vanilla $\l$-calculus.}
\label{fig:SC-dynamic-small}
\end{figure}
\paragraph{The Vanilla $\l$-Calculus.} For defining small-step cut elimination for all terms, we rely on the crucial observation that \emph{every} vanilla term $\stm$ \emph{splits uniquely} as $\stm = \lsctxp\sval$, where $\lsctx$ is a left context---defined in \reffig{SC-dynamic-small}---and $\sval$ is a value. Note that a similar splitting is also used in the value substitution rule $\tovs$ of the VSC, but there is a key difference: in the VSC applications cannot be split as $\lctxp\val$ (because application is not a left rule), which is why some ESs cannot be eliminated, while \emph{all} vanilla terms can be split (because application is replaced by subtraction, which is a left rule).

We know how to substitute $\sval$, and we are left to define what to do with $\lsctx$. There are two options, namely mimicking the substitution rules by name (namely $\tos$) and by value ($\tovs$) for natural terms: either we carry $\lsctx$ along with $\sval$ in the propagation of the substitution, possibly duplicating or erasing it, as in the \cbn rule $\tos$ of the SC, or we commute it out, as in the \cbv rule $\tovs$ of the VSC. While in principle both are valid choices, the latter seems more natural because---even if we adopt the \cbn approach---cut-free vanilla terms map to \cbv normal natural terms, as shown by \refprop{cbv-nfs}, and \emph{not} to the \cbn ones. Essentially, to catch a \cbn semantics we should also modify something about cut-free proofs, but that goes against the starting point of our work, namely that vanilla cut-free proofs are a good representation of \cbv normal forms. 

The cut elimination rule $\stocut$ is then defined in \reffig{SC-dynamic-small} by commuting out the left context $\lsctx$, completing the definition of the vanilla $\l$-calculus. 

\paragraph{Subject Reduction.} The standard sanity check for the definition of cut elimination $\stocut$ is subject reduction.

\begin{toappendix}
\begin{proposition}[Subject reduction]
\NoteProof{prop:subject-reduction}
Let $\stm\in\scterms$ and  $\tderiv \exder \multiForm \sproves \stm\hastype\form$ be a derivation. If $\stm \stocut \stmtwo$ then there exists a derivation $\tderivtwo \exder \multiForm \sproves \stmtwo\hastype\form$.\label{prop:subject-reduction}
\end{proposition}
\end{toappendix}

\begin{proof}
Let $\stm =\gsctxp{\cuta{\lsctxp\sval}\svar\stmthree} \stocut \gsctxp{\lsctxp{\cutsub\sval\svar\stmthree}} = \stmtwo$. The proof is by induction on $\stmthree$, $\lsctx$, and $\gsctx$. Details in the Appendix\withoutproofs{ of the technical report \cite{accattoli2024vanillasequentcalculuscallbyvalue}}.\qed
\end{proof}


\section{Simulating the Vanilla $\l$-calculus in the VSC}
\label{sect:simulation-vanilla-by-vsc}
To simulate the vanilla $\l$-calculus in the VSC, we need to look at how the $\trndtosc\cdot$ translation relates the respective notions of meta-level substitution. On vanilla terms, substitution on left variable occurrences does also the work that is done by rule $\todb$ on the VSC. Unsurprisingly, then, the translation commutes with substitution only up to $\todb$.  

\begin{toappendix}
\begin{lemma}[Substitution and $\trsctond\cdot$ commute up to $\todb$]
\NoteProof{l:sc-to-nd-subst}$\subi{\trsctond\sval}\var \trsctond\stm \todb^* \trsctond{\cutsub\sval\svar \stm}$.\label{l:sc-to-nd-subst}
\end{lemma}
\end{toappendix}

To establish the simulation, we need a lemma about the commutation of contexts and of the translation. \emph{Technicality}: the translations are extended to contexts by setting $\trsctond\ctxhole\defeq \ctxhole$ and $\trndtosc\ctxhole\defeq \ctxhole$, and defining them as for terms on the other cases.

\begin{toappendix}
\begin{lemma}[Contexts and $\trsctond\cdot$ translation]
\NoteProof{l:transl-sc-to-nd-ctxs}\label{l:transl-sc-to-nd-ctxs}
\hfill
\begin{enumerate}
	\item $\trsctond\lsctx$ is a $\lctx$ context and $\trsctond{\lsctxp\stm} = \trsctond\lsctx\ctxholep{\trsctond\stm}$.
	\item $\trsctond\gsctx$ is a  $\gctx$ context and $\trsctond{\gsctxp\stm} = \trsctond\gsctx\ctxholep{\trsctond\stm}$.
\end{enumerate}
\end{lemma}
\end{toappendix}

\begin{proposition}[VSC simulates vanilla]
If $\stm \stocut \stmtwo$ then  $\trsctond\stm \tovs\todb^* \trsctond\stmtwo$.\label{prop:sc-to-nd-simulation}
\end{proposition}

\begin{proof}
For a root step $\stm= \cuta{\lsctxp\sval}\svar \stmthree  \srtocut  \lsctxp{\cutsub\sval\svar \stmthree}= \stmtwo$, we have:
\begin{center}
$\begin{array}{ccllllllllll}
\trsctond{\cuta{\lsctxp\sval}\svar \stmthree} &= &\sube{\trsctond{\lsctxp\sval}}\var \trsctond\stmthree
 &=_{\reflemmaeq{transl-sc-to-nd-ctxs}.1}& 
\sube{\trsctond\lsctx\ctxholep{\trsctond\sval}}\var \trsctond\stmthree
 &\rtovs& 
\trsctond\lsctx\ctxholep{\subi{\trsctond\sval}\var \trsctond\stmthree}
\\[3pt] (\reflemmaeq{sc-to-nd-subst}) &\todb^*&   
\trsctond\lsctx\ctxholep{\trsctond{\cutsub\sval\svar\stmthree}}
 &=_{\reflemmaeq{transl-sc-to-nd-ctxs}.1}& 
\trsctond{\lsctxp{\cutsub\sval\svar\stmthree}}.
\end{array}$
\end{center}
For steps in contexts, the simulation follows from the root case and \reflemma{transl-sc-to-nd-ctxs}.2.\qed
\end{proof}

Putting together the simulation of cut elimination with the translation of cut-free terms, we obtain the following property.
\begin{lemma}[Preservation of VSC termination]
Let $\stm \stocut^* \stmtwo$ with $\stmtwo$ cut-free. Then $\trsctond\stm \tovsc^* \trsctond\stmtwo$ with $\trsctond\stmtwo$ is $\tovsc$-normal.
\end{lemma}
\begin{proof}
By \refprop{sc-to-nd-simulation}, we obtain $\trsctond\stm \tovsc^* \trsctond\stmtwo$. By \refprop{cbv-nfs}, $\trsctond\stmtwo$ is a $\tovsc$-normal form. \qed
\end{proof}
\section{Simulating the VSC in the Vanilla $\l$-Calculus} 
\label{sect:simulating-vsc-vanilla}
In the other direction, the translation and meta-level substitution commute neatly.
\begin{toappendix}
\begin{lemma}[Substitution and $\trndtosc\cdot$ commute]
\NoteProof{l:nd-to-sc-subst}
$\cutsub{\trndtosc{\val}}\svar\trndtosc{\tm} = \trndtosc{\subi\val\var\tm}$.\label{l:nd-to-sc-subst}
\end{lemma}
\end{toappendix}

As before, for the simulation we need a lemma about contexts.
\begin{toappendix}
\begin{lemma}[Contexts and $\trndtosc\cdot$ translation]
\NoteProof{l:transl-nd-to-sq-ctxs}\label{l:transl-nd-to-sq-ctxs}
\hfill
\begin{enumerate}
	\item $\trndtosc\lctx$ is an $\lsctx$ context and $\trndtosc{\lctxp\tm} = \trndtosc\lctx\ctxholep{\trndtosc\tm}$.
	\item $\trndtosc\ctx$ is a $\gsctx$ context and $\trndtosc{\ctxp\tm} = \trndtosc\ctx\ctxholep{\trndtosc\tm}$.
\end{enumerate}
\end{lemma}
\end{toappendix}

\begin{proposition}[Vanilla simulates VSC]
\NoteProof{prop:nd-simulation-by-sc}\label{prop:nd-simulation-by-sc}\hfill
\begin{enumerate}
\item If $\tm \todb \tmtwo$ then $\trndtosc\tm \stocut\stocut \trndtosc\tmtwo$.
\item If $\tm \tovs \tmtwo$ then $\trndtosc\tm \stocut \trndtosc\tmtwo$.
\end{enumerate}
\end{proposition}
\begin{proof}
For root steps:\applabel{prop:nd-simulation-by-sc}
\begin{enumerate}
\item \emph{$\beta$ at a distance}, \ie $\tm  = \sctxp{\la\var\tmthree}\tmfour \rtodb \sctxp{\sube\tmfour\var\tmthree} = \tmtwo$. With $\svartwo$ and $\svarthree$ fresh:
\begin{center}
$\begin{array}{clllllllllll}
\trndtosc{\sctxp{\la\var\tmthree}\tmfour} 
&= &
\cuta{\trndtosc{\sctxp{\la\var\tmthree}}}\svartwo \nsuba\svartwo {\trndtosc\tmfour}\svarthree\svarthree
\\[2.5pt] & =_{\reflemmaeq{transl-sc-to-nd-ctxs}.1} & 
\cuta{\trndtosc\sctx\ctxholep{\la\svar\trndtosc\tmthree}}\svartwo \nsuba\svartwo {\trndtosc\tmfour}\svarthree\svarthree
\\[2.5pt] & \stocut & 
\multicolumn{3}{l}{\trndtosc\sctx\ctxholep{  \cuta{\cuta{\cutsub{\la\svar\trndtosc\tmthree}\svartwo \trndtosc\tmfour}\svar\trndtosc\tmthree}\svarthree \cutsub{\la\svar\trndtosc\tmthree}\svartwo\svarthree }}
\\[2.5pt] & =_{\reflemmaeq{vacuous-subs}} & 
\trndtosc\sctx\ctxholep{  \cuta{\cuta{\trndtosc\tmfour}\svar\trndtosc\tmthree}\svarthree \svarthree }
\\[2.5pt] (\star) & \stocut & 
\trndtosc\sctx\ctxholep{ \cuta{\trndtosc\tmfour}\svar \trndtosc\tmthree }
\\[2.5pt] & = & 
\trndtosc\sctx\ctxholep{ \trndtosc{\sube\tmfour\svar \tmthree} }
& =_{\reflemmaeq{transl-sc-to-nd-ctxs}.1} & 
\trndtosc{\sctxp{\sube\tmfour\var\tmthree}}
\end{array}$
\end{center}
Where step $(\star)$ is given by the general fact that $\cuta\stm\svar\svar \stocut \stm$, since $\stm = \lsctxp\sval$ for some $\lsctx$ and $\sval$, thus $\cuta\stm\svar\svar = \cuta{\lsctxp\sval}\svar\svar \stocut \lsctxp\sval = \stm$.
\item \emph{Value substitution}, \ie $\tm  = \sube{\sctxp\val}\var\tmthree \rtodb \sctxp{\subi\val\var\tmthree} = \tmtwo$. Then:
\begin{center}
$\begin{array}{clllllllllll}
\trndtosc{ \sube{\sctxp\val}\var\tmthree } 
&= &
\cuta{\trndtosc{\sctxp\val}}\svar \trndtosc\tmthree
\\[2.5pt] & =_{\reflemmaeq{transl-sc-to-nd-ctxs}.1} & 
\cuta{\trndtosc\sctx\ctxholep{\trndtosc\val}}\svar \trndtosc\tmthree
\\[2.5pt] & \stocut & 
\trndtosc\sctx\ctxholep{\cutsub{\trndtosc\val}\svar \trndtosc\tmthree}
\\[2.5pt]  & =_{\reflemmaeq{nd-to-sc-subst}} & 
\trndtosc\sctx\ctxholep{\trndtosc{\subi\val\var\tmthree}}
& =_{\reflemmaeq{transl-sc-to-nd-ctxs}.1} & 
\trndtosc{ \sctxp{\subi\val\var\tmthree} }
\end{array}$
\end{center}
\end{enumerate}
For steps in contexts, the simulation follows from the root case and \reflemma{transl-nd-to-sq-ctxs}.2.\qed
\end{proof}

\paragraph{Translation of Normal Natural Terms.} As anticipated in \refsect{small-step-cut-el}, the translation of normal VSC terms does not give cut-free vanilla terms, because the translation of applications introduces cuts. Here we show that nonetheless the obtained terms are \emph{almost} cut-free, since they are cut-free up to some trivial cut elimination steps, dubbed \emph{renaming steps}. In particular, they strictly reduce the size of a vanilla term. 

\begin{definition}[Renaming Cut Elimination Steps]
A step $\stm = \gsctxp{\cuta{\lsctxp\sval}\svar\stmthree} \stocut \gsctxp{\lsctxp{\cutsub\sval\svar\stmthree}}=\stmtwo$ is a renaming step, noted $\stm \storencut \stmtwo$, if $\sval$ is a variable.
\end{definition}

\begin{toappendix}
\begin{proposition}[The $\trndtosc\cdot$ translation of VSC normal forms is almost cut-free]
\NoteProof{prop:trans-normal-nd-terms}
Let $\tm\in\ndesterms$ be $\tovsc$-normal. Then there exists $\stm$ cut-free such that $\trndtosc\tm \storencut^k \stm$ with $k\leq\size\tm$, where $\size\tm$ is the size (\ie number of constructors) in $\tm$. \label{prop:trans-normal-nd-terms}
\end{proposition}
\end{toappendix}

As for the other direction, we can put together the simulation and the translation of normal forms, obtaining the preservation of termination by $\trndtosc\cdot$.
\begin{lemma}[Preservation of termination]
Let $\stm \tovsc^* \tmtwo$ with $\tmtwo$  $\tovsc$-normal. Then there exists a cut-free term $\stmtwo$ such that $\trndtosc\tm \stocut^* \stmtwo$.
\end{lemma}
\begin{proof}
By \refprop{nd-simulation-by-sc}, we obtain $\trndtosc\tm \stocut^* \trndtosc\tmtwo$. By \refprop{trans-normal-nd-terms}, there exists $\stmtwo$ cut-free such that $\trndtosc\tmtwo \stocut^* \stmtwo$. \qed
\end{proof}

\section{Strong Normalization}
The typical theorem for Curry-Howard correspondences is that the logical system ensures strong normalization of the typed terms. We thus provide such a result for our new correspondence between the vanilla $\l$-calculus and the vanilla sequent calculus. 

In this section, we give only the statement, since the proof is too technical for a fresh perspective paper. The proof is developed in \withproofs{Appendix \ref{app:SN-preliminaries} and Appendix \ref{sect:SN}}\withoutproofs{sections D and E of the Appendix of the technical report \cite{accattoli2024vanillasequentcalculuscallbyvalue}}, and is based on the bi-orthogonal reducibility  method, adapting and slightly simplifying its presentation in Accattoli \cite{DBLP:journals/lmcs/Accattoli23}.

\begin{definition}[Strong normalization]
A vanilla term $\stm$ is strongly normalizing, also noted $\stm \in \sn\cutsym$, if $\stm \stocut \stmtwo$ implies $\stmtwo\in\sn\cutsym$.
\end{definition}

\begin{theorem}[Typable terms are SN]
Let $\stm\in\scterms$ and  $\multiForm \sproves \stm\hastype\form$. Then $\stm\in\sn\cutsym$.
\end{theorem}

\section{Conclusions} 
We introduced the \emph{vanilla} $\l$-calculus, a computational interpretation of the simplest sequent calculus, and showed that it simulates and it is simulated by call-by-value evaluation. Technically, the simulations are clean and compact, thanks to the use of rewriting rules at a distance for both the new cut elimination rule for the vanilla $\l$-calculus and the presentation of call-by-value evaluation that we adopt from the literature.

Our study nicely complements the two famous cornerstones by Curry and Howard about minimal intuitionistic logic resting only on basic logical concepts: 
\begin{enumerate}
\item \emph{Curry 1958} \cite{curry1958combinatory}: Hilbert's system matches combinatory logic;
\item \emph{Howard 1969} (but published only in 1980) \cite{Howard1980-HOWTFN-2}: Gentzen's natural deduction matches the (call-by-name) $\l$-calculus;
\item \emph{Here}: Gentzen's sequent calculus matches call-by-value evaluation.
\end{enumerate}

We believe that our work provides a fresh perspective over the sequent calculus. Its modern computational interpretation is usually defined starting from sequent calculi with stoups for classical logic, following Curien and Herbelin \cite{DBLP:conf/icfp/CurienH00}. The basic, \emph{vanilla} presentation of the intuitionistic case seems to have fallen into a blind spot of the literature. This work shows that it is far from being unworthy of attention.

\paragraph{Future Work.} Maraist et al. \cite{DBLP:journals/tcs/MaraistOTW99} propose a Curry-Howard correspondence for \emph{call-by-need} using an affine logic with a duplication modality, tweaking the linear logic one for call-by-value. In ongoing work, we are using the results of this paper as the starting point for a Curry-Howard for call-by-need that is \emph{not} based on linear/affine concepts. 


%
%
%
\newpage
\bibliographystyle{splncs04}
\bibliography{main.bbl}
\withproofs{
\newpage
\appendix
\setboolean{appendix}{true}

\section{Proofs Removed from \refsect{small-step-cut-el} (Defining Cut Elimination for the Vanilla $\l$-Calculus)}
\label{sect:app-small-step-cut-el}

\begin{lemma}[Weakening]
\label{l:sc-weakening}
Let $\tderiv \exder \multiForm \sproves \stm\hastype\form$ be a derivation and $\svar\notin \fv\stm \cup \dom\multiForm$. Then, for any formula $\formtwo$, there exists a derivation $\tderivtwo_\formtwo \exder \multiForm, \svar\hastype\formtwo \sproves \stm\hastype\form$.
\end{lemma}
\begin{proof}
Straightforward induction on $\tderiv$.\qed
\end{proof}

\gettoappendix{prop:subject-reduction}
\begin{proof}
\applabel{prop:subject-reduction}
Let $\stm =\gsctxp{\cuta{\lsctxp\sval}\svar\stmthree} \stocut \gsctxp{\lsctxp{\cutsub\sval\svar\stmthree}} = \stmtwo$. The statement is decomposed in three more basic ones, the second and the third one using the preceding statement as base case for the induction:
\begin{enumerate}
\item Let $\tderiv \exder \multiForm \sproves \cuta\sval\svar\stmthree \hastype\form$. Then there exists $\tderivtwo \exder \multiForm \sproves \cutsub\sval\svar\stmthree \hastype\form$.
\item Let $\tderiv \exder \multiForm \sproves \cuta{\lsctxp\sval}\svar\stmthree \hastype\form$. Then there exists $\tderivtwo \exder\multiForm \sproves \lsctxp{\cutsub\sval\svar\stmthree} \hastype\form$.
\item Let $\tderiv \exder \multiForm \sproves \gsctxp{\cuta{\lsctxp\sval}\svar\stmthree} \hastype\form$. Then there exists $\tderivtwo \exder\multiForm \sproves \gsctxp{\lsctxp{\cutsub\sval\svar\stmthree}} \hastype\form$.
\end{enumerate}

Proofs:
\begin{enumerate}
\item By induction on $\stmthree$. There are seven cases, corresponding to the cases defining $\cutsub\sval\svar\stmthree$, and divided into principal and commutative cases:
\begin{itemize}
\item \emph{Principal axiom case 1}, \ie $\stmthree= \svar$ and $\cutsub\sval\svar\stmthree = \sval$. The derivation in the hypotheses is:
\[
	\AxiomC{$ \tderiv_\sval \exder \multiForm \sproves \sval\hastype\form$}
		\AxiomC{}
		\RightLabel{$ \ax $}
		\UnaryInfC{$  \multiForm,  \svar\hastype\form \sproves \svar\hastype\form$}
\RightLabel{$ \cut $}
\BinaryInfC{$  \multiForm \sproves \cuta\sval\svar\svar \hastype\form$}
	\DisplayProof
\]
The type derivation for $\sval$ then simply is given by the sub-proof $\tderiv_\sval \exder \multiForm \sproves \sval\hastype\form$.

\item \emph{Principal axiom case 2}, \ie $\stmthree= \svartwo$ and $\cutsub\sval\svar\stmthree = \svartwo$. The derivation in the hypotheses is:
\[
	\AxiomC{$\tderiv_\sval \exder\multiForm \sproves \sval\hastype\formtwo$}
		\AxiomC{}
		\RightLabel{$ \ax $}
		\UnaryInfC{$\multiForm,  \svar\hastype\formtwo,\svartwo\hastype\form \sproves \svartwo\hastype\form$}
\RightLabel{$ \cut $}
\BinaryInfC{$\multiForm,\svartwo\hastype\form \sproves \cuta\sval\svar\svartwo \hastype\form$}
	\DisplayProof
\]
The type derivation for $\svartwo$ then simply is given by the following axiom:
\[
		\AxiomC{}
		\RightLabel{$ \ax $}
		\UnaryInfC{$\multiForm,\svartwo\hastype\form \sproves \svartwo \hastype\form$}
	\DisplayProof
\]

\item \emph{Principal subtraction case 1}, \ie $\cutsub\sval\svar\stmthree$ is: 
\[\begin{array}{l\colspace l\colspace l}
\cutsub\svartwo\svar\nsuba\svar\stmfour\svarthree\stmfive  
& = &
\nsuba\svartwo{\cutsub\svartwo\svar\stmfour}\svarthree \cutsub\svartwo\svar\stmfive.
\end{array}\]
The derivation $\tderiv$ in the hypothesis then is:
\[
	\AxiomC{}
	\RightLabel{$ \ax $}
	\UnaryInfC{$\multiForm \sproves \svartwo\hastype\formtwo\imply\formthree$}
		\AxiomC{$\tderiv_\stmfour \exder \multiFormthree  \sproves \stmfour \hastype\formtwo$}
		\AxiomC{$\tderiv_\stmfive \exder \multiFormthree,  \svarthree\hastype\formthree \sproves \stmfive \hastype\form$}
		\RightLabel{$ \implyLeftRule $}
		\BinaryInfC{$\multiForm, \svar\hastype\formtwo\imply\formthree \sproves \nsuba\svar\stmfour\svarthree\stmfive  \hastype \form$}
\RightLabel{$ \cut $}
\BinaryInfC{$\multiForm \sproves \cuta\svartwo\svar\nsuba\svar\stmfour\svarthree\stmfive   \hastype \form$}
	\DisplayProof
\]
With $\multiForm = \multiFormtwo,\svartwo\hastype\formtwo\imply\formthree$ and either $\multiFormthree=\multiForm$ or $\multiFormthree=\multiForm, \svar\hastype\formtwo\imply\formthree$. We treat the case $\multiFormthree=\multiForm, \svar\hastype\formtwo\imply\formthree$, the other one is simpler. Then, we actually have $\multiFormthree = \multiFormtwo,\svar\hastype\formtwo\imply\formthree, \svartwo\hastype\formtwo\imply\formthree$.

In order to apply the \ih, note that our hypotheses allow us to build the following two derivations:
\[
	\AxiomC{}
	\RightLabel{$ \ax $}
	\UnaryInfC{$\multiFormtwo,\svartwo\hastype\formtwo\imply\formthree \sproves \svartwo\hastype\formtwo\imply\formthree$}
		\AxiomC{$\tderiv_\stmfour \exder \multiFormtwo,\svar\hastype\formtwo\imply\formthree, \svartwo\hastype\formtwo\imply\formthree  \sproves \stmfour \hastype\formtwo$}
\RightLabel{$ \cut $}
\BinaryInfC{$\multiFormtwo,\svartwo\hastype\formtwo\imply\formthree \sproves \cuta\svartwo\svar\stmfour   \hastype \formtwo$}
	\DisplayProof
\]
and
\[\small
	\AxiomC{}
	\RightLabel{$ \ax $}
	\UnaryInfC{$\multiFormtwo,\svartwo\hastype\formtwo\imply\formthree,\svarthree\hastype\formthree \sproves \svartwo\hastype\formtwo\imply\formthree$}
		\AxiomC{$\tderiv_\stmfive \exder \multiFormtwo,\svar\hastype\formtwo\imply\formthree, \svartwo\hastype\formtwo\imply\formthree,  \svarthree\hastype\formthree \sproves \stmfive \hastype\form$}
\RightLabel{$ \cut $}
\BinaryInfC{$\multiFormtwo,\svartwo\hastype\formtwo\imply\formthree,\svarthree\hastype\formthree \sproves \cuta\svartwo\svar\stmfive   \hastype \form$}
	\DisplayProof
\]
By \ih, we obtain two derivations:
\begin{enumerate}
\item $\tderivtwo_\stmfour \exder \multiFormtwo,\svartwo\hastype\formtwo\imply\formthree \sproves \cutsub\svartwo\svar\stmfour   \hastype \formtwo$, and

\item $\tderivtwo_\stmfive \exder \multiFormtwo,\svartwo\hastype\formtwo\imply\formthree,\svarthree\hastype\formthree \sproves \cutsub\svartwo\svar\stmfive   \hastype \form$.
\end{enumerate}
Then, the derivation $\tderivtwo$ proving the statement is built as follows:
\[
\AxiomC{$\tderivtwo_\stmfour \exder \multiFormtwo,\svartwo\hastype\formtwo\imply\formthree \sproves \cutsub\svartwo\svar\stmfour   \hastype \formtwo$}
		\AxiomC{$\tderivtwo_\stmfive \exder \multiFormtwo,\svartwo\hastype\formtwo\imply\formthree,\svarthree\hastype\formthree \sproves \cutsub\svartwo\svar\stmfive   \hastype \form$}
		\BinaryInfC{$\multiFormtwo,\svartwo\hastype\formtwo\imply\formthree \sproves \nsuba\svartwo{\cutsub\svartwo\svar\stmfour}\svarthree{\cutsub\svartwo\svar\stmfive}  \hastype \form$}
			\DisplayProof
\]

\item \emph{Principal subtraction case 2}: \ie $\cutsub\sval\svar\stmthree$ is: 
\[\begin{array}{l\colspace l\colspace l}
\cutsub{\la\svartwo\stmfour}\svar \suba\svar\stmfive\svarthree\stmsix
& = &
 \cuta{\cuta{\cutsub{\la\svartwo\stmfour}\svar\stmfive}\svartwo\stmfour}\svarthree \cutsub{\la\svartwo\stmfour}\svar\stmsix.
\end{array}\]
The derivation $\tderiv$ in the hypothesis then is:
\[
	\AxiomC{$\tderiv_\stmfour \exder \multiForm, \svartwo\hastype\formtwo \sproves \stmfour\hastype\formthree$}
	\RightLabel{$ \implyRightRule $}
	\UnaryInfC{$\multiForm \sproves \la\svartwo\stmfour \hastype\formtwo\imply\formthree$}
		\AxiomC{$\tderiv_\stmfive \exder \multiFormthree  \sproves \stmfive \hastype\formtwo$}
		\AxiomC{$\tderiv_\stmsix \exder \multiFormthree,  \svarthree\hastype\formthree \sproves \stmsix \hastype\form$}
		\RightLabel{$ \implyLeftRule $}
		\BinaryInfC{$\multiForm, \svar\hastype\formtwo\imply\formthree \sproves \nsuba\svar\stmfive\svarthree\stmsix  \hastype \form$}
\RightLabel{$ \cut $}
\BinaryInfC{$\multiForm \sproves \cuta{\la\svartwo\stmfour}\svar\nsuba\svar\stmfive\svarthree\stmsix   \hastype \form$}
	\DisplayProof
\]
With either $\multiFormthree=\multiForm$ or $\multiFormthree=\multiForm, \svar\hastype\formtwo\imply\formthree$. We treat the case $\multiFormthree=\multiForm, \svar\hastype\formtwo\imply\formthree$, the other one is simpler.

In order to apply the \ih, note that our hypotheses allow us to build the following two derivations:
\[
	\AxiomC{$\multiForm \sproves \la\svartwo\stmfour \hastype\formtwo\imply\formthree$}
		\AxiomC{$\tderiv_\stmfive \exder \multiForm, \svar\hastype\formtwo\imply\formthree  \sproves \stmfive \hastype\formtwo$}
\RightLabel{$ \cut $}
\BinaryInfC{$\multiForm \sproves \cuta{\la\svartwo\stmfour}\svar\stmfive   \hastype \formtwo$}
	\DisplayProof
\]
and
\[\small
	\AxiomC{$\tderiv_\stmfour' \exder \multiForm, \svartwo\hastype\formtwo,  \svarthree\hastype\formthree \sproves \stmfour\hastype\formthree$}
	\RightLabel{$ \implyRightRule $}
	\UnaryInfC{$\multiForm,  \svarthree\hastype\formthree \sproves \la\svartwo\stmfour \hastype\formtwo\imply\formthree$}
		\AxiomC{$\tderiv_\stmsix \exder \multiForm, \svar\hastype\formtwo\imply\formthree,  \svarthree\hastype\formthree \sproves \stmsix \hastype\form$}
\RightLabel{$ \cut $}
\BinaryInfC{$\multiForm,\svarthree\hastype\formthree \sproves \cuta{\la\svartwo\stmfour}\svar\stmsix   \hastype \form$}
	\DisplayProof
\]
Where the sub-derivation $\tderiv_\stmfour'$ is obtained by applying the weakening lemma (\reflemma{sc-weakening}) to $\tderiv_\stmfour$.

By \ih, we obtain two derivations:
\begin{enumerate}
\item $\tderivtwo_\stmfive \exder \multiForm \sproves \cutsub{\la\svartwo\stmfour}\svar\stmfive   \hastype \formtwo$, and

\item $\tderivtwo_\stmsix \exder \multiForm,\svarthree\hastype\formthree \sproves \cutsub{\la\svartwo\stmfour}\svar\stmsix   \hastype \form$.
\end{enumerate}

Then, the derivation $\tderivtwo$ proving the statement is built as follows:
\[ \footnotesize
\AxiomC{$\tderivtwo_\stmfive \exder \multiForm \sproves \cutsub{\la\svartwo\stmfour}\svar\stmfive   \hastype \formtwo$}
	\AxiomC{$\tderiv_\stmfour \exder \multiForm, \svartwo\hastype\formtwo \sproves \stmfour\hastype\formthree$}
\RightLabel{$ \cut $}
\BinaryInfC{$\multiForm \sproves \cutsub{\la\svartwo\stmfour}\svar\stmfive   \hastype \formthree$}
		\AxiomC{$\tderivtwo_\stmsix \exder \multiForm,\svarthree\hastype\formthree \sproves \cutsub{\la\svartwo\stmfour}\svar\stmsix   \hastype \form$}
\RightLabel{$ \cut $}
\BinaryInfC{$\multiForm \sproves \cuta{\cuta{\cutsub{\la\svartwo\stmfour}\svar\stmfive}\svartwo\stmfour}\svarthree \cutsub{\la\svartwo\stmfour}\svar\stmsix  \hastype \form$}
	\DisplayProof
\]
\item \emph{Commutative cases}: the following three cases (they correspond to so-called \emph{commutative cut elimination cases}) in the definition of $\cutsub\sval\svar\stmthree$ follow immediately from the \ih and the weakening lemma (\reflemma{sc-weakening}):
\begin{enumerate}
\item $\cutsub\sval\svar \la\svartwo\stmfour  \defeq  \la\svartwo\cutsub\sval\svar \stmfour$;
\item $\cutsub\sval\svar \cuta\stmfour\svartwo \stmfive  \defeq   \cuta{\cutsub\sval\svar\stmfour}\svartwo \cutsub\sval\svar\stmfive$;
\item $\cutsub\sval\svar \suba\svartwo\stmfour\svarthree\stmfive  \defeq  \suba\svartwo{\cutsub\sval\svar\stmfour}\svarthree \cutsub\sval\svar\stmfive$.
\end{enumerate}
We give the details of the first one, to showcase the reasoning. The derivation $\tderiv$ in the hypotheses is (with $\form = \formtwo \imply \formthree$):
\[
	\AxiomC{$\tderiv_\sval\exder \multiForm \sproves \sval\hastype\formfour$}
		\AxiomC{$\tderiv_\stmfour \exder \multiForm,  \svar\hastype\formfour,\svartwo\hastype\formtwo \sproves \stmfour \hastype\formthree$}
		\RightLabel{$ \implyRightRule $}
		\UnaryInfC{$\multiForm,  \svar\hastype\formfour \sproves \la\svartwo\stmfour\hastype \formtwo \imply \formthree$}
\RightLabel{$ \cut $}
\BinaryInfC{$\multiForm \sproves \cuta\sval\svar\la\svartwo\stmfour\hastype \formtwo \imply \formthree$}
	\DisplayProof
\]
By the weakening lemma (\reflemma{sc-weakening}) applied to $\tderiv_\sval$, there is a derivation $\tderiv'_\sval \exder \multiForm,\svartwo\hastype\formtwo \sproves \sval\hastype\formfour$.
Then, we can build the following derivation $\tderiv'$:
\[
	\AxiomC{$\tderiv'_\sval \exder \multiForm,\svartwo\hastype\formtwo \sproves \sval\hastype\formfour$}
		\AxiomC{$\tderiv_\stmfour \exder \multiForm,  \svar\hastype\formfour,\svartwo\hastype\formtwo \sproves \stmfour \hastype\formthree$}
\RightLabel{$ \cut $}
\BinaryInfC{$\multiForm,\svartwo\hastype\formtwo \sproves \cuta\sval\svar\stmfour\hastype  \formthree$}
	\DisplayProof
\]
By \ih applied to $\tderiv'$, there is a derivation $\tderivtwo' \exder \multiForm,\svartwo\hastype\formtwo \sproves \cutsub\sval\svar\stmfour\hastype \formthree$. Finally, the derivation $\tderivtwo$ for $\la\svartwo\cuta\sval\svar\stmfour$ is built as follows:
\[
		\AxiomC{$\tderivtwo' \exder \multiForm,\svartwo\hastype\formtwo \sproves \cutsub\sval\svar\stmfour\hastype \formthree$}
		\RightLabel{$ \implyRightRule $}
		\UnaryInfC{$\multiForm \sproves \la\svartwo\cutsub\sval\svar\stmfour\hastype \formtwo\imply\formthree$}
	\DisplayProof
\]

\end{itemize}

\item By induction on $\lsctx$. Cases:

\begin{itemize}
\item \emph{Empty}, \ie $\lsctx=\ctxhole$. Then it follows from the previous point.

\item \emph{Cut}, \ie $\lsctx=\cuta\stmfour\svartwo\lsctxtwo$. The derivation $\tderiv$ in the hypothesis then is:
\[
\AxiomC{$\tderiv_\stmfour \exder \multiForm \sproves \stmfour\hastype\formthree$}
		\AxiomC{$\tderiv_\sval \exder \multiFormtwo, \svartwo\hastype\formthree \sproves \sval\hastype\formtwo$}
\doubleLine\dashedLine
 	\RightLabel{$\lsctxtwo$}
	\UnaryInfC{$\multiForm, \svartwo \hastype \formthree \sproves \lsctxtwop\sval \hastype \formtwo$}
\RightLabel{$ \cut $}
\BinaryInfC{$\multiForm \sproves \cuta\stmfour\svartwo\lsctxtwop\sval \hastype\formtwo$}
		\AxiomC{$\tderiv_\stmthree \exder \multiForm,  \svar\hastype\formtwo \sproves \stmthree \hastype\form$}
\RightLabel{$ \cut $}
\BinaryInfC{$\multiForm \sproves \cuta{\cuta\stmfour\svartwo\lsctxtwop\sval}\svar\stmthree \hastype\form$}
	\DisplayProof
\]
In order to apply the \ih, note that our hypotheses allow us to build the following derivation:
\[
		\AxiomC{$\tderiv_\sval \exder \multiFormtwo, \svartwo\hastype\formthree \sproves \sval\hastype\formtwo$}
\doubleLine\dashedLine
 	\RightLabel{$\lsctxtwo$}
	\UnaryInfC{$\multiForm, \svartwo \hastype \formthree \sproves \lsctxtwop\sval \hastype \formtwo$}
		\AxiomC{$\tderiv_\stmthree' \exder \multiForm, \svartwo\hastype\formthree,  \svar\hastype\formtwo \sproves \stmthree \hastype\form$}
\RightLabel{$ \cut $}
\BinaryInfC{$\multiForm, \svartwo\hastype\formthree \sproves \cuta{\lsctxtwop\sval}\svar\stmthree \hastype\form$}
	\DisplayProof
\]
Where the sub-derivation $\tderiv_\stmthree'$ is obtained by applying the weakening lemma (\reflemma{sc-weakening}) to $\tderiv_\stmthree$.

By \ih, we obtain the derivation $\tderivtwo' \exder \multiForm, \svartwo\hastype\formthree \sproves \lsctxtwop{\cutsub\sval\svar\stmthree} \hastype\form$. Finally, the derivation $\tderivtwo$ for $\la\svartwo\cuta\sval\svar\stmfour$ is built as follows:
\[
\AxiomC{$\tderiv_\stmfour \exder \multiForm \sproves \stmfour\hastype\formthree$}
		\AxiomC{$\tderivtwo' \exder \multiForm, \svartwo\hastype\formthree \sproves \lsctxtwop{\cutsub\sval\svar\stmthree} \hastype\form$}
\RightLabel{$ \cut $}
\BinaryInfC{$\multiForm \sproves \cuta\stmfour\svartwo\lsctxtwop{\cutsub\sval\svar\stmthree}  \hastype\form$}
	\DisplayProof
\]

\item \emph{Subtraction}, \ie $\lsctx=\nsuba\svartwo\stmfour\svarthree\lsctxtwo$. Similar to the previous case.
\end{itemize}

\item A straightforward induction on $\gsctx$, using the previous point for the case $\gsctx = \ctxhole$.\qed

\end{enumerate}
\end{proof}

\section{Proofs Removed from \refsect{simulation-vanilla-by-vsc} (Simulating the Vanilla $\l$-calculus in the VSC)}
\label{sect:app-simulation-vanilla-by-vsc}

\gettoappendix{l:sc-to-nd-subst}
\begin{proof}
By induction on $\stm$. 
Cases: \applabel{l:sc-to-nd-subst}
\begin{itemize}
\item \emph{Variable}: if $\stm = \svar$ then: 
\begin{center}
$\begin{array}{llllllllllll}
\subi{\trsctond\sval}\var\trsctond\svar &=&\subi{\trsctond\sval}\var\var &=& \trsctond\sval &=& \trsctond{\cutsub\sval\svar \svar}.
\end{array}$
\end{center}
If instead $\stm = \svartwo$ then: 
\begin{center}
$\begin{array}{llllllllllll}
\subi{\trsctond\sval}\var\trsctond\vartwo &=&\subi{\trsctond\sval}\var\vartwo &=& \vartwo &=& \trsctond\svartwo &=& \trsctond{\cutsub\sval\svar \svartwo}.
\end{array}$
\end{center}

\item \emph{Abstraction}, \ie $\stm = \la\svartwo\stmtwo$. Then:
\begin{center}
$\begin{array}{clllllllllll}
\subi{\trsctond\sval}\var \trsctond{\la\svartwo\stmtwo} &= &\subi{\trsctond\sval}\var\la\vartwo\trsctond\stmtwo 
\\[3pt] &=& 
\la\vartwo \subi{\trsctond\sval}\var\trsctond\stmtwo 
\\[3pt] (\ih) &\todb^*&  
\la\vartwo \trsctond{\cutsub\sval\svar \stmtwo}
\\[3pt] &=& 
\trsctond{\la\svartwo\cutsub\sval\svar \stmtwo}
 &=& 
\trsctond{\cutsub\sval\svar\la\svartwo \stmtwo}.
\end{array}$
\end{center}

\item \emph{Subtraction}, \ie $\stm = \nsuba\svarthree\stmtwo\svartwo \stmthree$. If $\varthree\neq\svar$ then:
\begin{center}
$\begin{array}{clllllllllll}
\subi{\trsctond\sval}\var \trsctond{\nsuba\svarthree\stmtwo\svartwo \stmthree} &= &\subi{\trsctond\sval}\var \sube{\varthree\trsctond\stmtwo}\vartwo \trsctond\stmthree
\\[3pt] &=& 
 \sube{\varthree{\subi{\trsctond\sval}\var\trsctond\stmtwo}}\vartwo \subi{\trsctond\sval}\var\trsctond\stmthree
\\[3pt] (\ih) &\todb^*&   
\sube{\varthree{\trsctond{\cutsub\sval\svar\stmtwo}}}\vartwo \trsctond{\cutsub\sval\svar\stmthree}
\\[3pt] &=&
\trsctond{\nsuba\svarthree{\cutsub\sval\svar\stmtwo}\svartwo \cutsub\sval\svar\stmthree}
&=& 
\trsctond{\cutsub\sval\svar\nsuba\svarthree\stmtwo\svartwo \stmthree}.
\end{array}$
\end{center}
If $\varthree=\var$ then we have to inspect $\sval$. If $\sval$ is a variable $\svarfour$ then:
\begin{center}
$\begin{array}{clllllllllll}
\subi{\trsctond\svarfour}\var \trsctond{\nsuba\svar\stmtwo\svartwo \stmthree} &= &\subi\varfour\var \sube{\var\trsctond\stmtwo}\vartwo \trsctond\stmthree
\\[3pt] &=& 
 \sube{\varfour{\subi\varfour\var\trsctond\stmtwo}}\vartwo \subi\varfour\var\trsctond\stmthree
\\[3pt] (\ih) &\todb^*&   
\sube{\varfour{\trsctond{\cutsub\svarfour\svar\stmtwo}}}\vartwo \trsctond{\cutsub\svarfour\svar\stmthree}
\\[3pt] &=&
\trsctond{\nsuba\svarfour{\cutsub\svarfour\svar\stmtwo}\vartwo \cutsub\svarfour\svar\stmthree}
 &=& 
\trsctond{\cutsub\svarfour\svar\nsuba\svar\stmtwo\svartwo \stmthree}.
\end{array}$
\end{center}
If instead $\sval$ is an abstraction $\la\svarfour\stmfour$ then this is the interesting case requiring $\todb$:
\begin{center}
$\begin{array}{clllllllllll}
\subi{\trsctond\sval}\var \trsctond{\nsuba\svar\stmtwo\svartwo\stmthree}
&= &
\subi{\trsctond\sval}\var \sube{\var\trsctond\stmtwo}\vartwo \trsctond\stmthree
\\[3pt] &=& 
 \sube{\trsctond\sval \subi{\trsctond\sval}\var\trsctond\stmtwo}\vartwo \subi{\trsctond\sval}\var\trsctond\stmthree
\\[3pt] (\ih) &\todb^*&   
\sube{\trsctond\sval{\trsctond{\cutsub\sval\svar\stmtwo}}}\vartwo \trsctond{\cutsub\sval\svar\stmthree}
\\[3pt] &=&
\sube{(\la\varfour\trsctond\stmfour){\trsctond{\cutsub\sval\svar\stmtwo}}}\vartwo \trsctond{\cutsub\sval\svar\stmthree}
\\[3pt] &\todb& 
\sube{\sube{\trsctond{\cutsub\sval\svar\stmtwo}}\varfour\trsctond\stmfour}\vartwo \trsctond{\cutsub\sval\svar\stmthree}
\\[3pt] &=& 
\trsctond{\cuta{\cuta{\cutsub\sval\svar\stmtwo}\svarfour\stmfour}\svartwo \cutsub\sval\svar\stmthree}
\\[3pt] &=& 
\trsctond{\cutsub{\la\svarfour\stmfour}\svar\nsuba\svar\stmtwo\svartwo\stmthree}.
\\[3pt] &=& 
\trsctond{\cutsub\sval\svar\nsuba\svar\stmtwo\svartwo\stmthree}.
\end{array}$
\end{center}

\item \emph{Cut}, \ie $\stm = \cuta\stmtwo\svartwo\stmthree$. Then:
\begin{center}
$\begin{array}{clllllllllll}
\subi{\trsctond\sval}\var \trsctond{\cuta\stmtwo\svartwo \stmthree} &= &\subi{\trsctond\sval}\var \sube{\trsctond\stmtwo}\vartwo \trsctond\stmthree
\\[3pt] &=& 
 \sube{{\subi{\trsctond\sval}\var\trsctond\stmtwo}}\vartwo \subi{\trsctond\sval}\var\trsctond\stmthree
\\[3pt] (\ih) &\todb^*&   
\sube{{\trsctond{\cutsub\sval\svar\stmtwo}}}\vartwo \trsctond{\cutsub\sval\svar\stmthree}
\\[3pt] &=&
\trsctond{\cuta{\cutsub\sval\svar\stmtwo}\vartwo \cutsub\sval\svar\stmthree}
&=& 
\trsctond{\cutsub\sval\svar\cuta\stmtwo\vartwo \stmthree}.
\end{array}$
\end{center}\qed
\end{itemize}
\end{proof}

\gettoappendix{l:transl-sc-to-nd-ctxs}
\begin{proof}
\applabel{l:transl-sc-to-nd-ctxs}
\begin{enumerate}
	\item By induction on $\lsctx$. Case:
	\begin{itemize}
	\item \emph{Empty left context}, \ie $\lsctx=\ctxhole$: by definition empty contexts are translated to empty contexts. 
	
	\item \emph{Subtraction}, \ie $\lsctx= \nsuba\svar\stmtwo\svartwo\lsctxtwo$. Then $\trsctond{\nsuba\svar\stmtwo\svartwo\lsctxtwo} = \sube{\var\trsctond\stmtwo}\vartwo\trsctond\lsctxtwo =_{\ih} \sube{\var\trsctond\stmtwo}\vartwo\lctxtwo$ which is a $\lctx$ context, and similarly $\trsctond{\nsuba\svar\stmtwo\svartwo\lsctxtwop\stm} = \sube{\var\trsctond\stmtwo}\vartwo\trsctond{\lsctxtwop\stm} =_{\ih} \sube{\var\trsctond\stmtwo}\vartwo\trsctond\lsctxtwo\ctxholep{\trsctond\stm} = \trsctond\lsctx \ctxholep{\trsctond\tm}$.
	
	\item \emph{Cut}, \ie $\lsctx= \cuta\stmtwo\svar\lsctxtwo$. Similar to the previous case.
	\end{itemize}
	\item By induction on $\gsctx$ and similar to the previous point.\qed
\end{enumerate}

\end{proof}

\section{Proofs Removed from \refsect{simulating-vsc-vanilla} (Simulating the VSC in the Vanilla $\l$-Calculus)}
\label{sect:app-simulating-vsc-vanilla}

\gettoappendix{l:nd-to-sc-subst}
\begin{proof}
By induction on $\tm$. Cases:\applabel{l:nd-to-sc-subst}
\begin{itemize}
\item \emph{Variable}: if $\tm = \var$ then: 
\begin{center}
$\begin{array}{llllllllllll}
\cutsub{\trndtosc\val}\svar\trndtosc\var &=&\cutsub{\trndtosc\val}\svar\svar &=& \trndtosc\sval &=& \trndtosc{\subi\sval\svar \svar}.
\end{array}$
\end{center}
If instead $\tm = \vartwo$ then: 
\begin{center}
$\begin{array}{llllllllllll}
\cutsub{\trndtosc\sval}\svar\trndtosc\vartwo &=&\cutsub{\trndtosc\sval}\svar\svartwo &=& \svartwo &=& \trndtosc\vartwo &=& \trndtosc{\subi\sval\var \vartwo}.
\end{array}$
\end{center}

\item \emph{Abstraction}, \ie $\tm = \la\vartwo\tmtwo$. Then:
\begin{center}
$\begin{array}{clllllllllll}
\cutsub{\trndtosc\val}\svar \trndtosc{\la\vartwo\tmtwo} &= &\cutsub{\trndtosc\val}\svar\la\svartwo\trndtosc\tmtwo 
\\[3pt] &=& 
\la\svartwo \cutsub{\trndtosc\val}\svar\trndtosc\stmtwo 
\\[3pt] (\ih) &=&  
\la\svartwo \trndtosc{\subi\val\var \tmtwo}
\\[3pt] &=& 
\trndtosc{\la\vartwo\subi\val\var \tmtwo}
\\[3pt]  &=& 
\trndtosc{\subi\val\var\la\vartwo \tmtwo}.
\end{array}$
\end{center}

\item \emph{Application}, \ie $\tm = \tmtwo \tmthree$. We have (with $\svartwo$ and $\svarthree$ fresh):
\begin{center}
$\begin{array}{clllllllllll}
\cutsub{\trndtosc{\val}}\svar\trndtosc{\tmtwo \tmthree}
 &=& 
\cutsub{\trndtosc{\val}}\svar \cuta{\trndtosc\tmtwo}\svartwo \nsuba\svartwo {\trndtosc\tmthree}\svarthree \svarthree
\\[3pt] &=& 
 \cuta{\cutsub{\trndtosc{\val}}\svar\trndtosc\tmtwo}\svartwo \nsuba\svartwo {\cutsub{\trndtosc{\val}}\svar\trndtosc\tmthree}\svarthree \svarthree
\\[3pt] (\ih) &=&   
 \cuta{\trndtosc{\cutsub\val\var\tmtwo}}\svartwo \nsuba\svartwo {\trndtosc{\cutsub\val\var\tmthree}}\svarthree 
\\[3pt] &=&
\trndtosc{\cutsub\sval\var\stmtwo \, \cutsub\sval\var\stmthree}
\\[3pt] &=& 
\trndtosc{\cutsub\sval\var(\stmtwo\stmthree)}
\end{array}$
\end{center}

\item \emph{Explicit substitution}, \ie $\tm = \sube\tmtwo\vartwo\tmthree$. Then:
\begin{center}
$\begin{array}{clllllllllll}
\cutsub{\trndtosc{\val}}\svar\trndtosc{\sube\tmtwo\vartwo\tmthree}
 &=& 
\cutsub{\trndtosc{\val}}\svar \cuta{\trndtosc\tmtwo}\svartwo \trndtosc\tmthree
\\[3pt] &=& 
 \cuta{\cutsub{\trndtosc{\val}}\svar\trndtosc\tmtwo}\svartwo \cutsub{\trndtosc{\val}}\svar\trndtosc\tmthree
\\[3pt] (\ih) &=&   
 \cuta{\trndtosc{\cutsub\val\var\tmtwo}}\svartwo \trndtosc{\cutsub\val\var\tmthree}
\\[3pt] &=&
\trndtosc{ \sube{\cutsub\sval\var\tmtwo}\vartwo \cutsub\sval\var\tmthree}
\\[3pt] &=& 
\trndtosc{\cutsub\sval\var \sube\tmtwo\vartwo\tmthree}
\end{array}$
\end{center}
\qed
\end{itemize}
\end{proof}

\gettoappendix{l:transl-nd-to-sq-ctxs}
\begin{proof}
\applabel{l:transl-nd-to-sq-ctxs}\hfill
\begin{enumerate}
	\item By induction on $\lctx$. By definition empty contexts are translated to empty contexts. If $\lctx= \sube\tmtwo\var\lctxtwo$ then $\trndtosc{\sube\tmtwo\var\lctxtwo} = \cuta{\trndtosc\tmtwo}\svar\trndtosc\lctxtwo =_{\ih} \cuta{\trndtosc\tmtwo}\svar\lsctxtwo$ which is a $\lsctx$ context, and similarly $\trndtosc{\sube\tmtwo\var\lctxtwop\tm} = \cuta{\trndtosc\tmtwo}\svar\trndtosc{\lctxtwop\tm} =_{\ih} \cuta{\trndtosc\tmtwo}\svar\trndtosc\lctxtwo\ctxholep{\trndtosc\tm} = \trndtosc\lctx \ctxholep{\trndtosc\tm}$.
	\item By induction on $\ctx$ and similar to the previous point.\qed
\end{enumerate}
\end{proof}

\gettoappendix{prop:trans-normal-nd-terms}
\begin{proof}
For the proof, we have to strengthen the statement with a \emph{moreover} clause, as follows:\applabel{prop:trans-normal-nd-terms}
\begin{center}
\emph{Let $\tm\in\ndesterms$ be $\tovsc$-normal. Then $\trndtosc\tm \storencut^k \stm$ with $\stm$ cut-free and $k\leq\size\tm$. Moreover, if $\tm$ does not have shape $\lctxp{\la\var\tm'}$ then $\stm$ does not have shape $\lsctxp{\la\svar\stm'}$.}
\end{center}
The proof is by induction on $\tm$. Cases:
\begin{itemize}
\item \emph{Variable}: if $\tm = \var$ then the statement holds by taking $\stm \defeq\svar$ and $k=0\leq\size\var=1$, as indeed $\trndtosc\var=\svar=\stm$ is cut-free. The \emph{moreover} part of the statement holds.

\item \emph{Abstraction}, \ie $\tm = \la\vartwo\tmtwo$. It follows from the \ih

\item \emph{Application}, \ie $\tm = \tmtwo \tmthree$. Then $\size\tm = \size\tmtwo + \size\tmthree + 1$. With $\svartwo$ and $\svarthree$ fresh:
\begin{center}
$\begin{array}{c\colspace c\colspace llllllllll}
\trndtosc{\tmtwo \tmthree}
 &=& 
\cuta{\trndtosc\tmtwo}\svartwo \nsuba\svartwo {\trndtosc\tmthree}\svarthree \svarthree
\\ (\mbox{by }\ih) & \storencut^{\leq\size\tmtwo} &
\cuta{\stmtwo}\svartwo \nsuba\svartwo {\trndtosc\tmthree}\svarthree \svarthree
\\ (\mbox{by }\ih)  & \storencut^{\leq\size\tmthree} &
\cuta{\stmtwo}\svartwo \nsuba\svartwo {\stmthree}\svarthree \svarthree
\end{array}$
\end{center}
Note that $\tmtwo$ cannot have shape $\lctxp{\la\var\tmtwo'}$, otherwise $\tm$ would have a $\todb$-redex. By \ih, $\stmtwo$ cannot have shape $\lsctxp{\la\svar\stmtwo'}$, that is, it has shape $\lsctxp{\svar}$ for some $\lsctx$ and $\svar$. Then, we can perform a renaming step:
\begin{center}
$\begin{array}{c\colspace c\colspace llllllllll}
\cuta{\stmtwo}\svartwo \nsuba\svartwo {\stmthree}\svarthree \svarthree
 &=& 
\cuta{\lsctxp{\svar}}\svartwo \nsuba\svartwo {\stmthree}\svarthree \svarthree
\\  & \storencut &
\lsctxp{\nsuba\svar {\stmthree}\svarthree \svarthree}
\end{array}$
\end{center}
Note that $\lsctxp{\nsuba\svar {\stmthree}\svarthree \svarthree}$ is cut-free because by the \ih both $\lsctx$ and $\stmthree$ are cut-free, and that the \emph{moreover} part of the statement holds.
\item \emph{Explicit substitution}, \ie $\tm = \sube\tmtwo\vartwo\tmthree$. Then $\size\tm = \size\tmtwo + \size\tmthree + 1$. We have:
\begin{center}
$\begin{array}{c\colspace c\colspace llllllllll}
\trndtosc{\sube\tmtwo\vartwo \tmthree}
 &=& 
\cuta{\trndtosc\tmtwo}\svartwo \trndtosc\tmthree
\\ (\mbox{by }\ih) & \storencut^{\leq\size\tmtwo} &
\cuta{\stmtwo}\svartwo \trndtosc\tmthree
\\ (\mbox{by }\ih)  & \storencut^{\leq\size\tmtwo} &
\cuta{\stmtwo}\svartwo \stmthree
\end{array}$
\end{center}
Note that $\tmtwo$ cannot have shape $\lctxp{\la\var\tmtwo'}$, otherwise $\tm$ would have a $\tovs$-redex. By \ih, $\stmtwo$ cannot have shape $\lsctxp{\la\svar\stmtwo'}$, that is, it has shape $\lsctxp{\svar}$ for some $\lsctx$ and $\svar$. Then, we can perform a renaming step:
\begin{center}
$\begin{array}{c\colspace c\colspace llllllllll}
\cuta{\stmtwo}\svartwo \stmthree
 &=& 
\cuta{\lsctxp{\svar}}\svartwo \stmthree
\\  & \stocut &
\lsctxp{\cutsub\svar\svartwo\stmthree}
\end{array}$
\end{center}
Note that $\lsctxp{\cutsub\svar\svartwo\stmthree}$ is cut-free because by the \ih both $\lsctx$ and $\stmthree$ are cut-free, and that the \emph{moreover} part of the statement holds by the \ih on $\tmthree$.
\qed
\end{itemize}
\end{proof}
\section{Rewriting Preliminaries for the Proof of Strong Normalization}
\label{app:SN-preliminaries}

A first straightforward property of SN is the following one.
\begin{lemma}[Stability of SN under renaming]
If $\stm \in \sn\cutsym$  then $\cutsub\svar\svartwo\stm \in \sn\cutsym$ for all variables $\svar,\svartwo\in\vars$.\label{l:deformations}
\end{lemma}

\begin{proof}
By induction on $\stm\in\sn\cutsym$. Since renaming cannot create, erase, or duplicate redexes, 
if $\cutsub\svar\svartwo\stm \stocut \stmtwo$ then there exists $\stmthree$ such that $\stm \stocut \stmthree$ and 
$\cutsub\svar\svartwo\stmthree = \stmtwo$. Then, by \ih on $\stmthree$ we obtain $\cutsub\svar\svartwo\stmthree=\stmtwo\in\sn\cutsym$.
\qed
\end{proof}

\paragraph{Structure of This Section.} We are now going to prove three rewriting properties which shall be the main rewriting tools in the proof of SN by the reducibility method developed in the next section. Namely,
\begin{enumerate}
\item \emph{Extension}: there shall be two extension properties, that extend SN from root sub-terms to the whole term, in special cases;
\item \emph{Root cut expansion}: this property states that if the reduct of a root cut is SN then SN lifts to the term before reducing the cut.
\item \emph{Structural stability}: we shall introduce a notion of structural equivalence for vanilla terms and prove that it preserves SN.
\end{enumerate}
The first two properties can be found (under various names) in \emph{all} proofs of SN, often without specific emphasis. The third one is characteristic of calculi with some form of explicit substitutions and rewriting rules at a distance.

\paragraph{Extension.} The \emph{extension} property is the easy fact that, for abstractions and subtractions, SN follows from SN of the root sub-terms.

\begin{lemma}[Neutral extension] 
Let $\stm, \stmtwo \in \sn\cutsym$. 
Then $\la\svar\stm$ and $\suba\svar\stmtwo\svartwo  \stm$ are in $\sn\cutsym$. \label{l:extension}
\end{lemma}
\begin{proof}
By induction on $\stm \in \sn\cutsym$ for $\la\svar\stm$ and by induction on $(\stm\in \sn\cutsym, \stmtwo \in \sn\cutsym)$ for $\suba\svar\stmtwo\svartwo\stm$. In both cases, one shows that all reducts are in $\sn\cutsym$, which follows immediately from the \ih, because there cannot be interaction between the immediate sub-terms.\qed
\end{proof}

Similarly, when we extend a term $\stm$ with a left contexts $\lsctx$ that captures no variables of $\stm$, SN follows from SN of $\stm$ and $\lsctx$.
\begin{lemma}[Disjoint left context extension]
Let $\stm \in \sn\cutsym$ and $\lsctx$ be a left context such that:\label{l:disj-left-ctx-extension}
\begin{enumerate}
\item $\lsctx$ does not captures variables in $\fv\stm$, and;
\item $\lsctxp\svar \in \sn\cutsym$ for a  variable $\svar$ not captured by $\lsctx$. 
\end{enumerate}
Then, $\lsctxp\stm \in\sn\cutsym$.
\end{lemma}
\begin{proof}
By induction on $(\stm\in \sn\cutsym, \lsctxp\svar \in \sn\cutsym)$. We look at the reducts of $\lsctxp\stm$:
\begin{itemize}
\item $\lsctxp\stm \stocut \lsctxp{\stm'}$ because $\stm\stocut\stm'$. By \ih (1st component), $\lsctxp{\stm'}\in\sn\cutsym$.
\item $\lsctxp\stm \stocut \lsctxtwop{\stm}$ because $\lsctxp\svar\stocut\lsctxtwop\svar$. By \ih (2nd component), $\lsctxtwop\stm \in \sn\cutsym$.
\item There are no other possible reducts, because by hypothesis $\lsctx$ does not captures variables in $\fv\stm$, and so $\lsctx$ and $\stm$ cannot interact.\qed
\end{itemize}
\end{proof}

\paragraph{Substitutivity.} For proving the second main rewriting property, namely \emph{root cut expansion}, we need the following substitutivity properties.

\begin{lemma}[Substitutivity]
\label{l:substitutivity-of-red}\hfill
\begin{enumerate}
\item
If $\stm \stocut \stm'$ then $\cutsub\sval\svar \stm \stocut \cutsub\sval\svar \stm'$;
\item \label{p:substitutivity-of-red-right}
If $\stm \stocut \stm'$ then $\cutsub\stmtwo\svar \stm \stocut \cutsub\stmtwo\svar \stm'$;
\item
If $\sval \stocut \svaltwo$ then $\cutsub\sval\svar \stm \stocut^* \cutsub{\svaltwo}\svar \stm$.
\item \label{p:substitutivity-of-red-left}
If $\stmtwo \stocut \stmtwo'$ then $\cutsub\stmtwo\svar \stm \stocut^* \cutsub{\stmtwo'}\svar \stm$.
\end{enumerate}
\end{lemma}

\begin{proof}
Point 1 is an easy induction on $\stm\stocut\stm'$, and Point 2 is an immediate consequence. Point 3 is another easy induction on $\stm$. For Point 4, let $\stmtwo=\lsctxp\sval$. There are various cases, depending on the step in $\stmtwo$:
\begin{itemize}
\item  \emph{Step in $\lsctx$}, that is, $\stmtwo = \lsctxp\sval \stocut \lsctxtwop\sval = \stmtwo'$. Then:
\[\begin{array}{r\hcolspace c\hcolspace l\hcolspace c\hcolspace l\hcolspace c\hcolspace l}
\cutsub{\lsctxp\sval}\svar \stm 
& = & \lsctxp{\cutsub\sval\svar \stm} 
\\ & \stocut & 
\lsctxtwop{\cutsub\sval\svar \stm}
 & = & \cutsub{\lsctxtwop\sval}\svar \stm 
 \end{array}
\]

\item  \emph{Step in $\sval$}, that is, $\stmtwo = \lsctxp\sval \stocut \lsctxp\svaltwo = \stmtwo'$. Then:
\[\begin{array}{r\hcolspace c\hcolspace l\hcolspace c\hcolspace l\hcolspace c\hcolspace l}
\cutsub{\lsctxp\sval}\svar \stm 
& = & \lsctxp{\cutsub\sval\svar \stm} 
\\\mbox{(by P. 3)} & \stocut & 
\lsctxp{\cutsub\svaltwo\svar \stm}
 & = & \cutsub{\lsctxp\svaltwo}\svar \stm 
 \end{array}
\]

\item  \emph{Step involving both $\lsctx$ and $\sval$}, that is, $\lsctx = \lsctxtwop{\cuta{\lsctxthreep\svaltwo}\svar\lsctxfour}$ and 
\[\begin{array}{r\hcolspace c\hcolspace l\hcolspace c\hcolspace l\hcolspace c\hcolspace l}
\stmtwo 
&= &
\lsctxtwop{\cuta{\lsctxthreep\svaltwo}\svartwo\lsctxfourp\sval} 
&\stocut& 
\lsctxtwop{\lsctxthreep{\cutsub\svaltwo\svartwo\lsctxfourp\sval}} 
&=& \stmtwo'.
\end{array}\]
Let $\lsctx^\bullet \defeq \cutsub\svaltwo\svartwo\lsctxfour$, so that we can write $\stmtwo' = \lsctxtwop{\lsctxthreep{\lsctx^\bullet\ctxholep{\cutsub\svaltwo\svartwo\sval}}}$ and note that $\cutsub\svaltwo\svartwo\sval$ is a value.
Then:
\[\begin{array}{r\hcolspace c\hcolspace l\hcolspace c\hcolspace l}
\cutsub{\lsctxp\sval}\svar \stm&=&\cutsub{\lsctxtwop{\cuta{\lsctxthreep\svaltwo}\svartwo\lsctxfourp\sval}}\svar \stm 
\\ &= & \lsctxtwop{\cuta{\lsctxthreep\svaltwo}\svartwo\lsctxfourp{\cutsub\sval\svar \stm}}
\\ &\stocut & 
\lsctxtwop{\lsctxthreep{\lsctx^\bullet\ctxholep{\cutsub\svaltwo\svartwo\cutsub\sval\svar \stm}}}
\\ & = & 
\lsctxtwop{\lsctxthreep{\lsctx^\bullet\ctxholep{\cutsub{\cutsub\svaltwo\svartwo\sval}\svar \stm}}}
\\ & = & 
\cutsub{\lsctxtwop{\lsctxthreep{\lsctx^\bullet\ctxholep{\cutsub\svaltwo\svartwo\sval}}}}\svar \stm
& = & 
\cutsub{\stmtwo'}\svar \stm
 \end{array}
\]
\qed
\end{itemize}
\end{proof}

\paragraph{Root Cut Expansion.} The \emph{root cut expansion} property is the less obvious fact that if the reduct of a root \emph{small-step} cut is $\sn\cutsym$ then the reducing term also is. It is a key property playing a crucial role in all proofs of SN.

\begin{proposition}[Root cut expansion]
If $\stmtwo\in\sn\cutsym$ and $\cutsub\stmtwo\svar\stm\in\sn\cutsym$ then 
$\cuta\stmtwo\svar \stm \in\sn\cutsym$.\label{prop:root-sn-expansion}
\end{proposition}

\begin{proof}
By induction on $(\stmtwo \in \sn\cutsym, \cutsub\stmtwo\svar \stm \in \sn\cutsym)$, proving that any reduct of $\cuta\stmtwo\svar \stm$ is in $\sn\cutsym$. Cases:
\begin{itemize}
\item \emph{Reduction of the root cut}: $\cuta\stmtwo\svar \stm \stocut \cutsub\stmtwo\svar \stm$, which is in $\sn\cutsym$ by hypothesis.
\item \emph{Reduction of a cut of $\stm$}: that is, $\cuta\stmtwo\svar \stm \stocut \cuta\stmtwo\svar \stm'$ with $\stm \stocut \stm'$. By stability under substitution (\reflemmap{substitutivity-of-red}{right}), $\cutsub\stmtwo\svar \stm \stocut \cutsub\stmtwo\svar \stm'$ and so $\cutsub\stmtwo\svar \stm' \in\sn\to$. By \ih (2nd component, the 1st is unchanged), $\cuta\stmtwo\svar \stm'\in\sn\cutsym$.

\item \emph{Reduction of a cut of $\stmtwo$}: that is, $\cuta\stmtwo\svar \stm \stocut \cuta{\stmtwo'}\svar \stm$ with $\stmtwo \stocut \stmtwo'$. By \reflemmap{substitutivity-of-red}{left}, $\cutsub\stmtwo\svar \stm \stocut^{*} \cutsub{\stmtwo'}\svar \stm$, and so $\cutsub{\stmtwo'}\svar \stm \in\sn\cutsym$. By \ih (1st component), $\cuta{\stmtwo'}\svar \stm \in\sn\cutsym$.\qed
\end{itemize}
\end{proof}

\subsection{Structural Equivalence}
The proof of strong normalization shall exploit \emph{structural equivalence}, which is standard and pervasive concept in $\l$-calculi with explicit substitutions and rewriting rules at a distance. In the vanilla $\l$-calculus, structural equivalence allows one to displace left rules everywhere but inside values. To define it, we need the notion of \emph{weak contexts}, which are contexts whose hole can appear everywhere but inside abstractions.
\begin{definition}[Weak contexts]
Weak contexts are defined by the following grammar:
\begin{center}
$\begin{array}{r\colspace rllllll}
\textsc{Weak contexts} & \wsctx & \grameq & \ctxhole  \mid \cuta\stm\svar\wsctx \mid \cuta\wsctx\svar\stm \mid \nsuba\vartwo\stm\svar\wsctx \mid \nsuba\svartwo\wsctx\svar\stm
\end{array}$
\end{center}
\end{definition}

\begin{definition}[Structural equivalence]
Let $\dom\wsctx$ be the set of variables captured by $\wsctx$. Root structural equivalence $\rcuteq$ is defined as follows:
\begin{center}
$\begin{array}{r\hcolspace c\hcolspace l\colspace l}
\cuta\stmtwo\svar \wsctxp\stm
& \rcuteq & 
\wsctxp{\cuta\stmtwo\svar \stm} &\svar\notin\fv\wsctx,\, \fv\stmtwo\cap\dom\wsctx = \emptyset
\\
\suba\svartwo\stmtwo\svar \wsctxp\stm
& \rcuteq & 
\wsctxp{\suba\svartwo\stmtwo\svar \stm} &\svar\notin\fv\wsctx,\, (\fv\stmtwo\cup\set\svartwo)\cap\dom\wsctx = \emptyset
\end{array}$
\end{center}
\emph{Structural equivalence} $\cuteq$ is the closure of $\rcuteq$ under reflexivity, symmetry, transitivity, and general contexts $\gsctx$.
\end{definition}

The key property of structural equivalences in $\l$-calculi at a distance is that they strongly commute with the rewriting rules at a distance, which is formalized below via the notion of \emph{strong bisimulation}. The adjective \emph{strong} here means that they preserve the number of steps. Consequently, structural equivalence preserves strong normalization.

For proving the strong bisimulation property, we need the following lemma.

\begin{lemma}[Stability by substitution of $\cuteq$]
\label{l:cuteq-subs} 
\hfill
\begin{enumerate}
\item \label{p:cuteq-subs-one}
If $\stm \cuteq \stm'$ then $\cutsub\stmtwo\svar\stm \cuteq \cutsub\stmtwo\svar\stm'$.
\item \label{p:cuteq-subs-two}
If $\stmtwo \cuteq \stmtwo'$ then $\cutsub\stmtwo\svar\stm \cuteq \cutsub{\stmtwo'}\svar\stm$.
\end{enumerate}
\end{lemma}
\begin{proof}
Easy inductions on $\stm \cuteq \stm'$ and $\stm$. \qed
\end{proof}

\begin{toappendix}
\begin{proposition}[Strong bisimulation]
Structural equivalence $\cuteq$ is a strong bisimulation with respect to cut elimination $\stocut$, that is, if $\stm\cuteq\stmtwo$ and $\stm \stocut \stm'$ then there exists $\stmtwo'$ such that $\stmtwo \stocut \stmtwo'$ and $\stm' \cuteq\stmtwo'$. Diagrammatically:\label{prop:strong-bisim}
\begin{center}
\begin{tabular}{c\colspace c \colspace c}
\begin{tikzpicture}[ocenter]
		\node at (0,0)[align = center](source){\normalsize$\stm$};
		\node at (source.east)[right = 35pt](source-right){\normalsize $\stm'$};
		\node at (source.center)[below = 20pt](source-down){\normalsize $\stmtwo$};

		\node at \med{source.center}{source-down.center}(cuteq-left){\normalsize $\cuteq$};
		\draw[->](source) to node[above] {\scriptsize $\cutsym $} (source-right);
	\end{tikzpicture}
& implies $\exists\stmtwo'$ such that&
\begin{tikzpicture}[ocenter]
		\node at (0,0)[align = center](source){\normalsize$\stm$};
		\node at (source.east)[right = 35pt](source-right){\normalsize $\stm'$};
		\node at (source.center)[below = 20pt](source-down){\normalsize $\stmtwo$};
		
		\node at (source-right|-source-down)(target){\normalsize $\stmtwo'$};

		\node at \med{source.center}{source-down.center}(cuteq-left){\normalsize $\cuteq$};
		\node at \med{source-right.center}{target.center}(cuteq-right){\normalsize $\cuteq$};
		\draw[->](source) to node[above] {\scriptsize $\cutsym $} (source-right);
		\draw[->, dotted](source-down) to node[above] {\scriptsize $\cutsym $} (target);
	\end{tikzpicture}
\end{tabular}
\end{center}
 \end{proposition}
 \end{toappendix}
 
 \begin{proof}
 It is a very long but otherwise straightforward check of all possible diagrams (which is detailed in Appendix \ref{app:bisim}). The following two diagrams for $\stocut$ are where \reflemma{cuteq-subs} is used.
\begin{center}
\begin{tabular}{c\colspace\colspace c}
\begin{tikzpicture}[ocenter]
		\node at (0,0)[align = center](source){\normalsize$\cuta{\lsctxp\sval}\svar\stm$};
		\node at (source.east)[right = 35pt](source-right){\normalsize $\lsctxp{\cutsub\sval\svar\stm}$};	
		\node at (source.center)[below = 20pt](source-down){\normalsize $\cuta{\lsctxp\sval}\svar\stm'$};
		
		\node at (source-right|-source-down)(target){\normalsize $\lsctxp{\cutsub\sval\svar\stm'}$};

		\node at \med{source.center}{source-down.center}(cuteq-left){\normalsize $\cuteq$};
		\node at \med{source-right.center}{target.center}(cuteq-right){\normalsize $\cuteq$};
		\draw[|->](source) to node[above] {\scriptsize $\cutsym $} (source-right);
		\draw[|->, dotted](source-down) to node[above] {\scriptsize $\cutsym $} (target);
	\end{tikzpicture}
	&
	\begin{tikzpicture}[ocenter]
		\node at (0,0)[align = center](source){\normalsize$\cuta{\lsctxp\sval}\svar\stm$};
		\node at (source.east)[right = 35pt](source-right){\normalsize $\lsctxp{\cutsub\sval\svar\stm}$};
		\node at (source.center)[below = 20pt](source-down){\normalsize $\cuta{\lsctxp\svaltwo}\svar\stm$};
		
		\node at (source-right|-source-down)(target){\normalsize $\lsctxp{\cutsub\svaltwo\svar\stm}$};

		\node at \med{source.center}{source-down.center}(cuteq-left){\normalsize $\cuteq$};
		\node at \med{source-right.center}{target.center}(cuteq-right){\normalsize $\cuteq$};
		\draw[|->](source) to node[above] {\scriptsize $\cutsym $} (source-right);
		\draw[|->, dotted](source-down) to node[above] {\scriptsize $\cutsym $} (target);
	\end{tikzpicture}
\end{tabular}

\qed
\end{center}
 \end{proof}
 
From the strong bisimulation property, it immediately follows the preservation of SN, as well as a strong postponement property, showing that structural equivalence is harmless, in the sense that it is never needed to unblock a redex, and can be delayed to end of the computation.

\begin{corollary}
\label{coro:cuteq-bisim} 
\hfill
\begin{enumerate}
\item \label{p:cuteq-bisim-sn}
\emph{Structural stability of SN}: if $\stm \cuteq \stm'$ and $\stm\in\snsubst$ then $\stm'\in\snsubst$.
\item \emph{Postponement of $\cuteq$}: if $\stm \,(\cuteq\stocut\cuteq)^k \, \stmtwo$ then $\stm \stocut^k \cuteq \stmtwo$.
\end{enumerate}
\end{corollary}

\withproofs{
\begin{proof}
\hfill
\begin{enumerate}
	\item By induction on $\tm\in\snsubst$. One simply uses the fact that $\cuteq$ is a strong bisimulation (\refprop{strong-bisim}) and the \ih
	\item By induction on $k$, using the strong bisimulation property (\refprop{strong-bisim}).\qed
\end{enumerate}
\end{proof}}

The reader might wonder why strong equivalence is not allowed to move left rules inside values. The reason is that it would break the strong bisimulation property. Let $\equiv_{\symfont{ex}}$ be the extended equivalence which can move left rules inside values. Now, consider the following diagram, where in general the two terms on the right are different and not structural equivalent (typically if there is more than one occurrence of $\svarthree$ in $\stm$):
\begin{center}
\begin{tikzpicture}[ocenter]
		\node at (0,0)[align = center](source){\normalsize$\cuta{\cuta\stmtwo\svar\la\svartwo\stmthree}\svarthree\stm$};
		\node at (source.east)[right = 35pt](source-right){\normalsize $\cuta\stmtwo\svar\cutsub{\la\svartwo\stmthree}\svarthree\stm$};
		\node at (source.center)[below = 20pt](source-down){\normalsize $\cuta{\la\svartwo\cuta\stmtwo\svar\stmthree}\svarthree\stm$};
		
		\node at (source-right|-source-down)(target){\normalsize $\cutsub{\la\svartwo\cuta\stmtwo\svar\stmthree}\svarthree\stm$};

		\node at \med{source.center}{source-down.center}(cuteq-left){\normalsize $\cuteq$};
		\node at \med{source-right.center}{target.center}(cuteq-right){\normalsize $\not\cuteq$};
		\draw[|->](source) to node[above] {\scriptsize $\cutsym $} (source-right);
		\draw[|->, dotted](source-down) to node[above] {\scriptsize $\cutsym $} (target);
	\end{tikzpicture}
\end{center}

\section{Typed Strong Normalization}
\label{sect:SN}
Here we prove strong normalization (SN) of MIL-typed vanilla $\l$-terms using the reducibility method. We adopt Girard's \emph{bi-orthogonal} technique for the method 
\cite{DBLP:journals/tcs/Girard87}, following Accattoli's presentation tailored for the intuitionistic sequent calculus (of linear logic) and reduction at a distance \cite{DBLP:journals/lmcs/Accattoli23}. 

The whole development is more complex than for the ordinary $\l$-calculus, roughly because, having cuts and subtractions, the vanilla $\l$-calculus is more akin to $\l$-calculi with explicit substitutions, which are known to require extra technicalities.
 
 \paragraph{Key Properties.} The reducibility method requires a number of definitions, detailed in the next paragraphs. Because of the many technicalities, it is easy to loose sight of what are the crucial concepts at work in the proof. From a high-level perspective, the proof is based only on three key properties of $\stocut$:
\begin{enumerate}
	\item \emph{Extension}, in the neutral case (\reflemma{extension}) and the disjoint left context case (\reflemma{disj-left-ctx-extension});
	\item
	\emph{Root cut expansion} (\refprop{root-sn-expansion});
		\item \emph{Structural stability of SN}: if $\stm\in\sn\cutsym$ and $\stm\cuteq\stmtwo$ then $\stmtwo\in\sn\cutsym$ 
(\refcorop{cuteq-bisim}{sn}).
\end{enumerate}
This section is also written as to explain some aspects of the reducibility method. Therefore, it is somewhat longer than strictly needed, but hopefully considerably more readable.

\paragraph{Elimination Contexts and Duality.} 
The bi-orthogonal technique we follow is based on a notion of duality defined via \emph{elimination contexts}, that are contexts of the form $\lctxp{\cuta{\ctxhole}\svar\stm}$, noted $\elctx$. Types can be extended to contexts by considering $\ctxhole$ as a free variable and typing it via an axiom, as follows:
\begin{center}
\AxiomC{}
	\RightLabel{$\ax_{\text{ctx}}$}
	\UnaryInfC{$  \typctx, \ctxhole\hastype\form \vdash \ctxhole\hastype \form$}	
	\DisplayProof
\end{center}
Let us set some notations:
\begin{itemize}
\item $\lltermsform$ for the set of terms $\stm$ of type $\form$, that is, such that $\multiForm\vdash \stm\hastype\form$. Note that $\vars\subseteq  \lltermsform$.  
 
\item $\elctxsform \defeq \set{\elctx \ 
|\ \multiForm, \ctxhole\hastype \form \vdash \elctx\hastype \formtwo}$ the set of typed elimination contexts with hole 
of type $\form$, and say that $\elctx$ has \emph{co-type} $\form$. 
\item \emph{Variable contexts} are defined as the elements of $\elvars \defeq \set{\elctx \ |\  \elctx = \cuta\ctxhole\svar\svar}$. Note that $\elvars\subseteq\elctxsform$.
\item $\elctx \in \sn\cutsym$ if $\elctx = \lctxp{\cuta\ctxhole\svar\stm}$ and $\lctxp{\cutsub\svartwo\svar\stm}\in\sn\cutsym$ 
for every variable $\svartwo$ which is not captured by $\lctx$.
\end{itemize}
\begin{remark}
\label{rem:sn-renaming}
Checking that $\elctx=\lctxp{\cuta\ctxhole\svar\stm}$ is in $\sn\cutsym$ amounts to prove that \\
$\lctxp{\cutsub\svartwo\svar\stm} = \cutsub\svartwo\svar\lctxp{\stm}\in\sn\cutsym$ for every appropriate $\svartwo$. By the stability of SN by renamings (\reflemma{deformations}), it is enough to prove that $\lctxp{\stm}\in\sn\cutsym$.
\end{remark}

\begin{definition}[Duality]
Given a set $\settone \subseteq \lltermsform$ of terms of type $\form$, the \emph{dual} set 
$\settone^\bot\subseteq\elctxsform$ contains the elimination contexts $\elctx$ of co-type 
$\form$ such that  $\elctxfp{\stm}$ is proper and in $\sn\cutsym$ for every 
$\stm\in\settone$. The dual of a set of elimination contexts $\settone \subseteq\elctxsform$ of co-type $\form$ is a set of terms $\settone^\bot$ defined symmetrically.\label{def:duality}
\end{definition}
Note the 
use of the capture-avoiding plugging symbol $\ctxholef$ in the definition of duality: the plugging in contexts at work in duality does \emph{not} capture the variables of the plugged 
term. Te use of capture-avoiding plugging is both crucial and standard for the reducibility method. 

The following properties of duality are standard; for a proof see, for instance, Riba's dissection of reducibility \cite{riba:hal-00779623}.

\begin{lemma}[Basic properties of duality]
Let $\settone \subseteq \lltermsform$ or $\settone\subseteq\elctxsform$. \label{l:basic-prop-duality}
 \begin{enumerate}
 \item \emph{Closure}: $\settone\subseteq\settone^{\bot\bot}$;
 \item \emph{Bi-orthogonal}: $\settone^{\bot\bot\bot}=\settone^{\bot}$.
 \end{enumerate}
 \end{lemma}

Next, we show some basic properties of duality, in particular with respect to variables and context variables. We recall that $\vars$ denotes the set of variables of vanilla $\l$-terms.
 \begin{lemma}[Duality and SN]
Let $\stm \in \lltermsform$ and $\elctx \in \elctxsform$. \label{l:sn-to-subterms}
\begin{enumerate}
\item \label{p:sn-to-subterms-one} 
If $\elctxfp\stm\in\sn\cutsym$ then $\stm \in \sn\cutsym$ and $\elctx\in\sn\cutsym$.

\item \label{p:sn-to-subterms-two} 
If $\stm \in \sn\cutsym$ and $\elctx \in \elvars$ then $\elctxfp\stm\in \sn\cutsym$.

\item \label{p:sn-to-subterms-three} 
If $\elctx \in \sn\cutsym$ and $\svartwo \in \vars$ then $\elctxfp\svartwo\in \sn\cutsym$.
\end{enumerate}
\end{lemma}

\begin{proof}
\hfill
\begin{enumerate}
\item By induction on $\elctxfp\stm\in\sn\cutsym$. Let $\stm = \lctxp\sval$ and $\elctx = 
\lctxtwop{\cuta\ctxhole\svar\stmtwo}$, so that $\elctxfp\stm = \lctxtwop{\cuta\stm\svar\stmtwo}$. The proof of the 
statement is based on the obvious fact that every step from $\stm$ or $\lctxtwop\stmtwo$ (we rely on \refrem{sn-renaming}) 
can be mimicked on $\lctxtwop{\cuta\stm\svar\stmtwo}$, so that one can then apply the \ih
\item If $\elctx \in \elvars$ then $\elctx = \cuta\ctxhole\svar\svar$. Let $\stm = \lctxp\sval$. We  have to 
prove that $\elctxfp\stm= \cuta\stm\svar\svar \in \sn\cutsym$. By induction on $\stm \in \sn\cutsym$ we show that all the reducts of $\elctxfp\stm$ are in $\sn\cutsym$. If $\elctxfp\stm = \cuta{\lctxp\sval}\svar\svar \stocut \lctxp\sval = \stm$ by reducing the cut on $\svar$, 
then the reduct is $\stm$, which is in $\sn\cutsym$ by hypothesis. Otherwise, $\elctxfp\stm$ makes a step in $\stm$. The same step can be done on $\stm$, and thus by \ih the reduct is in 
$\sn\cutsym$.
\item Let $\elctx=\lctxp{\cuta\ctxhole\svar\stm}$, so that $\elctxfp\svartwo = \lctxp{\cuta\svartwo\svar\stm}$. If 
$\elctxfp\svartwo \stocut \lctxp{\cutsub\svartwo\svar\stm}$ then the reduct is in $\sn\cutsym$ by the hypothesis on $\elctx$. 
Otherwise, $\elctxfp\svartwo$ makes a step in $\stm$, or in $\lctx$, or involving both $\lctx$ and $\stm$. The same step 
can be done on $\lctxp\stm$ (we rely on \refrem{sn-renaming}), and thus by \ih the reduct is in $\sn\cutsym$.\qed
\end{enumerate}
\end{proof}

\paragraph{Generators, Candidates, and Formulas.} The definition of reducibility candidates comes together with a notion 
of generator, justified by \refprop{generators-give-red-cands} below (whose proof requires the previous lemma).

 \begin{definition}[Generators and candidates]
A \emph{generator} of type $\form$ (resp. co-type $\form$) is a sub-set $\settone\subseteq\lltermsform$ (resp. 
$\settone\subseteq\elctxsform$) such that: 
\begin{enumerate}
\item \emph{Non-emptiness}: $\settone \neq \emptyset$, and 
\item \emph{Strong normalization}: $\settone\subseteq \sn\cutsym$.
\end{enumerate}
A generator $\settone$ of type $\form$ (resp. co-type $\form$) is a \emph{candidate} if: 
\begin{enumerate}
\item \emph{Bi-orthogonal}: $\settone=\settone^{\bot\bot}$.
\end{enumerate}
\end{definition}

The rationale behind the requirements for generators and candidates are captured by the following lemma.
\begin{lemma}
\label{l:candidates-rationale} 
Let $\settone \subseteq \lltermsform$ or $\settone\subseteq\elctxsform$.
\begin{enumerate}
\item \label{p:candidates-rationale-one}
\emph{Non-emptiness dualizes as SN}: if $\settone \neq\emptyset$ then $\settone^\bot\subseteq\sn\cutsym$.
\item \label{p:candidates-rationale-two}
\emph{SN dualizes as non-emptiness}: if $\settone \subseteq\sn\cutsym$ then $\settone^\bot \neq \emptyset$.
\item \label{p:candidates-rationale-three}
\emph{Bi-orthogonal implies closure under cut elimination}: let $\settone=\settone^{\bot\bot} \subseteq \lltermsform$ and $\stm\in \settone$. If $\stm \stocut\stmtwo$ then $\stmtwo\in \settone$.
\end{enumerate}
\end{lemma}
\begin{proof}
\hfill
\begin{enumerate}
\item Let $\settone \subseteq \lltermsform$, $\elctx = \lctxtwop{\cuta\ctxhole\svar\stmtwo} \in 
\settone^\bot$ and $\stm\in\settone \neq\emptyset$. By duality, $\elctxfp\stm \in\sn\cutsym$. By 
\reflemmap{sn-to-subterms}{one}, $\elctx\in\sn\cutsym$. 

Let $\settone\subseteq\elctxsform$, $\stm \in 
\settone^\bot$, and $\elctx = \lctxp{\cuta\ctxhole\svar\stmtwo} \in\settone\neq\emptyset$. By duality, $\elctxfp\stm 
\in\sn\cutsym$. By \reflemmap{sn-to-subterms}{one}, $\stm\in\sn\cutsym$. 

\item Let $\settone \subseteq \lltermsform$, $\stm\in \settone$, and $\cuta\ctxhole\svar\svar\in 
\elvars$. By $\settone \subseteq \sn\cutsym$ and \reflemmap{sn-to-subterms}{two} we obtain 
$\cuta\stm\svar\svar\in\sn\cutsym$, that is, $\cuta\ctxhole\svar\svar \in \settone^{\bot}$.

Let $\settone\subseteq\elctxsform$, $\elctx \in\settone$, and $\svar\in \vars$. By 
$\settone \subseteq \sn\cutsym$ and \reflemmap{sn-to-subterms}{three} we obtain $\elctxfp\svar\in\sn\cutsym$, that is, $\svar 
\in \settone^{\bot}$.

\item
Let $\elctx\in\settone^\bot$. By duality, $\elctxfp\stm\in\sn\cutsym$. Since $\elctxfp\stm \stocut \elctxfp\stmtwo$, we obtain $\elctxfp\stmtwo\in\sn\cutsym$, that is, $\stmtwo\in \settone^{\bot\bot}$. Since $\settone$ is a candidate, $\stmtwo\in \settone$.
\qed
\end{enumerate}
\end{proof}

Note that the proof of \reflemmap{candidates-rationale}{two} proves slightly more than non-emptiness, as it tells us in particular that $\settone^\bot$ contains the variables. It immediately follows that any candidate contains the variables.

\begin{lemma}[Candidates contains variables]
Let $\settone \subseteq \lltermsform$ (resp. $\settone\subseteq\elctxsform$) be a candidate. Then $\vars\subseteq \settone$ (resp. $\elvars\subseteq \settone$).\label{l:candidates-contain-variables}
\end{lemma}

The provided rationale for the properties of generators allows us to prove the property that justifies their name.
\begin{toappendix}
 \begin{proposition}
 \label{prop:generators-give-red-cands}
If $\settone$ is a generator then $\settone^\bot$ is a candidate. 
 \end{proposition}
\end{toappendix}

 \begin{proof}
Properties:
\begin{itemize}
\item \emph{Strong normalization}: by \reflemmap{candidates-rationale}{one}.

\item \emph{Non-emptiness}:  \reflemmap{candidates-rationale}{two}.

\item \emph{Bi-orthogonal}: by the by-orthogonal property of duality (\reflemma{basic-prop-duality}).\qed
\end{itemize}
\end{proof}

We now associate to every formula $\form$ a \emph{seed} set $\settone_{\redc\form}$ that we then prove to be a generator, so that 
its bi-orthogonal is a candidate (by \refprop{generators-give-red-cands}). A minor unusual point is that we define the seed 
$\settone_{\redc\aform}$ for the atomic formula as the set of variables. The literature rather defines it 
as $\lltermsp\aform\cap\sn\cutsym$, which is somewhat wacky, as this is actually what the method is meant to prove!

\begin{definition}[Formula sets and seeds]
Let $\form$ be a MIL formula. The formula set $\redc{\form}$ and formula seed $\settone_{\redc\form}$ are mutually defined as follows: $\redc{\form} \defeq \settone_{\redc\form}^{\bot\bot}$, and
$\settone_{\redc\form}$ is defined by induction on $\form$ as follows:
 \begin{itemize}
  \item  $\settone_{\redc\aform} \defeq \vars$;
  \item  $\settone_{\redc{\form\imply\formtwo}} \defeq \left\{\la\svar\stm\ |\ \cuta\stmtwo\svar\stm \in\redc{\formtwo}\, \forall\stmtwo \in\redc\form\right\}$.
 \end{itemize}
\end{definition}

The proof that $\redc{\form}$ is a candidate for every $\form$ requires a lemma about non-emptiness of the seed $\settone_{\redc{\form\imply\formtwo}}$ of implicative formulas. Please note that non-emptiness of $\settone_{\redc{\form\imply\formtwo}}$ is indeed not obvious: $\settone_{\redc{\form\imply\formtwo}}$ does not include the variables, nor it is defined by plainly  abstracting terms in $\redc\formtwo$, given the additional requirement concerning terms in $\redc\form$. 

\begin{lemma}[Implicative seeds are non-empty]
If $\redc\form$ and $\redc\formtwo$ are candidates then $\settone_{\redc{\form\imply\formtwo}} \neq \emptyset$.\label{l:formulas-give-red-aux}
\end{lemma}
\begin{proof}
Let $\svar \neq \svartwo$. We show that $\la\svar\svartwo \in \settone_{\redc{\form\imply\formtwo}}$. By definition, $\la\svar\svartwo \in \settone_{\redc{\form\imply\formtwo}}$ holds if $\cuta\stmtwo\svar \svartwo \in \redc{\formtwo}$ for every $\stmtwo\in\redc\form$, that is,  $\elctxfp{\cuta\stmtwo\svar \svartwo} \in \sn\cutsym$ for every $\elctx\in\redc\formtwo^{\bot}$. 
Since candidates contain variables (\reflemma{candidates-contain-variables}), $\vars \subseteq \redc\formtwo$. By duality, $\elctxfp\svartwo \in \sn\cutsym$. By hypothesis, we also have $\stmtwo \in \sn\cutsym$. Since we can assume that $\svar\notin\fv{\elctxfp\svartwo}$, by disjoint left context extension (\reflemma{disj-left-ctx-extension}) we obtain $\cuta\stmtwo\svar\elctxfp\svartwo \in \sn\cutsym$. By structural stability, $\cuta\stmtwo\svar\elctxfp\svartwo \cuteq \elctxfp{\cuta\stmtwo\svar\svartwo} \in \sn\cutsym$.\qed
\end{proof}

Now, we are ready to prove that formula sets $\redc\form$ are candidates, given by the third point of the next proposition. The proof is simple and yet tricky: by induction on $\form$, it uses 
\refprop{generators-give-red-cands} to prove 2 from 1, and 3 from 2, but it also needs 3 (on sub-formulas) to prove 1. 

\begin{proposition}[Formula sets are candidates]
Let $\form$ be a MIL formula.\label{prop:formulas-give-redc}
\begin{enumerate}
\item \emph{Seeds are generators}: $\settone_{\redc\form}$ is a generator.
\item \emph{Dual candidates}: $\settone_{\redc\form}^{\bot}$ is a  candidate.
\item \emph{Candidates}: $\redc{\form}$ is a candidate.
\end{enumerate}
\end{proposition}

\begin{proof}
We prove the first point, the second follows from the first and \refprop{generators-give-red-cands}, the third one follows from the second and \refprop{generators-give-red-cands}. By induction on $\form$. Cases:
\begin{itemize}
\item \emph{Base}, \ie $\form = \aform$. Note that $\settone_{\redc\aform}$ contains the variables by definition and that all its elements are normal, thus in $\sn\cutsym$.

\item \emph{Implication}, \ie $\form = \formtwo\lolli\formthree$. 
 By \ih (point 3), both $\redc\formtwo$ and $\redc\formthree$ are candidates. Then by \reflemma{formulas-give-red-aux}, $\settone_{\redc{\formtwo\lolli\formthree}}\neq\emptyset$.

 Since variables are in $\redc\formtwo$, if $\la\svar\stmthree\in \settone_{\redc{\formtwo\lolli\formthree}}$ then $\cuta\svar\svar\stmthree \in \redc\formthree$. By \ih (Point 3), $\cuta\svar\svar\stmthree \in \sn\cutsym$ and so does $\stmthree$. By extension, $\la\svar\stmthree \in \sn\cutsym$. Therefore $\settone_{\redc{\formtwo\lolli\formthree}} \subseteq \sn\cutsym$.\qed
\end{itemize}
\end{proof}

\paragraph{Reducibility and Adequacy.} It is now time to introduce the last key notion, namely \emph{reducible derivations}, and---as it is standard---prove \emph{adequacy}, from which SN shall follow.  When applying the bi-orthogonal reducibility method to a calculus with explicit substitutions (here cuts), it is essential to have \emph{two} equivalent formulations of reducibility (one based on generators and one based on candidates), as to use the most convenient one in each case of the proof of adequacy. We define both, and prove their equivalence below. 

Another technical detail (induced by the presence of explicit substitution) is that the rigid structure of terms 
forces an order between the assignments in the typing context $\multiform$, that is, it treats $\multiform$ as a 
\emph{list} rather than as a multi-set.  The proof of adequacy then considers that the sequent calculus comes with an 
exchange rule, treated by one of the cases.

\begin{definition}[Reducible derivations]
Let $\tderiv\pof\multiForm \vdash \stm\hastype \form$ be a typing derivation, and  $\multiForm = 
\svar_{1}\hastype\formtwo_{1},\mydots, 
\svar_{k}\hastype\formtwo_{k} $. Then:
\begin{itemize}
\item $\tderiv$ is \emph{value reducible} if 
$\cuta{\sval_{1}}{\svar_{1}}\ldots \cuta{\sval_{k}}{\svar_{k}} \stm \in \redc\form$
 for every value 
$\sval_{i}\in\settone_{\redc{\formtwo_{i}}}$ such that the introduced cuts are independent, that is, 
$\fv{\sval_{i}} \cap \dom\multiForm = \emptyset$. 
\item $\tderiv$ is \emph{term reducible} if 
$\cuta{\stmtwo_{1}}{\svar_{1}}\ldots \cuta{\stmtwo_{k}}{\svar_{k}} \stm \in \redc\form$
 for every term
$\stmtwo_i\in \redc{\formtwo_{i}}$  such that the introduced cuts are independent, that is, 
$\fv{\stmtwo_{i}} \cap \dom\multiForm = \emptyset$. 
\end{itemize}
To ease notations, we shorten $\cuta{\sval_{1}}{\svar_{1}}\ldots \cuta{\sval_{k}}{\svar_{k}} \stm$ to 
$\cuta{\sval_{i}}{\svar_{i}}_{\multiForm}\stm$, and $\cuta{\stmtwo_{1}}{\svar_{1}}\ldots \cuta{\stmtwo_{k}}{\svar_{k}} \stm$ to $\cuta{\stmtwo_i}{\svar_{i}}_{\multiForm}\stm$.
\end{definition}

\begin{lemma}[Value and term reducibility coincide]
Let $\tderiv\pof \multiForm \vdash \stm\hastype \form$ be a type derivation. Then $\tderiv$ is value reducible if and 
only if $\tderiv$ is term reducible.\label{l:reduc-hyp-simpl}
\end{lemma}

\begin{proof}
That term reducibility implies value reducibility is obvious because $\settone_{\redc{\formtwo_{i}}} \subseteq \redc{\formtwo_{i}}$. We prove that value reductibility implies term reducibility. Let $\multiform = \multiFormthree,\svar\hastype\formtwo,\multiFormtwo$.  The hypothesis is: 
\begin{equation}
\begin{array}{l\hcolspace l\hcolspace llllll}
\elctxfp{\cuta{\sval_{i}}{\svartwo_{i}}_{\multiFormthree}\cuta{\sval}\svar\cuta{\sval_{j}}{\svarthree_{j}}_{\multiFormtwo}\stm} & \in &\sn\cutsym,
\end{array}
\label{eq:red-hyp}
\end{equation} 
for every $\sval\in\settone_{\redc\formtwo}$ and appropriate $\sval_{i}$ and $\sval_{j}$. We show that we can replace $\sval$ with $\stmtwo \in \redc\formtwo$, that is, that the following holds:
\begin{equation}
\begin{array}{l\hcolspace l\hcolspace llllll}
\elctxfp{\cuta{\sval_{i}}{\svartwo_{i}}_{\multiFormthree}\cuta{\stmtwo}\svar\cuta{\sval_{j}}{\svarthree_{j}}_{\multiFormtwo}\stm} & \in &\sn\cutsym.
\end{array}
\label{eq:red-concl}
\end{equation}
By iterating the reasoning on all other values $\sval_{i}$ and $\sval_{j}$ one obtains the statement. Since \refequa{red-hyp} holds for all $\sval\in\settone_{\redc\formtwo}$ we have that 
$
\elctxtwo \defeq \elctxfp{\cuta{\sval_{i}}{\svartwo_{i}}_{\multiFormthree}\cuta{\ctxhole}\svar\cuta{\sval_{j}}{\svarthree_{j}}_{\multiFormtwo}\stm}\in \settone_{\redc\formtwo}^{\bot} = \redc\formtwo^{\bot}
$.
Note that the cuts $\cuta{\sval_{i}}{\svartwo_{i}}_{\multiFormthree}$ and $\cuta{\sval_{j}}{\svarthree_{j}}_{\multiFormtwo}$ can be added to the elimination context exactly because they are independent. Now, by duality we obtain $\elctxtwofp\stmtwo \in \sn\cutsym$ for every $\stmtwo \in \redc\formtwo$, which is exactly \refequa{red-concl}.\qed
\end{proof}

From now on, we do not distinguish between value and term reducibility, and simply refer to \emph{reducible derivations}. For the proof of adequacy, we need two auxiliary lemmas. The first one guarantees that the notion of reducible derivations does what it is supposed to do, despite its somewhat convoluted formulation (which is necessary for the proof to go through).

\begin{lemma}[Intended meaning of reducibility]
Let $\tderiv\pof \multiForm \vdash \stm\hastype \form$ be a reducible derivation. Then $\stm\in\redc\form\subseteq\sn\cutsym$.\label{l:reduc-implies-sn}
\end{lemma}
\begin{proof}
Since $\redc\formtwo$ is a candidate (\refprop{formulas-give-redc}), it contains the variables (\reflemma{candidates-contain-variables}). By reducibility, $\cuta{\svar_{i}}{\svar_{i}}_{\multiform}\stmthree \in\redc\formtwo$. Note that $\cuta{\svar_{i}}{\svar_{i}}_{\multiform}\stmthree \stocut^*\stmthree$, and that $\stmthree\in\redc\formtwo$ because candidates are closed under cut elimination (\reflemmap{candidates-rationale}{three}). Finally, since $\redc\formtwo$ is a candidate, $\redc\formtwo\subseteq\sn\cutsym$.\qed
\end{proof}

The second auxiliary lemma slightly generalizes the root cut expansion property of the previous section to when the cut to expand appears in the hole of an elimination context $\elctx$, which is how we shall apply it when proving adequacy for reducible derivations.

\begin{lemma}[Generalized root cut expansion]
If $\elctxfp{\cutsub\stmtwo\svar\stm}\in \sn\cutsym$ and $\stmtwo\in\sn\cutsym$ then $\elctxfp{\cuta\stmtwo\svar\stm}\in \sn\cutsym$.
\label{l:struct-stability-adequacy}
\end{lemma}

\begin{proof}
Let $\stmtwo = \lsctxp\sval$. Note that: 
\[\begin{array}{l\hcolspace l\hcolspace l\hcolspace l\hcolspace lll}
\elctxfp{\cutsub\stmtwo\svar\stm} & = &\elctxfp{\lsctxp{\cutsub\sval\svar\stm}}
\\& \cuteq& \lsctxp{\elctxfp{\cutsub\sval\svar\stm}} 
\\ &= &
\lsctxp{\cutsub\sval\svar\elctxfp{\stm}} 
 &=  & \cutsub\stmtwo\svar\elctxfp{\stm}
\end{array}\]
 and that $\cutsub\stmtwo\svar\elctxfp{\stm}\in\sn\cutsym$ by structural stability of SN. By root cut expansion, $\cuta\stmtwo\svar\elctxfp{\stm}\in\sn\cutsym$. By structural stability, $\cuta\stmtwo\svar\elctxfp{\stm} \cuteq \elctxfp{\cuta\stmtwo\svar\stm} \in\sn\cutsym$.\qed
\end{proof}

\begin{toappendix}
\begin{theorem}[Adequacy]
Let $\tderiv \pof \multiForm \vdash \stm\hastype \form$ a type derivation. Then $\tderiv$ is reducible.$\label{tm:adequacy}$
\end{theorem}
\end{toappendix}

\begin{proof}
By induction on $\tderiv$. The proof rests on the three rewriting properties mentioned at the beginning of the section, namely extension, (generalized) root cut expansion, and  structural stability of SN. Let $\multiForm = \svar_{1}\hastype\formtwo_{1},\mydots, 
\svar_{k}\hastype\formtwo_{k} $. Cases of the last rule of $\tderiv$:
\begin{itemize}
\item \emph{Axiom}: 
\begin{center}
\AxiomC{}
	\RightLabel{$\ax$}
	\UnaryInfC{$\svar\hastype\form \vdash 
\svar\hastype\form$}
	\DisplayProof
\end{center}
We need to show that $\elctxfp{\cuta{ \stmtwo}\svar\svar} 
\in\sn\cutsym$ for every term $\stmtwo\in\redc{\form}\subseteq\sn\cutsym$ and every elimination context $\elctx\in\redc{ \form}^{\bot}$. Note that: 
\[\begin{array}{l\hcolspace l\hcolspace l\hcolspace l\hcolspace l}
\elctxfp{\cuta{ \stmtwo}\svar\svar} 
&\stocut&
 \elctxfp{\cutsub{ \stmtwo}\svar\svar} 
&=&
\elctxfp{ \stmtwo}
\end{array}\]
Which 
is in $\sn\cutsym$ by duality. Note also  that $\elctxfp{ \stmtwo} = \elctxfp{\cutsub{ \stmtwo}\svar\svar}$. By generalized root cut expansion (\reflemma{struct-stability-adequacy}), $\elctxfp{\cuta{ \stmtwo}\svar\svar}\in\sn\cutsym$. 

\item \emph{Exchange}:
\begin{center}
\AxiomC{$ \tderivtwo\pof \multiformtwo,\svar\hastype\formtwo,\svartwo\hastype\formthree,\multiformthree \vdash \stm\hastype \form$}
	\RightLabel{$\exch$}
	\UnaryInfC{$  \multiformtwo,\svartwo\hastype\formthree,\svar\hastype\formtwo,\multiformthree \vdash \stm\hastype \form$}	
	\DisplayProof
	\end{center}
	with $\multiform = \multiformtwo,\svartwo\hastype\formthree,\svar\hastype\formtwo,\multiformthree$. Let $\multiFormtwo = \svar_{1}\hastype\formtwo_{1},\mydots, 
\svar_{k}\hastype\formtwo_{k} $ and $\multiFormthree = \svartwo_{1}\hastype\formthree_{1},\mydots, 
\svartwo_{h}\hastype\formthree_{h} $. For value reducibility, we have to 
show that:
\[\begin{array}{l\hcolspace l\hcolspace l\hcolspace l\hcolspace l}
\stm' &\defeq& \elctxfp{\cuta{\sval_{i}}{\svar_{i}}_{\multiFormtwo}\cuta\sval\svartwo\cuta\svaltwo\svar \cuta{\svaltwo_{j}}{\svartwo_{j}}_{\multiFormthree} \stm} &\in& \sn\cutsym.
\end{array}\]
for 
every $\sval\in \settone_{\redc{\formthree}}$, $\svaltwo\in \settone_{\redc{\formtwo}}$, $\sval_{i} \in\settone_{\redc{\formtwo_{i}}}$ for $i\in\set{1,\ldots,k}$, every $\svaltwo_{j} \in\settone_{\redc{\formthree_{j}}}$ for $i\in\set{1,\ldots,h}$, and every 
$\elctx\in\redc\form^{\bot}$.   By \ih 
on 
$\tderivtwo$, we have: 
\[\begin{array}{l\hcolspace l\hcolspace l\hcolspace l\hcolspace l}
\stm'' &\defeq& \elctxfp{\cuta{\sval_{i}}{\svar_{i}}_{\multiFormtwo}\cuta\svaltwo\svar \cuta\sval\svartwo\cuta{\svaltwo_{j}}{\svartwo_{j}}_{\multiFormthree} \stm} &\in&\sn\cutsym.
\end{array}\]
By structural stability and the independence of cuts in the definition of reducibility, $\stm'' \cuteq \stm' \in \sn\cutsym$.

\item \emph{Cut}:
\begin{center}
\AxiomC{$\tderiv_{\stmthree}\pof  \multiForm \vdash \stmthree\hastype\formtwo$}
	\AxiomC{$\tderiv_{\stmtwo}\pof  \multiForm, \svar\hastype\formtwo \vdash \stmtwo\hastype\form$}
	\RightLabel{$ \cut $}
	\BinaryInfC{$  \multiForm \vdash \cuta\stmthree\svar \stmtwo \hastype\form$}
	\DisplayProof
	\end{center} 
		with $\stm =\cuta\stmthree\svar \stmtwo$.  We have to 
show that: 
\[\begin{array}{l\hcolspace l\hcolspace l}
\cuta{\sval_{i}}{\svar_{i}}_{\multiform} 
\cuta\stmthree\svar\stmtwo &\in& \redc\form.
\end{array}\]
for 
every $\sval_{i} \in\settone_{\redc{\formtwo_{i}}}$ for $i\in\set{1,\ldots,k}$.   By \ih, 
$\tderiv_{\stmthree}$ is reducible. By \reflemma{reduc-implies-sn}, $\stmthree\in\redc\formtwo$.
By \ih,  
$\tderiv_{\stmtwo}$ is reducible, thus $\cuta{\sval_{i}}{\svar_{i}}_{\multiform} \cuta\stmthree\svar\stmtwo\in \redc\form$.
 \item \emph{Implication right}:
\begin{center}
\AxiomC{$\tderivtwo\pof  \svar\hastype\formtwo, \multiForm \vdash \stmtwo\hastype\formthree$}
	\RightLabel{$ \implyRightRule $}
	\UnaryInfC{$  \multiForm \vdash \la\svar\stmtwo\hastype\formtwo \imply \formthree$}
	\DisplayProof
		\end{center}
	with $\stm =\la\svar\stmtwo$ and $\form = \formtwo \imply \formthree$.
	We have to show that $\stm' \defeq \elctxfp{\cuta{\sval_{i}}{\svar_{i}}_{\multiform}\la\svar\stmtwo} 
\in\sn\cutsym$ for every $\sval_{i} \in \settone_{\redc{\formtwo_{i}}}$ for $i\in\set{1,\ldots,k}$ and every 
$\elctx\in\redc{\formtwo \imply \formthree}^{\bot}$.  By \ih on $\tderivtwo$ and the equivalence between value and term reducibility (\reflemma{reduc-hyp-simpl}), we have  
$\cuta\stmthree\svar\cuta{\sval_{i}}{\svar_{i}}_{\multiform}\stmtwo \in\redc\formthree$ for every 
$\stmthree\in\redc\formtwo$ and with $\svar\notin \fv{\sval_i}$ for all $i$ because of the independence of cuts in the definition of reducibility. Then $\la\svar\cuta{\sval_{i}}{\svar_{i}}_{\multiform}\stmtwo 
\in\redc{\formtwo\imply\formthree}$. By duality, 
$\elctxfp{\la\svar\cuta{\sval_{i}}{\svar_{i}}_{\multiform}\stmtwo} \in \sn\cutsym$. Note that:
\[\begin{array}{l\hcolspace c\hcolspace l\hcolspace l\hcolspace l}
\elctxfp{\la\svar\cuta{\sval_{i}}{\svar_{i}}_{\multiform}\stmtwo} 
& \stocut^* &
\elctxfp{\la\svar\cutsub{\sval_{i}}{\svar_{i}}_{\multiform}\stmtwo} 
\\ & = &
\elctxfp{\cutsub{\sval_{i}}{\svar_{i}}_{\multiform}\la\svar\stmtwo} 
& \in & \sn\cutsym.
\end{array}\]
By generalized root cut expansion (\reflemma{struct-stability-adequacy}), we obtain:	
\[\begin{array}{l\hcolspace c\hcolspace l\hcolspace l\hcolspace l}
\elctxfp{\cuta{\sval_{i}}{\svar_{i}}_{\multiform}\la\svar\stmtwo} &= &\stm' &\in &\sn\cutsym.
\end{array}\]
 \item \emph{Implication left}: 
\begin{center}
\AxiomC{$\tderiv_{\stmthree}\pof  \multiForm \vdash \stmthree\hastype\formthree$}
	\AxiomC{$\tderiv_{\stmtwo}\pof  \multiForm, \svar\hastype\formtwo \vdash \stmtwo\hastype\form$}
	\RightLabel{$ \implyLeftRule $}
	\BinaryInfC{$   \multiForm \cup \svartwo\hastype \formthree \imply \formtwo \vdash 
\suba\svartwo\stmthree\svar  \stmtwo \hastype\form$}
	\DisplayProof
	\end{center} 
	with $\stm =\suba\svartwo\stmthree\svar  \stmtwo$. There are two cases, depending on whether $\svartwo\hastype \formthree \imply \formtwo$ appears in $\multiform$.
	\begin{itemize}
	\item \emph{$\svartwo\hastype \formthree \imply \formtwo$ does not appear in $\multiform$}, and the derivation is:
	\[
\AxiomC{$\tderiv_{\stmthree}\pof  \multiForm \vdash \stmthree\hastype\formthree$}
	\AxiomC{$\tderiv_{\stmtwo}\pof  \multiForm, \svar\hastype\formtwo \vdash \stmtwo\hastype\form$}
	\RightLabel{$ \implyLeftRule $}
	\BinaryInfC{$   \multiForm , \svartwo\hastype \formthree \imply \formtwo \vdash 
\suba\svartwo\stmthree\svar  \stmtwo \hastype\form$}
	\DisplayProof
	\]
	Note that in particular this means that $\svartwo$ does not occur in $\stmthree$ nor $\stmtwo$.
	
 We have to show that: 
\[\begin{array}{l\hcolspace l\hcolspace l\hcolspace l\hcolspace l}
\stm' &\defeq &
\elctxfp{\cuta{\sval_{i}}{\svar_{i}}_{\multiform}\cuta{\la\svarthree\stmfour}\svartwo\suba\svartwo\stmthree\svar  \stmtwo} 
&\in&\sn\cutsym.
\end{array}\]
 for every $\sval_{i}\in\settone_{\redc{\formtwo_{i}}}$ for $i\in\set{1,\ldots,k}$, every 
$\la\svarthree\stmfour\in\settone_{\redc{\formthree\lolli\formtwo}}$, and every $\elctx\in\redc\form^{\bot}$. 
By \ih, $\tderiv_{\stmthree}$ is reducible. By \reflemma{reduc-implies-sn}, $\stmthree\in\redc\formthree$.

 By definition of $\la\svarthree\stmfour\in\settone_{\redc{\formthree\imply\formtwo}}$, we obtain 
$\cuta\stmthree\svarthree \stmfour  \in\redc\formtwo$. By \ih, $\tderiv_{\stmtwo}$ is reducible, thus we have:
\[\begin{array}{l\hcolspace l\hcolspace l}
\cuta{\sval_{i}}{\svar_{i}}_{\multiForm} \cuta{\cuta\stmthree\svarthree\stmfour}\svar\stmtwo &\in& \redc\form,
\end{array}\]	
That is:
\[\begin{array}{l\hcolspace l\hcolspace l\hcolspace l\hcolspace l}
	\stm'' &\defeq& \elctxfp{\cuta{\sval_{i}}{\svar_{i}}_{\multiForm} \cuta{\cuta\stmthree\svarthree\stmfour}\svar\stmtwo} & \in&\sn\cutsym.
	\end{array}\]
	Note that, since by hypotheses $\svartwo\notin\fv\stmthree \cup\fv\stmtwo$, the last term can be written as follows: 
\[\begin{array}{l\hcolspace l\hcolspace l}
	\stm'' &=& \elctxfp{\cuta{\sval_{i}}{\svar_{i}}_{\multiform} \cutsub{\la\svarthree\stmfour}{\svartwo} \suba\svartwo\stmthree\svar \stmtwo} 
		\end{array}\]	
By applying generalized root cut expansion (\reflemma{struct-stability-adequacy}) to $\stm''$, we obtain exactly $\stm' \in \sn\cutsym$.
	
	\item \emph{$\svartwo\hastype \formthree \imply \formtwo$ does appear in $\multiform$}, that is, $\multiform= \multiformtwo, \svartwo\hastype \formthree \imply \formtwo$ for some $\multiformtwo$, and the derivation is:
	\[
\AxiomC{$\tderiv_{\stmthree}\pof  \multiformtwo, \svartwo\hastype \formthree \imply \formtwo \vdash \stmthree\hastype\formthree$}
	\AxiomC{$\tderiv_{\stmtwo}\pof  \multiformtwo, \svartwo\hastype \formthree \imply \formtwo, \svar\hastype\formtwo \vdash \stmtwo\hastype\form$}
	\RightLabel{$ \implyLeftRule $}
	\BinaryInfC{$   \multiformtwo, \svartwo\hastype \formthree \imply \formtwo \vdash 
\suba\svartwo\stmthree\svar  \stmtwo \hastype\form$}
	\DisplayProof
	\]
 Let $\multiFormtwo = \svar_{1}\hastype\formtwo_{1},\mydots, 
\svar_{k}\hastype\formtwo_{k} $. We have to show that: 
\[\begin{array}{l\hcolspace c\hcolspace l\hcolspace l\hcolspace l}
\stm' &\defeq &
\elctxfp{\cuta{\sval_{i}}{\svar_{i}}_{ \multiFormtwo}\cuta{\la\svarthree\stmfour}\svartwo\suba\svartwo\stmthree\svar  \stmtwo} 
&\in&\sn\cutsym.
\end{array}\]
 for every $\sval_{i}\in\settone_{\redc{\formtwo_{i}}}$ for $i\in\set{1,\ldots,k}$, every 
$\la\svarthree\stmfour\in\settone_{\redc{\formthree\lolli\formtwo}}$, and every $\elctx\in\redc\form^{\bot}$. 
By \ih, $\tderiv_{\stmthree}$ is reducible, thus we have:
\[\begin{array}{l\hcolspace c\hcolspace l}
\cuta{\sval_{i}}{\svar_{i}}_{ \multiFormtwo} \cuta{\la\svarthree\stmfour}\svartwo \stmthree &\in& \redc\formthree.
\end{array}\]
Note that that term reduces to $\cutsub{\la\svarthree\stmfour}\svartwo \stmthree$, which is in $\redc\formthree$ by the closure of candidates under cut elimination (\reflemmap{candidates-rationale}{three}).

 By definition of $\la\svarthree\stmfour\in\settone_{\redc{\formthree\imply\formtwo}}$, we obtain 
$\cuta{\cutsub{\la\svarthree\stmfour}\svartwo\stmthree}\svarthree \stmfour  \in\redc\formtwo$. By \ih, $\tderiv_{\stmtwo}$ is reducible, thus we have:
\[\begin{array}{l\hcolspace c\hcolspace l}
\cuta{\sval_{i}}{\svar_{i}}_{ \multiFormtwo} \cuta{\la\svarthree\stmfour}\svartwo \cuta{\cuta{\cutsub{\la\svarthree\stmfour}\svartwo\stmthree}\svarthree\stmfour}\svar\stmtwo &\in& \redc\form.
\end{array}\]
That is:
\[\begin{array}{l\hcolspace l\hcolspace l\hcolspace l\hcolspace l}
	\stm'' &\defeq& \elctxfp{\cuta{\sval_{i}}{\svar_{i}}_{ \multiFormtwo} \cuta{\la\svarthree\stmfour}\svartwo \cuta{\cuta{\cutsub{\la\svarthree\stmfour}\svartwo\stmthree}\svarthree\stmfour}\svar\stmtwo} & \in&\sn\cutsym.
	\end{array}\]
	Note that $\stm''$ reduces to:
\[\begin{array}{l\hcolspace l\hcolspace l\hcolspace l\hcolspace l}
	\stm''' &\defeq& \elctxfp{\cuta{\sval_{i}}{\svar_{i}}_{ \multiFormtwo}  \cuta{\cuta{\cutsub{\la\svarthree\stmfour}\svartwo\stmthree}\svarthree\stmfour}\svar\cutsub{\la\svarthree\stmfour}\svartwo\stmtwo} & \in&\sn\cutsym.
	\end{array}\]
	And that $\stm'''$ can be written as follows: 
	\[\begin{array}{l\hcolspace l\hcolspace l}
	\stm''' &=& \elctxfp{\cuta{\sval_{i}}{\svar_{i}}_{\multiform} \cutsub{\la\svarthree\stmfour}{\svartwo} \suba\svartwo\stmthree\svar \stmtwo} 
		\end{array}\]	
By applying generalized root cut expansion (\reflemma{struct-stability-adequacy}) to $\stm'''$, we obtain exactly $\stm' \in \sn\cutsym$.
\qed	
	\end{itemize}
		\end{itemize}
\end{proof}

\begin{toappendix}
\begin{corollary}[Typable terms are SN]
Let $\stm$ be a typable term. Then $\stm\in\sn\cutsym$.
\end{corollary}
\end{toappendix}

\begin{proof}
Since $\stm$ is typable, we have $\tderiv\pof\svar_{1}\hastype\formtwo_{1},\ldots, \svar_{k}\hastype\formtwo_{k} \vdash 
\stm\hastype \form$, for some derivation $\tderiv$.  By adequacy (\reftm{adequacy}), $\tderiv$ is reducible. By 
\reflemma{reduc-implies-sn}, $\stm\in\sn\cutsym$.\qed
\end{proof}

\section{Proof of Strong Bisimulation for Structural Equivalence}
For the sake of completeness, in this section we give most details of the easy but tedious proof that structural equivalence is a strong bisimulation. The property is first stated in Appendix \ref{app:SN-preliminaries} (and recalled below) but omitting the proof.
\label{app:bisim}
\gettoappendix{prop:strong-bisim}
\begin{proof}
We first treat $\rcuteq$ (and its symmetric case) and then its context closure $\crcuteq$, defined as $\gsctxp\stm \crcuteq\gsctxp{\stm'}$ if $\stm\rcuteq\stm'$. In order to obtain the statement for $\cuteq$, one needs to further add closure under reflexivity, which is obvious,  and under transitivity, which is straightforward. 

Now, we deal with root structural equivalence $\rcuteq$ for cut, that is, we deal with the case:
	\[\begin{array}{c\hcolspace c\hcolspace cccc}
	\cuta\stmtwo\svar \wsctxp\stm &\rcuteq&  \wsctxp{\cuta\stmtwo\svar \stm}
	\end{array}\]
	With $\svar\notin\fv\wsctx$  and $\wsctx$ not capturing variables in $\fv\stmtwo$. Root structural equivalence for subtractions is analogous but simpler. We first consider the cases of reduction in $\cuta\stmtwo\svar \wsctxp\stm$:
\begin{itemize}
	\item \emph{A redex entirely in $\wsctx$}. Then:
	\[\begin{array}{c\hcolspace c\hcolspace cccc}
\cuta\stmtwo\svar \wsctxp\stm
& \stocut &
\cuta\stmtwo\svar \wsctxtwop\stm
\\[3pt]
\rcuteq && \rcuteq
\\[4pt]
\wsctxp{\cuta\stmtwo\svar \stm}
&\stocut &
\wsctxtwop{\cuta\stmtwo\svar \stm}
\end{array}\]

	\item \emph{A redex entirely in $\tm$}. Then:
	\[\begin{array}{c\hcolspace c\hcolspace cccc}
\cuta\stmtwo\svar \wsctxp\stm
& \stocut &
\cuta\stmtwo\svar \wsctxp{\stm'}
\\[3pt]
\rcuteq && \rcuteq
\\[4pt]
\wsctxp{\cuta\stmtwo\svar \stm}
&\stocut &
\wsctxp{\cuta\stmtwo\svar \stm'}
\end{array}
\]

	\item \emph{A redex entirely in $\stmtwo$}. Then:
	\[\begin{array}{c\hcolspace c\hcolspace cccc}
\cuta\stmtwo\svar \wsctxp\stm
& \stocut &
\cuta{\stmtwo'}\svar \wsctxp{\stm'}
\\[3pt]
\rcuteq && \rcuteq
\\[4pt]
\wsctxp{\cuta\stmtwo\svar \stm}
&\stocut &
\wsctxp{\cuta{\stmtwo'}\svar \stm'}
\end{array}
\]

	\item \emph{A redex of which the acting cut is the moving one}:
	\[\begin{array}{c\hcolspace c\hcolspace cccc}
\cuta{\lsctxp\sval}\svar \wsctxp\stm
& \stocut &
\lsctxp{\wsctxp{\cutsub\sval\svar\stm}}
\\[3pt]
\rcuteq && \cuteq
\\[4pt]
\wsctxp{\cuta{\lsctxp\sval}\svar\stm}
&\stocut &
\wsctxp{\lsctxp{\cutsub\sval\svar\stm}}
\end{array}
\]
	
	\item \emph{A redex such that the hole of $\wsctx$ falls in the cut sub-term of the acting cut}, that is, such that (given $\cuta\stmtwo\svar \wsctxp\tm \rcuteq  \wsctxp{\cuta\stmtwo\svar \tm}$) we have $\wsctx = \wsctxtwop{\cuta\lsctxtwo\svartwo\tmthree}$ and $\stm=\lsctxp\sval$:
	\[\begin{array}{c\hcolspace c\hcolspace cccc}
\cuta\stmtwo\svar \wsctxtwop{\cuta{\lsctxtwop{\lsctxp\sval}}\svartwo\tmthree}
& \stocut &
\cuta\stmtwo\svar \wsctxtwop{\lsctxtwop{\lsctxp{\cutsub\sval\svartwo\tmthree}}}
\\[3pt]
\rcuteq && \rcuteq
\\[4pt]
 \wsctxtwop{\cuta{\lsctxtwop{\cuta\stmtwo\svar\lsctxp\sval}}\svartwo\tmthree}
&\stocut &
 \wsctxtwop{\lsctxtwop{\cuta\stmtwo\svar\lsctxp{\cutsub\sval\svartwo\tmthree}}}
\end{array}
\]
	
	\item \emph{A redex involving a cut in $\wsctx$ acting on $\stm$}. Then for some $\stm'$ and $\wsctxtwo$ we have the following diagram:
	\[\begin{array}{c\hcolspace c\hcolspace cccc}
\cuta\stmtwo\svar \wsctxp{\stm}
& \stocut &
\cuta\stmtwo\svar \wsctxtwop{\stm'}
\\[3pt]
\rcuteq && \rcuteq
\\[4pt]
\wsctxp{\cuta\stmtwo\svar \stm}
& \stocut &
\wsctxtwop{\cuta\stmtwo\svar \stm'}
\end{array}
\]
\end{itemize}
	Now, for the cases where the $\stocut$ step is in $\wsctxp{\cuta\stmtwo\svar \stm}$ (rather than in $\cuta\stmtwo\svar \wsctxp{\stm}$) note that there cannot be a redex involving a cut in $\wsctx$ on a variable $\svartwo$ occurring in $\stmtwo$, because by hypothesis $\wsctx$ does not capture free variables of $\stmtwo$. Then the possible cases are simply those treated above for $\cuta\stmtwo\svar \wsctxp\stm$, just read backwards.
	
	For the inductive cases, that is, for $\crcuteq$, the only non trivial cases are those for cut, but we sketch the other ones as well:
	\begin{itemize}
	\item \emph{Abstraction}: it follows from the \ih
	\item \emph{Subtraction}: if the rewriting step and the $\cuteq$ are in the same sub-term then it follows from the \ih, otherwise they simply swap.
	
	\item \emph{Right sub-term of cut}, that is, $\cuta{\stmthree}\svartwo\stmfour \cuteq \cuta{\stmthree}\svartwo\stmfour'$ with $\stmfour \cuteq \stmfour'$. If the rewriting step takes place in $\stmfour$ then the rewriting and the equivalence step simply swap. If it takes place in $\stmthree$ then the strong bisimulation is given by the \ih Last, the rewriting step might involve the root cut. Let $\stmthree= \lsctxp\sval$. Then the diagram goes as follows:
	\[\begin{array}{c\hcolspace c\hcolspace cccc}
\cuta{\lsctxp\sval}\svartwo\stmfour
& \stocut &
\lsctxp{\cutsub\sval\svartwo\stmfour}
\\[3pt]
\crcuteq && \crcuteq
\\[4pt]
\cuta{\lsctxp\sval}\svartwo\stmfour'
& \stocut &
\lsctxp{\cutsub\sval\svartwo\stmfour'}
\end{array}
\]
Where the structural equivalence on the right side is given by \reflemmap{cuteq-subs}{one}.

	\item \emph{Left sub-term of cut}, that is, $\cuta{\stmthree}\svartwo\stmfour \cuteq \cuta{\stmthree'}\svartwo\stmfour$ with $\stmthree \cuteq \stmthree'$. If the rewriting step takes place in $\stmfour$ then the rewriting and the equivalence step simply swap. If it takes place in $\stmthree$ then the strong bisimulation is given by the \ih Last, the rewriting step might involve the root cut. Then let $\stmthree=\lsctxp\sval$. Sub-cases:
	\begin{itemize}
	\item The equivalence takes place in $\lsctx$. Then the diagram is:
		\[\begin{array}{c\hcolspace c\hcolspace cccc}
\cuta{\lsctxp\sval}\svartwo\stmfour
& \stocut &
\lsctxp{\cutsub\sval\svartwo\stmfour}
\\[3pt]
\crcuteq && \crcuteq
\\[4pt]
\cuta{\lsctxtwop\sval}\svartwo\stmfour
& \stocut &
\lsctxtwop{\cutsub\sval\svartwo\stmfour}
\end{array}
\]

	\item The equivalence takes place in $\sval$. Then the diagram is:
		\[\begin{array}{c\hcolspace c\hcolspace cccc}
\cuta{\lsctxp\sval}\svartwo\stmfour
& \stocut &
\lsctxp{\cutsub\sval\svartwo\stmfour}
\\[3pt]
\crcuteq && \cuteq
\\[4pt]
\cuta{\lsctxp\svaltwo}\svartwo\stmfour
& \stocut &
\lsctxp{\cutsub\svaltwo\svartwo\stmfour}
\end{array}\]
Where the structural equivalence on the right side is given by \reflemmap{cuteq-subs}{two}.\qed
	\end{itemize}
	\end{itemize}
\end{proof}
}

\end{document}